\documentclass[10pt,a4paper]{article}

\usepackage[T1]{fontenc}
\usepackage[utf8]{inputenc}
\usepackage{amsmath,amssymb,amsthm,amsfonts}
\usepackage{a4wide}

\usepackage{graphicx}
\usepackage{caption}
\usepackage{subcaption}

\theoremstyle{plain}
\newtheorem{theorem}{Theorem}[section]
\newtheorem{proposition}[theorem]{Proposition}
\newtheorem{lemma}[theorem]{Lemma}

\theoremstyle{definition}

\newtheorem{remark}[theorem]{Remark}
\newtheorem{example}[theorem]{Example}
\newtheorem{assumption}[theorem]{Assumption}
\theoremstyle{remark}
\numberwithin{equation}{section}

\usepackage[pdfborder={0 0 0},bookmarksopen]{hyperref}
\hypersetup{%
	pdftitle = {Hedging with Small Uncertainty Aversion},
	pdfauthor = {Sebastian Herrmann, Johannes Muhle-Karbe, Frank Thomas Seifried},
	pdfkeywords = {},
  pdfpagemode = UseNone,
}

\DeclareMathOperator*{\sgn}{sign}
\DeclareMathOperator{\D}{D}
\DeclareMathOperator{\Trace}{Trace}

\newcommand{\FF}{\mathbb{F}}
\newcommand{\NN}{\mathbb{N}}
\newcommand{\RR}{\mathbb{R}}

\newcommand{\bfD}{\mathbf{D}}

\newcommand{\cA}{\mathcal{A}}
\newcommand{\cB}{\mathcal{B}}

\newcommand{\cF}{\mathcal{F}}

\newcommand{\cV}{\mathcal{V}}
\newcommand{\cX}{\mathcal{X}}
\newcommand{\cZ}{\mathcal{Z}}

\newcommand{\fa}{\mathfrak{a}}
\newcommand{\fP}{\mathfrak{P}}
\newcommand{\fS}{\mathfrak{S}}

\newcommand{\diff}{\mathrm{d}}
\newcommand{\dd}{\,\mathrm{d}}

\newcommand{\1}{\mathbf{1}}

\newcommand*{\EX}[2][]{E^{#1}\left [ #2 \right ]}
\newcommand*{\cEX}[3][]{E^{#1}\left[ #2 \,\middle\vert\, #3 \right]}
\newcommand*{\PR}[2][P]{{#1}\left [ #2 \right ]}
\newcommand*{\wt}[1]{\widetilde{#1}}
\newcommand*{\wh}[1]{\widehat{#1}}
\newcommand*{\ol}[1]{\overline{#1}}

\newcommand{\tr}{\top}

\begin{document}
\title{%
Hedging with Small Uncertainty Aversion%
\footnote{The authors thank Ibrahim Ekren, Paolo Guasoni, Martin Herdegen, David Hobson, Jan Kallsen, Ariel Neufeld, Marcel Nutz, Oleg Reichmann, Martin Schweizer, and H.~Mete Soner for fruitful discussions and two referees and an associate editor for helpful comments.}
}
\date{}
\author{%
  Sebastian Herrmann%
  \thanks{
  ETH Z\"urich, Department of Mathematics, R\"amistrasse 101, CH-8092 Z\"urich, Switzerland, email
  \href{mailto:sebastian.herrmann@math.ethz.ch}{\nolinkurl{sebastian.herrmann@math.ethz.ch}}.
  Financial support by the Swiss Finance Institute is gratefully acknowledged.
  }
  \and
  Johannes Muhle-Karbe%
  \thanks{
  University of Michigan, Department of Mathematics, 530 Church Street, Ann Arbor, MI 48109, USA, email
  \href{mailto:johanmk@umich.edu}{\nolinkurl{johanmk@umich.edu}}.
  }
  \and
  Frank Thomas Seifried%
  \thanks{
  Department IV -- Mathematics, Universit\"at Trier, Universit\"atsring 19, D-54296 Trier, Germany, email
  \href{mailto:seifried@uni-trier.de}{\nolinkurl{seifried@uni-trier.de}}.
  }
}
\maketitle

\begin{abstract}
We study the pricing and hedging of derivative securities with uncertainty about the volatility of the underlying asset. Rather than taking all models from a prespecified class equally seriously, we penalise less plausible ones based on their ``distance'' to a reference local volatility model. In the limit for small uncertainty aversion, this leads to explicit formulas for prices and hedging strategies in terms of the security's cash gamma.
\end{abstract}

\vspace{0.5em}

{\small

\vspace{0.25em}

\vspace{0.25em}
}

\section{Introduction}
\label{sec:introduction}

Prices on financial markets are, ultimately, driven by human decisions. Whence, every quantitative model is at best a useful approximation. Moreover, even if a given model is well-suited for a particular application, its parameters typically cannot be estimated with arbitrary precision and may change suddenly due to unpredictable shocks. Therefore, it is crucial to assess the impact of model uncertainty and to derive decision rules that explicitly take it into account. The present study tackles this problem for the valuation and hedging of derivatives with uncertainty about the volatility of the underlying. 

The most widely used approach in the extant literature, dating back to the seminal papers of Avellaneda, Levy, and Par\'as  \cite{AvellanedaLevyParas1995} and Lyons \cite{Lyons1995}, is the so-called \emph{uncertain volatility model} (UVM).\footnote{A related approach is Mykland's ``conservative delta hedging''~\cite{Mykland2000,Mykland2003.prediction}.} Here, the single probabilistic model used in the classical setting is replaced by a whole class of different alternatives. These are evaluated with a worst-case approach, both with respect to the ``risk'' inherent in a given model and with respect to the ``Knightian uncertainty'' about which model should be applied in the first place. More specifically, the UVM assumes only that the true volatility evolves within a band $[\sigma_{\min},\sigma_{\max}]$, without any assumptions on its dynamics. In this setting, one then tries to determine the cheapest hedge that removes all downside risk under \emph{any} conceivable model, corresponding to infinite aversion against both risk and uncertainty. This approach is well-suited to determine universal no-arbitrage bounds and has initiated a tremendous amount of research; see \cite{Hobson1998.lookback, Hobson1998.coupling, Frey2000, DenisMartini2006, Hobson2011, AcciaioBeiglbockPenknerSchachermayer2015, BeiglbockHenryLaborderePenkner2013, NeufeldNutz2013, PossamaiRoyerTouzi2013, BiaginiBouchardKardarasNutz2015, DolinskySoner2014, GalichonHenryLabordereTouzi2014, Nutz2014, BouchardNutz2015, HouObloj2015} and the references therein. However, by its very definition it is also very conservative: unless a given model is ruled out a priori, it is treated in exactly the same way as the most plausible alternatives, say point estimates derived from market prices or statistical procedures. If the volatility band is chosen very wide, then the bid-ask spread induced by the worst-case approach is very large and a market maker quoting these bid-ask prices will not find customers in any competitive environment. If one artificially tightens the band of volatility scenarios to obtain competitive prices, the risk that the true volatility strays out of the band increases. Finally, the worst-case approach often incorporates options' specific sensitivities to volatility only in a bang-bang fashion. For instance, the ask price of a call option always equals the Black--Scholes price with volatility equal to the maximal value $\sigma_{\max}$. This holds even if the price of the underlying is far away from the strike, where the option's sensitivity to changes in volatility is very low. A more nuanced attitude towards model uncertainty should reflect that in the context of hedging, agents are typically more concerned about volatility misspecification when it matters -- that is, if the option value's sensitivity to volatility is high. 

To address these issues, one can consider preferences that interpolate between the classical and the worst-case approach by weighing different models according to their plausibility. Maccheroni, Marinacci, and Rustichini \cite{MaccheroniMarinacciRustichini2006} provide a decision-theoretic, axiomatic foundation for criteria of this kind. They show that uncertainty-averse decision makers rank payoffs $Y$ according to a numerical representation of their preferences of the following form: 
\begin{align}
\label{eqn:intro:numerical representation}
\inf_P \left(\EX[P]{U(Y)} + \alpha(P) \right).
\end{align}
Here, the \emph{utility function} $U$ and the \emph{penalty functional}\footnote{\label{ftn:penalty}This terminology stems from the literature on robust control \cite{HansenSargent2007} and also from the robust representation of convex risk measures (see, e.g., \cite[Section 5.2]{FollmerSchied2013}). Note that the penalty is not imposed on the agent, but on her fictitious adversary who minimises over $P$. Hence, the penalty in \eqref{eqn:intro:numerical representation} is added and not subtracted.} $\alpha$ describe the decision makers' attitudes towards risk and uncertainty, respectively. One interpretation is that the decision maker behaves as if she was facing a malevolent opponent (``nature'') who takes advantage of her uncertainty by choosing a probability scenario $P$ that minimises her penalised expected utility. The standard expected utility framework corresponds to the case where $\alpha$ is zero for a single model $P$ and $+\infty$ otherwise. If $\alpha$ is zero for a whole class of models $\fP$ that are considered (equally) reasonable and $+\infty$ otherwise, we are in the worst-case approach and \eqref{eqn:intro:numerical representation} reduces to the \emph{multiple priors preferences} of Gilboa and Schmeidler~\cite{GilboaSchmeidler1989}. Another well-known special case of \eqref{eqn:intro:numerical representation} is given by the \emph{multiplier preferences} of Hansen and Sargent \cite{HansenSargent2001}, which penalise models based on their relative entropy with respect to a fixed reference model. With this approach, the user is not required to draw a strict line between ``reasonable'' and ``wrong'' models; instead, less plausible models are penalised according to how much they differ from the reference model. This specification excludes any uncertainty about volatility as the relative entropy is only finite for models that are absolutely continuous with respect to the reference model. Accordingly, this line of research has mostly focused on portfolio choice problems with uncertainty about future expected returns. With a different penalty functional, however, the same philosophy can also be applied to price and hedge options with uncertainty about the volatility of the underlying, as we do in the present study.

More specifically, we use the following general approach:

\renewcommand{\labelenumi}{(\roman{enumi})}
\begin{enumerate}
\item Choose a \emph{reference model} that you believe describes the true dynamics of the stock price reasonably well. Make sure that this model matches the market prices of liquidly traded derivatives on the stock.

\item Devise a nonnegative \emph{penalty functional} that attaches a penalty to each alternative model and reflects how seriously you take it; the higher a model's penalty, the less plausible it is. Typically, your penalty functional is based on some ``distance'' to your reference model.

\item For each trading strategy, evaluate its performance using \eqref{eqn:intro:numerical representation}. That is, compute the expected utility from the corresponding terminal Profit\&Loss (henceforth P\&L; defined as the terminal wealth from trading minus the payoffs of the options you have sold) in each model, add the model's penalty, and then take the infimum over all models.

\item Find a trading strategy that maximises this performance criterion.
\end{enumerate}

In this paper, we choose a local volatility reference model \cite{Dupire1994}. This allows to incorporate the spot prices of vanilla options, but still isolates the impact of model uncertainty from additional complexities arising from incompleteness or additional state variables.\footnote{In view of the well-known deficiencies of local volatility models in capturing the dynamics of option prices, cf., e.g., \cite{DumasFlemingWhaley1998,Gatheral2006}, an extension to more general reference models is an important direction for future research; see~Remark~\ref{rem:stochastic volatility} for some further discussion.}

The penalty functionals we consider penalise deviations of the actual (instantaneous) volatility from its counterpart in the reference model. A typical example is the mean-squared  distance  between those two volatilities:\footnote{A similar penalty has been used by \cite{AvellanedaFriedmanHolmesSamperi1997} in the context of local volatility calibration with prior beliefs.}
\begin{align}
\label{eqn:intro:penalty functional}
\alpha(P)
&= \EX[P]{\frac{1}{2\psi}\int_0^T U'(Y_t) (\sigma_t - \bar\sigma(t,S_t))^2 \dd t}.
\end{align}
Here, $Y_t$ is your P\&L at time $t$, defined as wealth from trading minus the reference value of the option at time $t$, $\bar\sigma(t,S_t)$ is the local volatility of the reference model, and $\sigma_t$ is the actual instantaneous volatility of returns of the stock price under the model $P$.\footnote{The term $U'(Y_t)$ in \eqref{eqn:intro:penalty functional} renders the preferences invariant under affine transformations of the utility function; see also Remark~\ref{rem:marginal utility}.} The parameter $\psi > 0$ describes the magnitude of uncertainty aversion. Indeed, small values of $\psi$ lead to high penalties for models that deviate from the reference model, so that those alternative models are taken less and less seriously as $\psi$ approaches zero.

To obtain explicit formulas, we pass to the limit where uncertainty aversion $\psi$ tends to zero.\footnote{Asymptotic analyses of option pricing and hedging problems with the worst-case approach have been carried out by \cite{Lyons1995,AhnMuniSwindle1997,AhnMuniSwindle1999,FouqueRen2014}.} That is, we consider the problem at hand as a perturbation of its classical counterpart for the reference model, and then correct prices and hedges to take into account model uncertainty in an asymptotically optimal manner. 

\begin{figure}
\centering
\includegraphics[scale=.4]{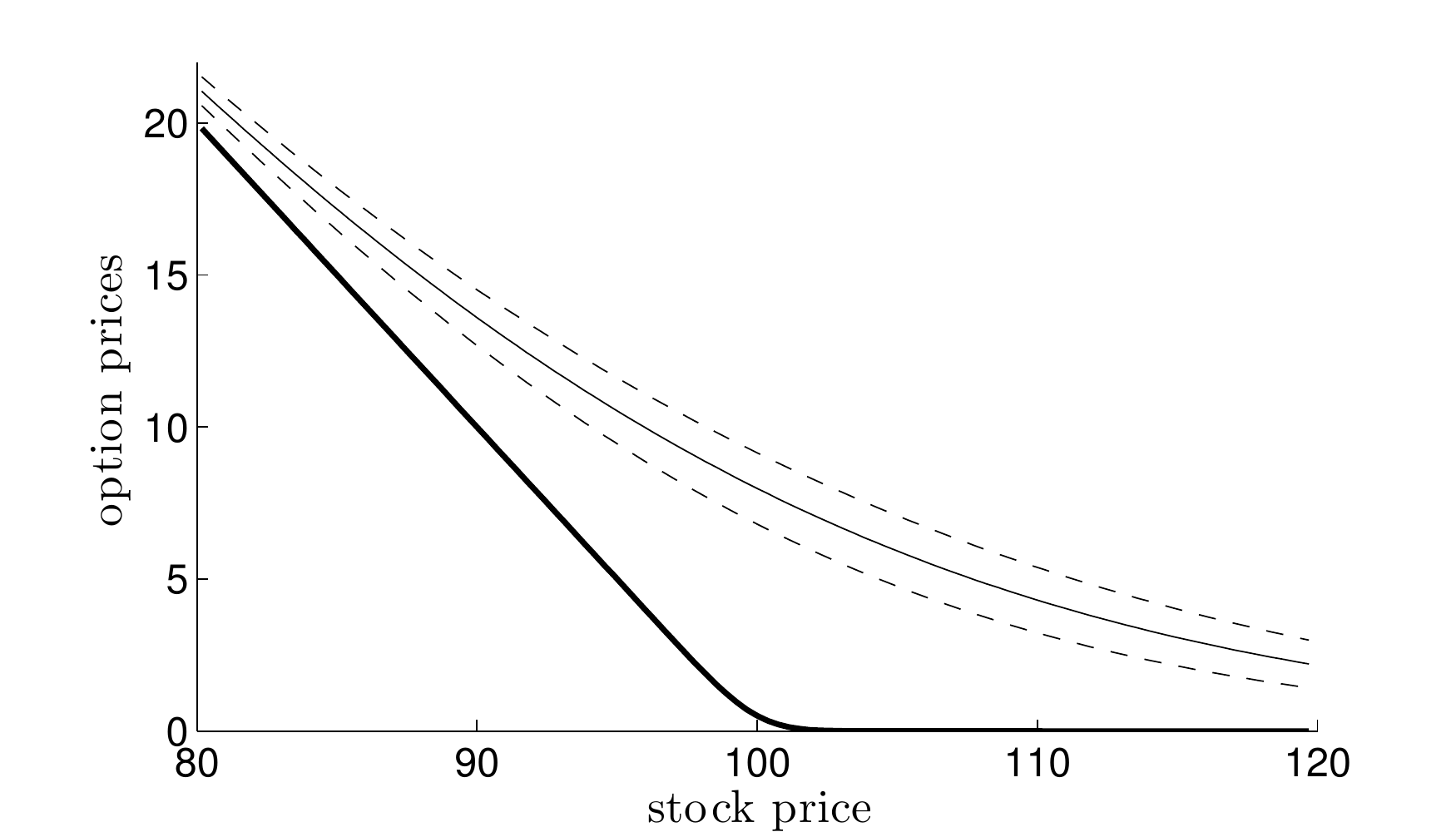}
\caption{Payoff of a ``smooth put'' with strike $100$ (solid, thick), Black--Scholes value of a $1$ year ``smooth put'' (solid, thin), and first-order bid and ask prices $p_b(\psi)$ and $p_a(\psi)$ (dashed), all as functions of the current stock price; the parameters are $\bar\sigma \equiv 20\%$ and $\psi = 10^{-3}$.}
\label{fig:smooth put:value}
\end{figure}

For a single vanilla option, hedged by trading in the underlying stock $S$, we establish a rigorous second-order expansion of the value $v(\psi)$ of this hedging problem for small values of the uncertainty aversion parameter $\psi$. The leading-order term in this expansion is attained by the delta hedge in the reference model; almost optimal strategies $\theta^\psi$ are identified by matching the next-to-leading order term. In analogy to perturbation analyses with transaction costs~\cite{WhalleyWilmott1997}, the leading-order coefficient of the expansion as well as the formula for the almost optimal strategy are determined by the solution $\wt w$ to a linear second-order parabolic partial differential equation (PDE) with a source term. The prices at which you are indifferent between a flat position and a long or short position in the option can in turn be described as follows. Starting from some initial capital $x_0$ and a flat position in options, your indifference bid and ask prices $p_b(\psi)$ and $p_a(\psi)$ for the option have the expansions
\begin{align*}
p_b(\psi)
&= \bar V - \wt w \psi + o(\psi)
\quad\text{and}\quad
p_a(\psi)
= \bar V + \wt w \psi + o(\psi),
\end{align*}
where $\bar V = \bar V(0,S_0)$ is the option price in the reference model. Whence, $\wt w\psi$ is the leading-order discount or premium that you demand as a compensation for exposing yourself to model uncertainty. We can thus interpret $\wt w$ as a measure for the option's susceptibility to model misspecification and call it the option's \emph{cash equivalent (of small uncertainty aversion)}. To better understand this correction term, consider its Feynman--Kac representation (cf.~Proposition \ref{prop:Feynman-Kac}):\footnote{In Section \ref{sec:introduction}, we only display the formulas for the special case of the penalty functional \eqref{eqn:intro:penalty functional}. Our analysis also applies to more general penalty functionals; cf.~Section~\ref{sec:hedging under uncertainty}.}

\begin{figure}
\centering
\begin{subfigure}{.5\textwidth}
  \centering
  \includegraphics[scale=.4]{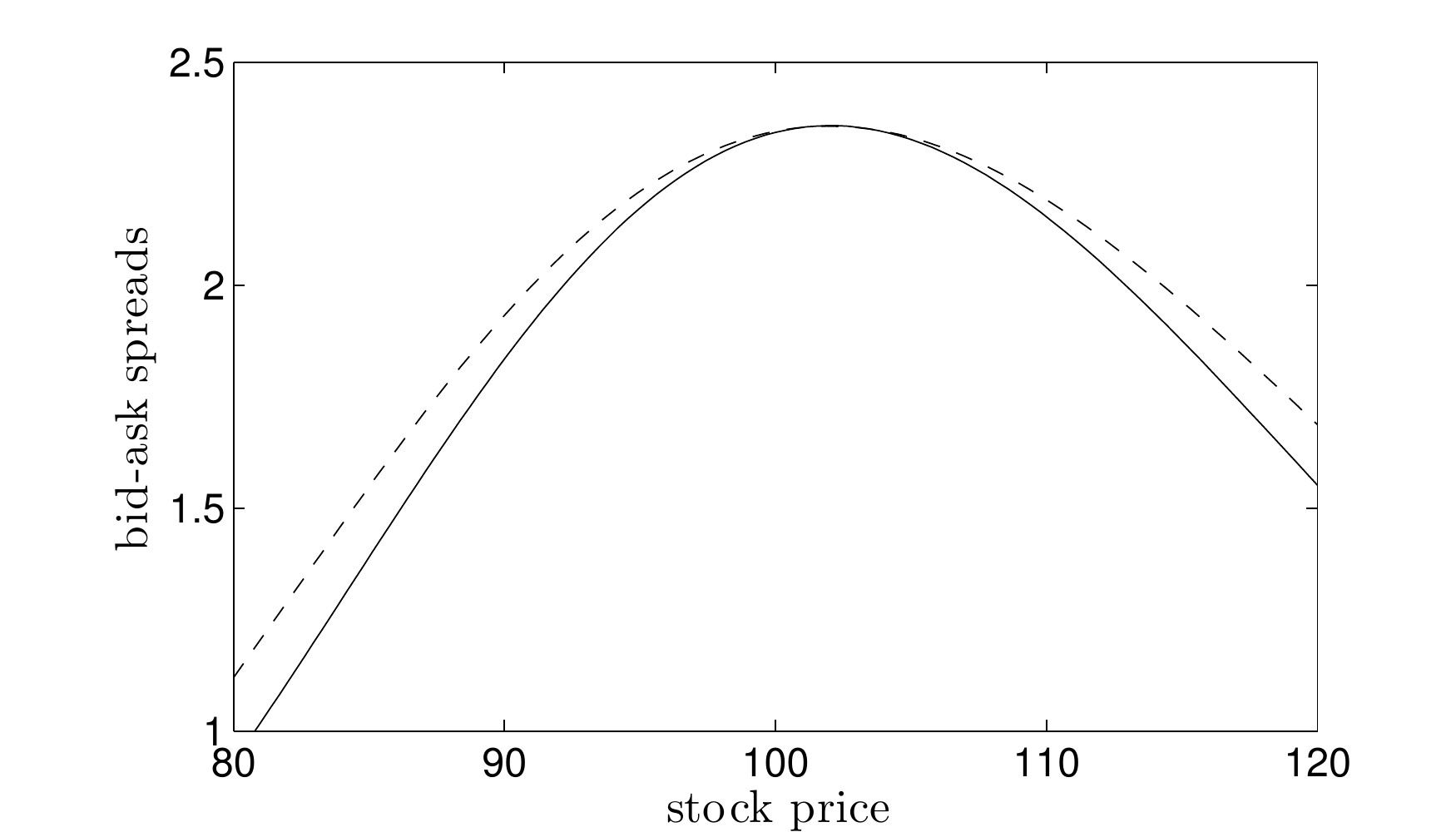}
\end{subfigure}%
\begin{subfigure}{.5\textwidth}
  \centering
  \includegraphics[scale=.4]{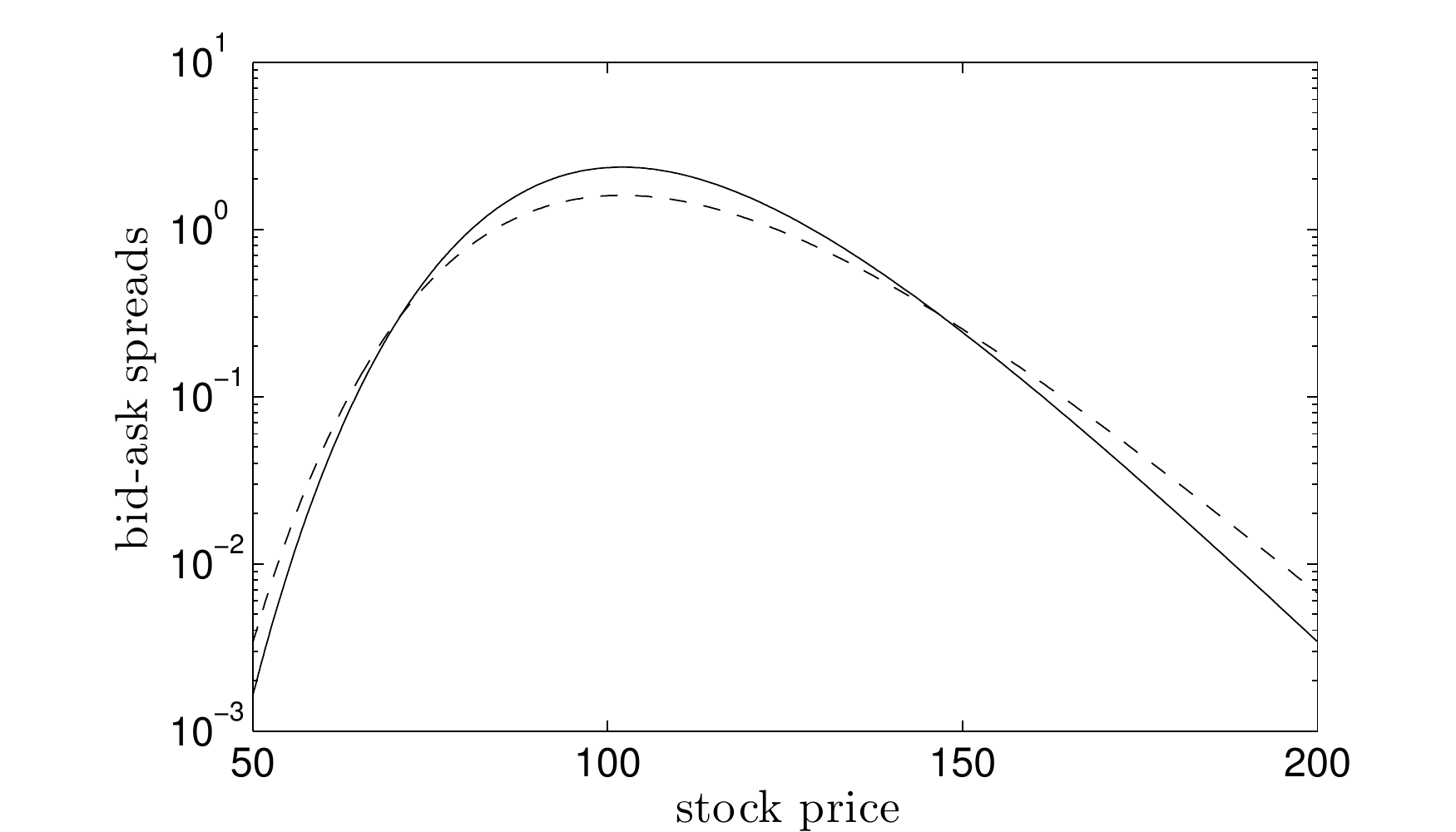}
\end{subfigure}%
\caption{First-order bid-ask spread $p_a(\psi)-p_b(\psi)$ (solid) and bid-ask spread corresponding to the UVM (dashed) with two different volatility bands and stock price ranges; the setting is the same as in Figure \ref{fig:smooth put:value}. In the left panel, the volatility band $[17.04\%,22.96\%]$ is chosen such that the at-the-money spreads coincide. In the right panel, the band is $[18\%,22\%]$ and the vertical axis has a logarithmic scale.}
\label{fig:smooth put:spreads}
\end{figure}

\begin{align}
\label{eqn:intro:Feynman-Kac}
\wt w(t,s)
&=\frac{1}{2}\cEX{\int_t^T \left(\bar\sigma(u,S_u) \bar\Gamma^\$_u\right)^2 \dd u}{S_t = s}.
\end{align}
We see that the option's sensitivity to volatility uncertainty is determined by the expected volatility-weighted \emph{cash gamma} $\bar\Gamma^\$_u :=  S_u^2 \bar V_{ss}(u,S_u)$ accumulated over the remaining lifetime of the option.\footnote{Our results can be formally linked to the UVM with a random, time-dependent volatility band depending on the option's cash gamma and the agent's uncertainty aversion; cf.~Remark~\ref{rem:BSB}. Note that the cash gamma also plays a crucial role in the asymptotic analysis of other frictions such as discrete rebalancing \cite{BertsimasKoganLo2000}, transaction costs \cite{WhalleyWilmott1997}, price impact \cite{MoreauMuhleKarbeSoner2015}, or jumps \cite{CernyDenklKallsen2013}.} Also note that formula \eqref{eqn:intro:Feynman-Kac} is independent of the investor's utility function $U$. Thus, in the setting of this paper, uncertainty aversion dominates risk aversion at the leading order.

Let us illustrate our results for a deliberately simple case of a constant volatility reference model and a ``smooth put'' option, whose payoff is a smoothed version of a standard put option with strike $100$.\footnote{\label{ftn:smooth put}Specifically, the payoff is the Black--Scholes value of a standard put option with strike $100$ and maturity $1$ day.} Figure~\ref{fig:smooth put:value} displays the first-order approximations of the indifference bid and ask prices $p_b(\psi)$ and $p_a(\psi)$ for $\psi = 10^{-3}$ as functions of the stock price. It is instructive to compare the corresponding first-order approximation of the bid-ask spread, $2 \wt w \psi$, with the bid-ask spread implied by the UVM. In the left panel of Figure~\ref{fig:smooth put:spreads}, the volatility band is chosen such that the at-the-money UVM-spread\footnote{As the option payoff is convex, this spread is simply the difference between the Black--Scholes values of the ``smooth put'' corresponding to the two endpoints of the volatility band.} equals our spread $2\wt w \psi$. The only reason for this choice is to improve the visual comparability of both spreads. Indeed, the worst-case approach requires the volatility band to be chosen in accordance with the agent's beliefs, which would arguably yield a wider volatility band. We see that the spread $2\wt w \psi$ becomes smaller than the UVM-spread if the stock price moves away from the strike to a region where the cash gamma of the option is lower. The right panel of Figure~\ref{fig:smooth put:spreads} compares the spreads for a wider range of stock prices and an even narrower volatility band. Our bid-ask spread exceeds the UVM-spread if the stock price is close to the strike, where the cash gamma and the option's sensitivity to changes in volatility is high. Conversely, our bid-ask spread is lower than the UVM-spread if the stock price is far away from the strike, where the option's sensitivity to volatility is low. This illustrates that our approach adjusts to the structure of the option's sensitivity to volatility misspecification in a more nuanced manner than the UVM.

Let us emphasise, however, that the bid-ask spread implied by our results is \emph{not} generically close to that of the UVM. For instance, the UVM-spread for two (smooth) put options is twice the spread for a single one, whereas the spread implied by our results scales quadratically in the option's cash gamma and thus grows by a factor of four (cf.~\eqref{eqn:intro:Feynman-Kac}). Put differently, if you are already short one option, you require a larger compensation for selling another option of the same type as this makes your position more vulnerable to volatility misspecification.

We next address the corresponding asymptotically optimal hedging strategy. Up to some technical modifications (cf.~Theorem~\ref{thm:second order}), it is given by the delta hedge derived from the uncertainty-adjusted option price:
\begin{align*}
\theta^\psi_t
&= \bar\Delta_t + \wt w_s(t,S_t) \psi,
\end{align*}
where $\bar\Delta_t = \bar V_s(t,S_t)$ is the option's \emph{delta} in the reference model.\footnote{Here and in the following, subscripts on functions denote the corresponding partial derivatives.} The adjustment $\wt w_s(t,S_t) \psi$ can be interpreted as a hedge against movements of the stock price into zones of high susceptibility to model misspecification; see~Section~\ref{sec:main results} for more details. Like the price adjustment, the leading-order hedge correction is also determined by uncertainty aversion alone, independent of risk aversion.\footnote{The strategy adjustment may become dependent on risk aversion for other penalty functionals; cf.~the discussion following Theorem~\ref{thm:second order}.}

For the results described so far, we provide a rigorous verification theorem subject to sufficient regularity conditions, which should be seen as a proof of concept. To illustrate the much wider scope of our approach, we also provide formal\footnote{A rigorous verification of these results would proceed along the same lines as for the simpler benchmark case discussed here. In order not to drown the ideas in further technicalities resulting from even more state variables, regularity conditions, etc. we do not pursue this here.} extensions that allow to cover practically relevant settings including exotic options,\footnote{Worst-case hedges for some exotics have been derived by \cite{Hobson1998.lookback, BrownHobsonRogers2001, CoxObloj2011.notouch, CoxObloj2011.touch, HobsonNeuberger2012, HobsonKlimmek2012, HobsonKlimmek2015, Stebegg2014}, for example.} option portfolios, or semi-static hedging with vanilla options and the underlying stock. We find that for many exotics like Asian, lookback, or barrier options, the above representations for the price and hedge corrections remain valid. The only change is that the option's price in the reference model and its corresponding cash gamma depend on further state variables in these cases. For other exotics such as options written on the realised variance, further greeks computed in the reference model come into play; cf.~Example~\ref{ex:exotic:realised variance}. Concerning option portfolios and semi-static hedging, suppose that your book consists of a variety of non-traded options with different maturities and payoff profiles, which you want to hedge by dynamic trading in the stock, and by setting up a static position in some liquidly traded vanilla options. Suppose further that you can buy or sell at time $0$ any quantities $\lambda_1,\ldots,\lambda_M$ of $M$ different liquid vanilla options, say calls, at prices to which your reference model is calibrated.\footnote{This is the analogue of the \emph{Lagrangian uncertain volatility model} \cite{AvellanedaParas1996}; also compare \cite{Mykland2003.options}.} Then given a choice of $\lambda$, the cash equivalent of the combined portfolio is\footnote{The same formula obtains -- mutatis mutandis -- for exotics of Asian-, lookback-, or barrier-type, after adding the appropriate state variables to the reference cash gamma.}
\begin{align}
\label{eqn:intro:portfolio cash equivalent}
\wt w(t,s;\lambda)
&= \frac{1}{2}\cEX{\int_t^T \left(\bar\sigma(u,S_u)\Big(\bar\Gamma^{\$,0}_u - \sum_{i=1}^M \lambda_i \bar\Gamma^{\$,i}_u \Big)\right)^2 \dd u}{S_t = s},
\end{align}
where $\bar\Gamma^{\$,0}$ is the net cash gamma of your original book (before buying or selling calls) and $\bar\Gamma^{\$,i}$ is the cash gamma of the $i$-th call. The corresponding hedge is
\begin{align*}
\theta^\psi_t
&= \bar\Delta^0_t - \sum_{i=1}^M \lambda_i \bar\Delta^i_t + \wt w_s(t,S_t;\lambda)\psi,
\end{align*}
where $\bar\Delta^0$ is the net delta of the original book and $\bar\Delta^i$ is the delta of the $i$-th call. Formula \eqref{eqn:intro:portfolio cash equivalent} provides a simple yet theoretically founded criterion to optimise the resilience of a derivative book against model misspecification. Indeed, since the cash equivalent is a measure for the portfolio's susceptibility to model misspecification, minimising \eqref{eqn:intro:portfolio cash equivalent} over $\lambda$ renders your combined portfolio more robust. This minimisation corresponds to balancing out the cash gamma of your net position by trading appropriately in the liquid options, in analogy to corresponding results for transaction costs \cite[Section 3.1]{KallsenMuhleKarbe2015}. The mapping
\begin{align*}
\bar\Gamma^{\$,0}
\mapsto
\inf_\lambda \frac{1}{2}\EX{\int_0^T \left(\bar\sigma(u,S_u)\Big(\bar\Gamma^{\$,0}_u - \sum_{i=1}^M \lambda_i \bar\Gamma^{\$,i}_u \Big)\right)^2 \dd u}
\end{align*}
can be interpreted as a ``measure of model uncertainty'' in the sense of Cont \cite{Cont2006}. That is, it satisfies certain desirable properties that, unlike standard risk measures, take into account model-independent hedging strategies and the availability of options as hedging instruments; cf.~Section~\ref{sec:measure} for more details.

From a mathematical point of view, this paper can be seen as a case study of an asymptotic analysis of a two-player, zero-sum stochastic differential game where both the drift and the diffusion coefficients of the state variables are controlled. Instead of the predominating Elliott--Kalton formulation used in the seminal work of Fleming and Souganidis \cite{FlemingSouganidis1989} and many articles thereafter, we use a weak formulation similar to Pham and Zhang \cite{PhamZhang2014}. This approach turns out to be both more natural and more convenient for the hedging problem at hand. The candidate expansion of the value function and  candidates for almost optimal controls are derived from an appropriate ansatz substituted into the Hamilton--Jacobi--Bellman--Isaacs equation. The proof adapts classical verification arguments of stochastic optimal control to the asymptotic setting.

The remainder of the article is organised as follows. Our framework for hedging with volatility uncertainty is introduced in Section~\ref{sec:problem formulation}. Subsequently, we state and discuss our main results in Section~\ref{sec:main results}. A partially heuristic derivation and formal extensions can be found in Section~\ref{sec:extensions}. Finally, the rigorous proofs of the results from Section~\ref{sec:main results} are collected in Section~\ref{sec:proofs}. 

\section{Problem formulation}
\label{sec:problem formulation}

\subsection{Hedging with a local volatility model}
\label{sec:hedging}

Consider an agent (hereafter called ``you'') who has written a vanilla option on a stock $S$ with payoff $G(S_T)$ at maturity $T > 0$. We assume that you can trade frictionlessly in the stock and a bank account with zero interest rate. If you are risk averse, you will try to hedge your exposure to the option by trading ``appropriately'' in the underlying. 

Let us assume that the true dynamics of the stock price process $S = (S_t)_{t\in[0,T]}$ are governed by the SDE
\begin{align}
\label{eqn:hedging:SDE}
\diff S_t
&= S_t \sigma_t \dd W_t, \quad S_0 = s_0 > 0,
\end{align}
for a Brownian motion $W$ and a volatility process $\sigma$. A typical approach to hedge options is to postulate a class of models for the stock price dynamics, calibrate it using market prices of liquidly traded derivatives, and act as if the model was correct. A simple, but popular choice for this procedure is the class of \emph{local volatility models} \cite{Dupire1994}, which assume that the volatility process in \eqref{eqn:hedging:SDE} is a deterministic function of time and the stock price itself, $\sigma_t = \bar\sigma(t,S_t)$. Let us briefly recall how this model can be used to hedge the (sufficiently integrable) option $G(S_T)$ through self-financing trading in the underlying (cf., e.g., \cite{Gatheral2006} for more details). Self-financing trading strategies are parametrised by progressively measurable processes $\theta=(\theta_t)_{t\in[0,T]}$ describing the number of shares held at time $t$. If your initial capital is $x_0$ and you trade according to $\theta$, your wealth at time $T$ is 
\begin{align*}
x_0 + \int_0^T \theta_t \dd S_t.
\end{align*}
Applying It\^o's formula to the process $\bar V(t,S_t)$ for some $\bar V \in C^{1,2}((0,T)\times\RR_+) \cap C([0,T]\times\RR_+)$ and using that the true dynamics of $S$ are given by \eqref{eqn:hedging:SDE}, we obtain
\begin{align}
\label{eqn:hedging:ito}
\bar V(t,S_t)
&= \bar V(0,s_0) + \int_0^t \bar V_s(u,S_u) \dd S_u + \int_0^t \left( \bar V_t(u,S_u) + \frac{1}{2} \sigma_u^2 S_u^2 \bar V_{ss}(u,S_u) \right) \dd u.
\end{align}
Note that if $\sigma_t = \bar\sigma(t,S_t)$, i.e., if your local volatility model is correct, and if $\bar V$ solves the Black--Scholes PDE
\begin{align}
\label{eqn:hedging:PDE}
\begin{split}
\bar V_t(t,s) + \frac{1}{2} \bar\sigma(t,s)^2 s^2 \bar V_{ss}(t,s)
&=0, \quad (t,s) \in (0,T)\times\RR_+,\\
\bar V(T,s)
&= G(s), \quad s \in \RR_+,
\end{split}
\end{align}
then \eqref{eqn:hedging:ito} for $t = T$ reduces to
\begin{align}
\label{eqn:hedging:replication}
G(S_T)
&= \bar V(T,S_T)
= \bar V(0,s_0) + \int_0^T \bar V_s(u,S_u) \dd S_u.
\end{align}
Whence, you can perfectly replicate the option payoff $G(S_T)$ by self-financing trading with initial capital $\bar V(0,s_0)$ and trading strategy $\bar\Delta_t := \bar V_s(t,S_t)$, the so-called \emph{delta hedge}. Given that your local volatility model is correct, $\bar{V}(t,S_t)$ is in turn the unique price for the option that is compatible with the absence of arbitrage. 

What happens if you delta-hedge even though your local volatility model is not correct? If the true dynamics are given by the SDE \eqref{eqn:hedging:SDE} for some volatility process $\sigma$, but the delta hedge $\bar\Delta$ is determined from a local volatility model via the PDE \eqref{eqn:hedging:PDE}, then \eqref{eqn:hedging:ito} can be rewritten as
\begin{align}
\label{eqn:hedging:theoretical value}
\bar V(t,S_t) 
&= \bar V(0,s_0) + \int_0^t \bar\Delta_u \dd S_u + \frac{1}{2}\int_0^t S_u^2 \bar V_{ss}(u,S_u) (\sigma_u^2 - \bar\sigma(u,S_u)^2) \dd u.
\end{align}
For a trading strategy $\theta$, define the \emph{Profit\&Loss (P\&L)} $Y_t$ at time $t$ as the difference between the current wealth in your hedge portfolio and the \emph{theoretical} value of the option \emph{in your model}:
\begin{align}
\label{eqn:hedging:Y definition}
Y_t
&= x_0 + \int_0^t \theta_u \dd S_u - \bar V(t,S_t).
\end{align}
As $\bar V(T,S_T) = G(S_T)$, the terminal value $Y_T$ is also your \emph{actual} P\&L at maturity $T$. Substituting \eqref{eqn:hedging:theoretical value} into \eqref{eqn:hedging:Y definition} gives
\begin{align}
\label{eqn:hedging:Y}
Y_t
&= x_0 - \bar V(0,s_0) + \int_0^t (\theta_u - \bar\Delta_u) \dd S_u + \frac{1}{2}\int_0^t S_u^2 \bar V_{ss}(u,S_u) (\bar\sigma(u,S_u)^2 - \sigma_u^2) \dd u.
\end{align}
The last term in \eqref{eqn:hedging:Y} describes the hedging error you incur if you use the delta hedge $\theta_t = \bar\Delta_t$, starting from initial capital $x_0 = \bar V(0,s_0)$.\footnote{This error incurred from hedging with a misspecified volatility is well known in the literature and in practice. To the best of our knowledge, it appeared first in Lyons \cite[Equation (27)]{Lyons1995} (see also \cite[Equation (6.7)]{ElKarouiJeanblancShreve1998}).}

Let us now formalise that risk-averse investors indeed hedge their exposure to the option. To this end, assume you are certain that the dynamics of the stock price follow your local volatility model, i.e., $\sigma_t = \bar\sigma(t,S_t)$, and that you want to maximise your expected utility from your final P\&L:
\begin{align}
\label{eqn:hedging:utility maximisation}
\EX{U(Y_T)}
&= \EX{U\Big(x_0 + \int_0^T \theta_t \dd S_t - G(S_T)\Big)}
\longrightarrow \max_\theta !
\end{align}
Here, $\theta$ runs through some subset of trading strategies such that $\int_0^\cdot ( \theta_t - \bar V_s(t,S_t) ) \dd S_t$ is a supermartingale (to rule out doubling strategies) and $U$ is a (strictly concave) utility function on $\RR$. As the option can be perfectly replicated, this is a trivial example of a utility maximisation problem with a random endowment. By Jensen's inequality, \eqref{eqn:hedging:replication}, and the supermartingale property of $\int_0^\cdot ( \theta_t - \bar V_s(t,S_t) ) \dd S_t$, we have
\begin{align*}
\EX{U\Big(x_0 + \int_0^T \theta_t \dd S_t - G(S_T)\Big)}
&\leq U\Big(x_0 + \EX{\int_0^T (\theta_t - \bar V_s(t,S_t)) \dd S_t} - \bar V(0,s_0)\Big)\\
&\leq U(x_0 - \bar V(0,s_0)),
\end{align*}
with equality for $\theta_t = \bar V_s(t,S_t) = \bar\Delta_t$. This proves the following lemma:

\begin{lemma}
\label{lem:hedging}
Suppose that $\sigma_t = \bar\sigma(t,S_t)$ and that $\bar V \in C^{1,2}((0,T)\times\RR_+)\cap C([0,T]\times\RR_+)$ solves the PDE~\eqref{eqn:hedging:PDE}. Then the delta hedge $\bar\Delta_t = \bar V_s(t,S_t)$ maximises \eqref{eqn:hedging:utility maximisation} over any set $\cA \ni \bar\Delta$ of strategies $\theta$ such that $\int_0^\cdot ( \theta_t - \bar V_s(t,S_t) ) \dd S_t$ is a supermartingale. Moreover:
\begin{align*}
\sup_{\theta \in \cA} \EX{U\Big(x_0 + \int_0^T \theta_t \dd S_t - G(S_T)\Big)}
&= U(x_0 - \bar V(0,s_0)).
\end{align*}
\end{lemma}

\begin{remark}
\label{rem:drift}
Let us briefly discuss why we assume that $S$ has zero drift. Suppose
\begin{align*}
\diff S_t
&= S_t (\mu_t \dd t + \bar\sigma(t,S_t) \dd W_t)
\end{align*}
for a suitable drift rate $\mu = (\mu_t)_{t\in[0,T]}$. If $\theta^*$ is an optimal strategy for the pure investment problem
\begin{align}
\label{eqn:rem:drift:investment problem}
\EX{U\Big(x_0 - \bar V(0,s_0) + \int_0^T \theta_t \dd S_t\Big)}
\longrightarrow \max_\theta !,
\end{align}
then it follows (using that the replication \eqref{eqn:hedging:replication} also works with drift) that $\theta^* + \bar\Delta$ is an optimal strategy for the mixed hedging/investment problem \eqref{eqn:hedging:utility maximisation}, irrespective of the drift rate $\mu$. Whence, without model uncertainty, the hedging component $\bar\Delta$ of the optimal strategy does not depend on $\mu$.
Therefore, it is expected that also with (small) uncertainty aversion (about this \emph{complete} reference model), the drift rate only has a small effect on the hedging component.

Assuming that $S$ has zero drift allows us to focus on the hedging rather than the investment component of the optimal strategy. Indeed, with zero drift, there is no incentive to trade the stock other than as a hedging instrument for the option. This is a reasonable approximation at least if the time horizon is not too long so that stock price movements are dominated by Brownian fluctuations.
\end{remark}

In Section~\ref{sec:hedging under uncertainty}, we introduce a modification of the utility maximisation problem \eqref{eqn:hedging:utility maximisation} that takes into account uncertainty about the true dynamics of the stock price. Its formulation requires a setup that can model the flow of observable information and the possibility of different volatility scenarios simultaneously, which we develop in Sections \ref{sec:setup} and \ref{sec:hedging under uncertainty}.

\subsection{Volatility uncertainty setup}
\label{sec:setup}

Fix a time horizon $T > 0$ and constants $s_0 > 0$ and $y_0 \in \RR$. Let
\begin{align*}
\Omega
&= \lbrace \omega = (\omega^S_t, \omega^Y_t)_{t\in[0,T]} \in C([0,T]; \RR^2) : \omega_0 = (s_0, y_0) \rbrace
\end{align*}
be the canonical space of continuous paths in $\RR^2$ starting in $(s_0, y_0)$, endowed with the topology of uniform convergence. Moreover, let $\cF = \cB(\Omega)$ be the Borel $\sigma$-algebra on $\Omega$. We denote by $S = (S_t)_{t \in [0,T]}$ and $Y = (Y_t)_{t\in[0,T]}$ the first and second component of the canonical process, respectively, i.e., $S_t(\omega) = \omega^S_t$ and $Y_t(\omega) = \omega^Y_t$, and by $\FF = (\cF_t)_{t\in[0,T]}$ the (non-augmented, non-right-continuous) filtration generated by $(S,Y)$. Unless otherwise stated, all probabilistic notions requiring a filtration, such as progressive measurability etc., pertain to $\FF$.

\begin{remark}
The measurable space $(\Omega, \cF)$ together with the filtration $\FF$ models the observable variables, the stock price $S$ and the P\&L $Y$, independently of any specific choice of probability measure on $(\Omega, \cF)$. As we shall see, this allows us to use progressively measurable processes as controls. If instead one formulated the hedging problem as a stochastic differential game on a single filtered probability space, one would have to allow Elliott--Kalton strategies (see, e.g.,~\cite{FlemingHernandezHernandez2011}) as controls; these are controls that may depend on the other player's control in a non-anticipative way. In the context of financial applications, it also seems more natural to treat the state processes as publicly observable, but not necessarily the controls. In our setting, the stock price is quoted on the market and you can easily determine your P\&L by looking at your trading account.
\end{remark}

Next, we introduce a large class of probability measures on $(\Omega, \cF)$ which correspond to different choices of strategies $\theta$ and volatility processes $\sigma$. Let $\ol\fS$ denote the set of triplets $(\theta, \sigma, P)$ consisting of progressively measurable processes $\theta = (\theta_t)_{t\in[0,T]}$ and $\sigma = (\sigma_t)_{t\in[0,T]}$ and a probability measure $P$ on $(\Omega, \cF)$ such that $S$ is a (continuous) local $P$-martingale with quadratic variation $\diff\langle S \rangle_t= S_t^2 \sigma_t^2 \dd t$ and $Y$ is a (continuous) $P$-semimartingale with canonical decomposition
\begin{align}
\label{eqn:setup:Y}
Y
&= y_0 + \int_0^\cdot (\theta_t - \bar V_s(t,S_t)) \dd S_t + \int_0^\cdot \frac{1}{2}S_t^2 \bar V_{ss}(t,S_t) (\bar\sigma(t,S_t)^2 - \sigma_t^2) \dd t
\end{align}
under $P$.\footnote{Under each $P$, the stochastic integrals of sufficiently integrable, progressively measurable integrands against It\^o process integrators are constructed as $\FF$-progressively measurable stochastic processes with $P$-a.s.~continuous paths; see the discussion before Lemma 4.3.3 in Stroock and Varadhan \cite{StroockVaradhan1979}.} For later reference, note that this implies that under $P$,
\begin{align}
\label{eqn:setup:covariation}
\begin{split}
\diff\langle S \rangle_t
&= S_t^2 \sigma_t^2 \dd t, \\
\diff\langle Y \rangle_t
&= (\theta_t - \bar V_s(t,S_t))^2 S_t^2 \sigma_t^2 \dd t,\\
\diff\langle S, Y \rangle_t
&= (\theta_t - \bar V_s(t,S_t)) S_t^2 \sigma_t^2 \dd t.
\end{split}
\end{align}
For each triplet $(\theta, \sigma, P)$, the probability measure $P$ corresponds to the distribution of the stock price and P\&L processes when you choose the strategy $\theta$ and ``nature'' chooses the volatility $\sigma$; cf.~the dynamics of $S$ in \eqref{eqn:hedging:SDE} and $Y$ in \eqref{eqn:hedging:Y}.

\begin{remark}
This is a setup for a weak formulation of a stochastic differential game similar to \cite{PhamZhang2014}. Instead of controlling a process directly, the players can control the distribution of the observable processes $S$ and $Y$ through their choice of probability measure. ``Nature's'' choice of volatility determines the stock price volatility and the drift of the P\&L process. Your choice of strategy only affects the volatility of your P\&L.
\end{remark}

\subsection{Hedging under volatility uncertainty}
\label{sec:hedging under uncertainty}

The set $\ol\fS$ defined in Section~\ref{sec:setup} describes a very large set of possible distributions for the stock price process $S$ and your P\&L process $Y$. As some of these scenarios might be implausible (they could, for instance, include doubling strategies), let us fix some subset $\fS\subset\ol\fS$; for our specific choice, see \eqref{eqn:process bounds} below. Next, fix sets $\cA$ and $\cV$ of trading strategies and volatilities that you want to take into account. In principle, we could consider pairs $(\theta,\sigma) \in \cA\times\cV$ for which there is at least one probability measure $P$ such that $(\theta,\sigma,P)\in\fS$. But as your choice of trading strategy should not restrict ``nature's'' choice of volatilities and vice versa, we require that the whole rectangle $\cA\times\cV$ has this property:
\begin{align*}
\cA \times \cV
&\subset \cZ
:= \lbrace (\theta, \sigma) : \fP(\theta, \sigma) \not=\emptyset \rbrace,
\end{align*}
where $\fP(\theta, \sigma)$ denotes the set of probability measures $P$ such that $(\theta,\sigma,P)\in\fS$.

\begin{remark}
Note that $\fP(\theta,\sigma)$ can contain more than one $P$. Indeed, each weak solution of the SDE
\begin{align*}
\diff S_t
&= S_t \sigma_t((S,Y)) \dd W_t,\\
\diff Y_t
&= (\theta_t((S,Y)) - \bar V_s(t,S_t)) \dd S_t + \frac{1}{2}S_t^2 \bar V_{ss}(t,S_t) (\bar\sigma(t,S_t)^2 - \sigma_t((S,Y))^2) \dd t,
\end{align*}
with initial distribution $\delta_{(s_0,y_0)}$ gives rise to a measure in $\fP(\theta,\sigma)$. So $\fP(\theta,\sigma)$ is a singleton if and only if the SDE has a solution which is unique in law.
\end{remark}

A popular approach to incorporate volatility uncertainty into utility maximisation problems is to treat all volatilities in $\cV$ equally and to maximise the expected utility for the worst-case volatility. In the setting of this paper, this would lead to\footnote{Since $\fP(\theta,\sigma)$ may contain more than one measure, we also have to allow ``nature'' to choose a specific measure.}
\begin{align*}
\inf_{\sigma\in\cV} \inf_{P\in\fP(\theta,\sigma)} \EX[P]{U(Y_T)}
\longrightarrow \max_{\theta \in \cA} !
\end{align*}
More generally, one can introduce a penalty functional $\fa:\cV \to L^0_+(\Omega,\cF)$ and consider\footnote{The penalty functional in the sense of the criterion \eqref{eqn:intro:numerical representation} is $\alpha(P) = \EX[P]{\fa(\sigma^P)}$, where $\sigma^P$ is the volatility of returns of $S$ under $P$.}
\begin{align*}
\inf_{\sigma\in\cV} \inf_{P\in\fP(\theta,\sigma)} \EX[P]{U(Y_T)+\fa(\sigma)}
\longrightarrow \max_{\theta \in \cA} !
\end{align*}
This approach allows to weigh volatility scenarios according to their plausibility -- the higher the penalty, the less seriously you take the model.\footnote{Recall from Footnote \ref{ftn:penalty} that $\fa$ penalises the fictitious adversary (``nature'') and not the agent.} We suppose that a local volatility model with volatility function $\bar\sigma$ is your best guess for the true dynamics of $S$. In the following, we consider the penalty functional 
\begin{align*}
\fa(\sigma)
&= \frac{1}{\psi}\int_0^T U'(Y_t)f(t, S_t,Y_t; \sigma_t) \dd t.
\end{align*}
Here, $\psi > 0$ is a parameter and the \emph{penalty function} $f: [0,T] \times \RR\times\RR \times \RR_+ \to \RR_+$ is a sufficiently smooth function such that $\varsigma\mapsto f(t,s,y;\varsigma)$ is strictly convex and has a minimum of $0$ at $\bar\sigma(t,s)$, i.e.,
\begin{align}
\label{eqn:penalty function:minimum conditions}
f(t,s,y;\bar\sigma(t, s))
&= \frac{\partial f}{\partial\varsigma}(t,s,y;\bar\sigma(t,s))
= 0.
\end{align}
In particular, $\fa$ does not penalise the reference local volatility model, i.e., ${\fa((\bar\sigma(t,S_t))_{t\in[0,T]}) = 0}$. A typical example for the penalty function is $f(t,s,y;\varsigma) =\frac{1}{2} \left(\varsigma - \bar\sigma(t, s) \right)^2$; then, deviations of the volatility from the reference model are penalised in a mean-square sense.\footnote{Notably, \cite{AvellanedaFriedmanHolmesSamperi1997} show that penalty functionals of this form can arise as the continuous-time limit of the relative entropy in a discrete-time approximation.} The function $f$ describes the ``shape'' of the uncertainty aversion, while the parameter $\psi$ quantifies its ``magnitude''. Indeed, if $\psi$ is small, then volatility processes alternative to the reference volatility are penalised heavily, i.e., uncertainty aversion is small.

For each $\psi > 0$, $\theta \in \cA$, $\sigma \in \cV$, and $P \in \fP(\theta,\sigma)$, we define the \emph{objective} of our hedging problem by
\begin{align}
\label{eqn:objective}
J^\psi(\sigma, P)
&:= \EX[P]{U(Y_T)+\frac{1}{\psi}\int_0^T U'(Y_t)f(t,S_t,Y_t;\sigma_t) \dd t}
\end{align}
and its \emph{value} by
\begin{align}
\label{eqn:value}
v(\psi)
&= v(\psi; \cA,\cV)
:= \sup_{\theta \in \cA} \inf_{\sigma\in\cV} \inf_{P\in\fP(\theta,\sigma)} J^\psi(\sigma,P).
\end{align}

\begin{remark}
\label{rem:marginal utility}
Note that the term $U'(Y_t)$ in the penalty part of the objective \eqref{eqn:objective} does not restrict the generality of our results. Indeed, the factor $U'(Y_t)$ could effectively be removed if desired by working with a penalty function of the form $f(t,s,y;\varsigma) = \wt f(t,s,y;\varsigma)/U'(y)$ for a suitable function $\wt f$. We choose the above formulation for the following reasons.

First, in the standard expected utility framework, preferences are invariant under affine transformations of the utility function. The term $U'(Y_t)$ ensures that this property is preserved for uncertainty-averse decision makers whose preferences are described by \eqref{eqn:objective}--\eqref{eqn:value}. In particular, optimal strategies and volatilities for these preferences are invariant under affine transformations of the scale on which utility is measured. Second, $U'(Y_t)$ (and not, e.g., $U'(y_0)$\footnote{Using $U'(y_0)$ instead of $U'(Y_t)$ would yield the same expansion for $v(\psi)$ as in Theorem~\ref{thm:second order} and, formally and at the leading-order, the same almost optimal strategies and volatilities. This is because in the asymptotic limit for small uncertainty aversion, the P\&L process converges to a constant.}) is the natural choice for a dynamic formulation of the hedging problem \eqref{eqn:value} in terms of a family of conditional problems parametrised by the initial time $t$, stock price $S_t = s$, and P\&L $Y_t = y$. Third, our results show that if $f$ does not depend on the P\&L variable $y$, then the preferences have approximately ``constant uncertainty aversion'' in the sense that the cash equivalent $\wt w$ does not depend on the P\&L (cf.~Proposition~\ref{prop:Feynman-Kac}). This would not be the case if one omitted the term $U'(Y_t)$ in \eqref{eqn:objective}.\footnote{In the context of robust portfolio choice, Maenhout \cite{Maenhout2004} also observes that some modification of the standard (non wealth-dependent) entropic penalty is reasonable to avoid that the agent's uncertainty aversion wears off as her wealth rises, and tackles this effect by directly modifying the HJBI equation.}
\end{remark}

\begin{remark}
\label{rem:stochastic volatility}
The hedging problem with moderate uncertainty aversion about a \emph{stochastic volatility reference model} can be tackled as follows. Standard stochastic volatility models assume that the spot volatility of the stock price follows an It\^o diffusion driven by a Brownian motion which may be correlated with the stock price. In particular, the spot volatility is assumed to be perfectly observable at any time. Therefore, there is no uncertainty about the spot volatility. However, one can instead introduce uncertainty about the parameters describing the dynamics of the spot volatility: its drift and diffusion coefficients as well as its correlation with the stock price. Using a similar penalty functional as above, but now for three control variables instead of only the spot volatility, one can write down the HJBI equation for the corresponding hedging problem and analyse it by similar methods as in the present study. Of course, the analysis becomes substantially more involved due to the extra state variable and the multidimensional controls for the fictitious adversary.
\end{remark}

\section{Main results}
\label{sec:main results}

Our main result is a second-order expansion of the value $v(\psi)$ of the hedging problem \eqref{eqn:value} for small values of the uncertainty aversion parameter $\psi$. Moreover, we provide strategies $\theta^\psi$ and volatilities $\sigma^\psi$ that are almost optimal in the sense that their performance coincides with the optimal value up to the next-to-leading order $O(\psi^2)$. These expansions depend on the solutions to two linear second-order parabolic PDEs with source terms. Section \ref{sec:ingredients} introduces these PDEs as well as the assumptions underlying our main result. On a first reading, the reader may wish to skip Section \ref{sec:ingredients} and jump directly to the statement of the main result in Section \ref{sec:main result}.

\subsection{Ingredients and assumptions}
\label{sec:ingredients}

In order to present the main ideas in the verification as clearly as possible, we do not strive for minimal technical conditions in the following. Rather, we focus on a set of sufficient conditions and consider only triplets $(\theta, \sigma, P)$ such that the processes $S, Y$, and $\sigma$ are $P$-a.s.~uniformly bounded by a constant independent of $(\theta,\sigma, P)$. So fix constants $K > s_0 \vee s_0^{-1}$, $y_l < y_0$, and $y_u > y_0$ and define $\fS \subset \ol\fS$ as the set of triplets $(\theta, \sigma, P)$ such that
\begin{align}
\label{eqn:process bounds}
\sigma_t(\omega) \in [0,K], \quad
S_t(\omega) \in[K^{-1},K ], \quad
Y_t(\omega) \in (y_l, y_u)
\quad \text{for } \diff t \times P\text{-a.e. } (t, \omega) \in [0,T]\times\Omega.
\end{align}

\begin{remark}
The boundedness assumption on $S$ and $Y$ simplifies the proofs our main results. For instance, the compactness of the range of $(S,Y)$ simplifies many estimates, and gives, together with continuity of the candidate value function and its derivatives, the existence of point-wise saddle points for the min-max problem in the HJBI equation (cf.~Lemmas~\ref{lem:optimal volatility} and \ref{lem:optimal strategy}).

In a follow-up work \cite{HerrmannMuhleKarbe2015} on dynamic hedging with options in a particular setting, only the lower bound on the P\&L process is needed. Whence, we expect that with considerable extra effort, the boundedness of $S$ and the \emph{upper} bound for $Y$ can be relaxed in the present setting. Some \emph{lower} bound on the P\&L process is crucial, however (for instance, to exclude doubling strategies).
\end{remark}

The coefficients of the expansion of $v(\psi)$ are expressed in terms of solutions to two PDEs. Set $\bfD = (0,T)\times (K^{-1},K) \times(y_l, y_u)$. For each $y \in (y_l, y_u)$, we consider the linear second-order parabolic PDEs
\begin{align}
\label{eqn:PDE:first order}
\begin{split}
\wt w_t (t,s,y) + \frac{1}{2}\bar\sigma(t,s)^2 s^2 \wt w_{ss}(t,s,y) + \wt g(t,s,y)
&=0, \quad (t,s) \in (0,T)\times (K^{-1},K),\\
\wt w(T,s,y)
&= 0, \quad  s \in (K^{-1},K),\\
\wt w(t,s, y)
&= 0, \quad  t \in (0,T), s \in \lbrace K^{-1}, K \rbrace,
\end{split}
\end{align}
and
\begin{align}
\label{eqn:PDE:second order}
\begin{split}
\wh w_t(t,s,y) + \frac{1}{2}\bar\sigma(t,s)^2 s^2 \wh w_{ss}(t,s,y)
+ \wh g(t,s,y)
&=0, \quad (t,s) \in (0,T)\times (K^{-1},K),\\
\wh w(T,s,y)
&= 0, \quad  s \in (K^{-1},K),\\
\wh w(t,s, y)
&= 0, \quad  t \in (0,T), s \in \lbrace K^{-1}, K \rbrace,
\end{split}
\end{align}
where the source terms $\wt g, \wh g: \bfD \to \RR$ are given by\footnote{Here, we assume that all relevant partial derivatives exist; precise conditions are given in Assumption~\ref{ass:second order}.}
\begin{align}
\label{eqn:source term:first order}
\wt g(t,s,y) 
&:=  \frac{\left(\bar\sigma(t,s) s^2 \bar V_{ss}(t,s)\right)^2}{2 f''(t,s,y;\bar\sigma(t,s))},\\
\label{eqn:source term:second order}
\begin{split}
\wh g(t,s,y)
&:=  \frac{1}{6} \wt\sigma(t,s,y)^3 f^{(3)}(t,s,y;\bar\sigma(t,s))
-\frac{1}{2}\wt\sigma(t,s,y)^2 s^2 \bar V_{ss}(t,s)\\
&\qquad+ \wt\sigma(t,s,y) \bar\sigma(t,s) s^2 \left(\bar V_{ss}(t,s) \wt w_y(t,s,y)  -  \wt w_{ss}(t,s,y) \right)\\
&\qquad-\frac{U''(y)}{U'(y)} \left( \frac{1}{2} \bar\sigma(t,s)^2 s^2  \wt\theta(t,s,y)^2  -\wt\sigma(t,s,y)^2 f''(t,s,y;\bar\sigma(t,s)) \wt w(t,s,y) \right),
\end{split}
\end{align}
with
\begin{align}
\label{eqn:theta tilde}
\wt\theta(t,s,y)
&:= \wt w_s(t,s,y) + \frac{U'(y)}{U''(y)}\wt w_{sy}(t,s,y),\\
\label{eqn:sigma tilde}
\wt\sigma(t,s,y)
&:= \frac{\bar\sigma(t,s) s^2 \bar V_{ss}(t,s)}{f''(t,s,y;\bar\sigma(t,s))}.
\end{align}

We prove our main result under the following assumptions.
\begin{assumption}~
\label{ass:second order}
\begin{enumerate}
\item \emph{PDEs:} There are $\wt w, \wh w \in C^{1,2,2}(\bfD) \cap C(\ol\bfD)$ such that  for each $y \in (y_l, y_u)$, $\wt w(\cdot,\cdot,y)$ and $\wh w(\cdot,\cdot, y)$ are classical solutions to the PDEs \eqref{eqn:PDE:first order}--\eqref{eqn:PDE:second order} and
\begin{align}
\label{eqn:ass:pde derivatives}
\vert w_t \vert, \vert w_s \vert, \vert w_y \vert, \vert w_{ss} \vert, \vert w_{sy} \vert, \vert w_{yy} \vert \leq K \text{ on } \bfD, \quad w \in \lbrace \wt w, \wh w \rbrace.
\end{align}
\item \emph{Reference local volatility function:} $\bar\sigma:[0,T]\times[K^{-1},K]\to[0,K]$ is Borel-measurable, there is $\varepsilon > 0$ such that
\begin{align}
\label{eqn:ass:reference volatility}
\varepsilon \leq \bar\sigma(t,s) \leq K-\varepsilon \text{ on } [0,T]\times(K^{-1},K),
\;\;\text{and}\;\;
\bar\sigma(t, K) = \bar\sigma(t,K^{-1}) = 0 \text{ for } t\in [0,T].
\end{align}
\item \emph{Reference value:} $\bar V:[0,T]\times[K^{-1},K]\to\RR$ is Borel-measurable, $\bar V(t,\cdot) \in C^2((K^{-1},K)) \cap C([K^{-1},K])$ for all $t\in(0,T)$, and
\begin{align}
\label{eqn:ass:reference price}
\left\vert s^2 \bar V_{ss}(t,s)\right\vert
&\leq K \text{ for } (t,s) \in (0,T)\times(K^{-1},K).
\end{align}
\item \emph{Penalty function:} $f$ is $C^4$ in $\varsigma$ and the partial derivatives $f^{(k)}:=\frac{\partial^k}{\partial \varsigma^k}f$, $k=2,3,4$, satisfy
\begin{align}
\label{eqn:ass:penalty function}
\frac{1}{K} \leq f''(t,s,y;\varsigma)
&\leq K, \quad
\vert f^{(3)}(t,s,y;\varsigma)\vert , \vert f^{(4)}(t,s,y;\varsigma) \vert
\leq K \quad \text{on } \bfD\times[0,K].
\end{align}
\item \emph{Utility function:} $U:\RR \to \RR$ is $C^3$ with $U' > 0$ and $U'' < 0$ everywhere.
\end{enumerate}
\end{assumption}

Formally plugging the assumption $\bar\sigma(\cdot,K) = \bar\sigma(\cdot,K^{-1}) = 0$ into \eqref{eqn:source term:first order}--\eqref{eqn:source term:second order} and \eqref{eqn:sigma tilde} motivates to extend the definitions of $\wt\sigma, \wt g$, and $\wh g$ by setting
\begin{align}
\label{eqn:source term:boundary}
\wt\sigma(t,s,y)
&= \wt g(t,s,y) = \wh g(t,s,y) = 0, \quad t \in (0,T), s \in \lbrace K^{-1},K\rbrace, y \in (y_l, y_u).
\end{align}

\begin{remark}
\label{rem:assumptions}
(ii)--(iii) are assumptions on the reference volatility and the reference value of the option, while (iv)--(v) are regularity conditions for the objects describing your risk and uncertainty aversion. In contrast, (i) is an assumption on objects derived from these primitives through PDEs. Therefore, let us indicate here what kind of regularity assumptions are sufficient for (i) to hold. We focus on the PDE \eqref{eqn:PDE:first order} for $\wt w$; the PDE for $\wh w$ can be treated analogously. We first fix $y$ and note that the diffusion coefficient $\bar\sigma(t,s)^2 s^2$ is bounded away from zero on $(0,T)\times(K^{-1},K)$ by \eqref{eqn:ass:reference volatility}. Now, a classical existence and regularity result (see Friedman \cite{Friedman1964}, Theorem 7 in Section 3.3) guarantees the existence of a classical solution $\wt w(\cdot,\cdot,y) \in C^{1,2}((0,T)\times(K^{-1},K))$ with bounded and H\"older-continuous (in $t$ and $s$) partial derivatives $\wt w_t, \wt w_s, \wt w_{ss}$ provided that the diffusion coefficient and the source term are regular enough,\footnote{H\"older continuity uniformly on $(0,T)\times(K^{-1},K)$ suffices for this step.} and that the source term is compatible with the zero boundary condition in the sense that $\wt g(t,s,y) = 0$ for $s \in \lbrace K^{-1},K\rbrace$. Next, one can show that $\wt w$ has the Feynman--Kac representation (see also Proposition~\ref{prop:Feynman-Kac} below)
\begin{align}
\label{eqn:rem:assumptions:Feynman-Kac}
\wt w(t,s,y)
&= \EX[t,s]{\int_t^T \wt g(u,S_u,y) \dd u},
\end{align}
where the expectation is computed under a measure such that $S$ has the dynamics $\diff S_u = S_u \bar\sigma(u,S_u) \dd W_u$ and starts in $S_t = s$. Now, if $\wt g$ is $C^2$ in $y$ with bounded partial derivatives, one can infer from \eqref{eqn:rem:assumptions:Feynman-Kac} that the partial derivatives $\wt w_y, \wt w_{yy}$ exist and are bounded on $\bfD$. Finally, to obtain the existence and bounds for the cross partial derivative $\wt w_{sy}$, we can differentiate the PDE \eqref{eqn:PDE:first order} with respect to $y$ to obtain a PDE for $\wt w_y$. Then imposing even further regularity and compatibility conditions, the classical result cited above yields existence and boundedness of $\wt w_{sy}$.
\end{remark}

The uniform boundedness assumptions \eqref{eqn:process bounds} might appear restrictive. However, as the bounds can be chosen arbitrarily large or small, this restriction is of little practical relevance. A simple model satisfying the above assumptions is a geometric Brownian motion which is stopped as soon as it hits $K^{-1}$ or $K$ (for some large $K$). The corresponding local volatility function $\bar\sigma(t,s)$ would be a suitable constant on $[0,T]\times(K^{-1},K)$ and equal to zero on $[0,T]\times\lbrace K^{-1},K\rbrace$, in accordance with \eqref{eqn:ass:reference volatility}. Finally, the option payoff has to be regular enough (e.g., a ``smooth put'' shown in Figure~\ref{fig:smooth put:value} and defined in Footnote~\ref{ftn:smooth put}), so that (iii) holds and the arguments outlined in Remark \ref{rem:assumptions} go through.

\subsection{Main result}
\label{sec:main result}

We are now in a position to state our main result. A possible choice for the sets of trading strategies $\cA$ and volatility scenarios $\cV$ is provided in Theorem~\ref{thm:existence} below.

\begin{theorem}
\label{thm:second order}
Suppose that Assumption~\ref{ass:second order} holds and set $\tau := \inf\lbrace t \in [0,T]: \vert Y_t-y_0 \vert \geq 1 \rbrace \wedge T$. For each $\psi > 0$, define the strategy $\theta^\psi = (\theta^\psi_t)_{t\in[0,T]}$ and the volatility $\sigma^\psi = (\sigma^\psi_t)_{t\in[0,T]}$ by\footnote{Strictly speaking, $\theta^\psi$ and $\sigma^\psi$ have to be defined for \emph{every} $\omega \in \Omega$, even those for which $S$ or $Y$ exceeds the bounds \eqref{eqn:process bounds}. Outside these bounds, however, the functions $\bar V_s,\wt\theta,\bar\sigma,\wt\sigma$ are not defined. As we only consider measures $P$ such that \eqref{eqn:process bounds} holds, we do not make explicit the corresponding straightforward adjustments, which would only hamper readability.}
\begin{align*}
\theta^\psi_t
&= \bar\Delta_t + \wt\theta(t,S_t,Y_t)\1_{\lbrace t < \tau \rbrace}\psi,\\
\sigma^\psi_t
&= \bar\sigma(t,S_t) + \wt\sigma(t,S_t,Y_t) \psi,
\end{align*}
where $\wt\theta$ and $\wt\sigma$ are given in \eqref{eqn:theta tilde}--\eqref{eqn:sigma tilde}. Moreover, let $\cA$ and $\cV$ be sets of progressively measurable processes such that for all $\psi>0$ small enough, $(\theta^\psi, \sigma^\psi) \in \cA \times \cV \subset \cZ$. Then, as $\psi \downarrow 0$:
\begin{align*}
\sup_{\theta \in \cA} \inf_{\sigma \in \cV} \inf_{P \in \fP(\theta,\sigma)} J^\psi(\sigma, P)
&= U(y_0) - U'(y_0) \wt w(0,s_0,y_0)\psi + U'(y_0) \wh w(0,s_0,y_0)\psi^2 + o(\psi^2)\\
&= \inf_{P\in\fP(\theta^\psi,\sigma^\psi)} J^\psi(\sigma^\psi, P) + o(\psi^2),
\end{align*}
where $\wt w$ and $\wh w$ are the solutions to the PDEs \eqref{eqn:PDE:first order}--\eqref{eqn:PDE:second order}. In particular, $\theta^\psi$ is an optimal strategy at the next-to-leading order $O(\psi^2)$ among all strategies in $\cA$, and $\sigma^\psi$ is a ``worst-case'' volatility at the next-to-leading order $O(\psi^2)$ among all volatilities in $\cV$.
\end{theorem}

\begin{remark}
An inspection of the proof of Theorem \ref{thm:second order} shows that the strategy adjustment $\wt\theta$ only affects the performance of the hedge at the next-to-leading order $O(\psi^2)$. Put differently, the delta hedge $\bar\Delta$ is already optimal at the leading order $O(\psi)$. Indeed, with controls of the form $\bar\Delta_t + \wt\theta_t\psi$ and $\bar\sigma(t,S_t) + \wt \sigma_t \psi$ both the drift and the diffusion coefficient of your P\&L are of order $O(\psi)$ (cf.~the dynamics of Y in \eqref{eqn:setup:Y}). A formal Taylor expansion therefore suggests that the drift dominates the leading order $O(\psi)$ of the value expansion. As your choice of trading strategy only affects the diffusion coefficient while the drift coefficient is determined by the actual volatility, the strategy adjustment therefore only becomes visible at the order $O(\psi^2)$.
\end{remark}

The lengthy proof of Theorem \ref{thm:second order} is postponed to Section~\ref{sec:proof of main result}. Here, we first discuss the asymptotic formulas for the value $v(\psi)$ and the corresponding hedging strategy. The first-order terms in the value expansion and the hedge are both determined by the function $\wt w$. Its Feynman--Kac representation in turn allows to identify the main drivers of the hedging problem with small uncertainty aversion:

\begin{proposition}[Feynman--Kac representation]
\label{prop:Feynman-Kac}
Suppose that Assumption \ref{ass:second order} holds. Moreover, assume that for each $t\in[0,T]$, $\bar\sigma(t,\cdot)$ is uniformly continuous on $(K^{-1},K)$. Then for each $t \in [0,T]$, $s \in [K^{-1},K]$, and $y \in (y_l, y_u)$,
\begin{align}
\label{eqn:Feynman-Kac}
\wt w(t,s,y)
&= \EX[t,s]{\int_t^T \frac{\left(\bar\sigma(u,S_u)S_u^2 \bar V_{ss}(u,S_u)\right)^2}{2f''(u,S_u,y;\bar\sigma(u,S_u))} \dd u}
\end{align}
where the expectation is computed under a measure such that $S$ has the reference dynamics $\diff S_u = S_u \bar\sigma(u,S_u)\dd W_u$ and starts in $S_t = s$.
\end{proposition}

We postpone the proof of Proposition~\ref{prop:Feynman-Kac} to Section~\ref{sec:proof of Feynman-Kac}. From the Feynman--Kac representation \eqref{eqn:Feynman-Kac} we see that if $f''$ is constant, $\wt w$ measures the expected volatility-weighted cash gamma accumulated over the remaining lifetime of the option. In particular, $\wt w$ is high whenever it is likely that the stock price will spend a significant amount of time in zones where the option's cash gamma is high. The volatility weighting means that the cash gamma accrues in ``business time'' $\bar{\sigma}(t,S_t)^2 \dd t$, i.e., its effect is amplified in turbulent markets.

\paragraph{Discussion of the almost optimal strategy.}
The almost optimal strategy $\theta^\psi$ from Theorem~\ref{thm:second order} has the form
\begin{align*}
\theta^\psi_t
&= \bar \Delta_t + \left(\wt w_s(t,S_t,Y_t) + \frac{U'(Y_t)}{U''(Y_t)}\wt w_{sy}(t,S_t,Y_t) \right)\1_{\lbrace t < \tau\rbrace}\psi.
\end{align*}
The first summand is simply the delta hedge in the reference model, whereas the second term is a \emph{strategy adjustment} accounting for volatility uncertainty. The strategy adjustment is only active as long as the P\&L process does not deviate too much from its initial value. This is a technical modification that ensures that your P\&L stays within a bounded interval under any volatility scenario.\footnote{Mutatis mutandis, the threshold $1$ in the definition of the stopping time $\tau$ can be replaced by any other constant. The same modification also appears in the asymptotic analysis of models with transaction costs \cite{KallsenLi2013}.} To understand the substantial features of the strategy adjustment, let us first focus on the case $\wt w_{sy} = 0$.\footnote{A sufficient condition for $\wt w_{sy} = 0$ is that $f''(t,s,y;\varsigma)$ does not depend on $y$.} Then the strategy adjustment does not depend on your current P\&L, nor on the shape of your utility function, and is simply $\wt w_s(t,S_t)\psi$. Recall that a high (low) $\wt w$ is associated with a big (small) loss in value due to volatility uncertainty. Suppose that at time $t$, $\wt w_s(t,S_t)$ is positive. Then an increase (decrease) in the stock price makes your position more (less) vulnerable to volatility misspecification. But as the almost optimal strategy holds $\wt w_s(t,S_t)\psi$ more shares than the pure delta hedge, the higher vulnerability to volatility misspecification in the case of a rising stock price is compensated by an extra profit. Conversely, in the case of a falling stock price, the strategy adjustment leads to a loss compared to the pure delta hedge, which is compensated by the fact that the stock price has moved towards a zone of less vulnerability to volatility misspecification. To summarise, the almost optimal strategy hedges against movements of the stock price into zones of high vulnerability to volatility misspecification. Recalling that the magnitude of $\wt w$ is mainly determined by the option's cash gamma, the almost optimal strategy can thus be seen as a \emph{hedge against movements of the stock price into zones of high cash gamma}.

Now, let us consider the case $\wt w_{sy} \neq 0$. Suppose that $f(t,s,y;\varsigma) = g(t,s;\varsigma)/h(y)$ for suitable functions $g$ and $h>0$. The previous paragraph corresponds to a constant $h$. Now, $\wt w$ has the Feynman--Kac representation
\begin{align*}
\wt w(t,s,y)
&= h(y)\EX[t,s]{\int_t^T \frac{\left(\bar\sigma(u,S_u)S_u^2 \bar V_{ss}(u,S_u)\right)^2}{2g''(u,S_u;\bar\sigma(u,S_u))} \dd u}.
\end{align*}
We see that $h(y)$ acts as an amplifier of your uncertainty aversion. A decreasing (increasing) $h$ can be interpreted as \emph{decreasing (increasing) uncertainty aversion} (with respect to your P\&L), similar to the well-known notions of increasing or decreasing (absolute) risk aversion. Suppose that $h$ is strictly decreasing. If $\wt w_s(t,S_t,Y_t) > 0$, then $\wt w_{sy}(t,S_t,Y_t) < 0$ and the strategy adjustment
\begin{align*}
\left(\wt w_s(t,S_t,Y_t) + \frac{U'(Y_t)}{U''(Y_t)}\wt w_{sy}(t,S_t,Y_t) \right)\1_{\lbrace t < \tau\rbrace}\psi
\end{align*}
is larger than in the case of $\wt w_{sy} = 0$. As above, this leads to a profit if the stock rises into zones of high $\wt w$ and to a loss if the stock falls into zones of small $\wt w$. Anticipating a decrease in your P\&L and therefore an increase in your uncertainty aversion, you want to avoid zones of high vulnerability to volatility misspecification even more than if $h$ was a constant. Therefore, your strategy adjustment is even larger than in the case of constant $h$.

\paragraph{Indifference prices.}
An \emph{indifference bid price} (\emph{indifference ask price}) (for an option) is a price at which you are indifferent between keeping your current position and buying (selling) the option for that price. We emphasise that indifference prices typically depend on your wealth, the stock price, and on your current asset allocation. We suppose here that you are currently \emph{flat}, i.e., your current position in options is zero. Then you might demand a premium for adding risk to your portfolio by selling the option short. Throughout, $x_0$ denotes your initial capital (before buying or selling any options) and $\bar V = \bar V(0,s_0)$ is the initial reference value of the option. To ease notation, we write $\wt w = \wt w(0,s_0,x_0)$ and let $v(y_0;\psi)$ denote the value of our hedging problem corresponding to initial P\&L $y_0$.

If you are flat and decide to sell the option for a price $p_a(\psi)$, then your initial P\&L is ${x_0 + p_a(\psi) - \bar V}$. Therefore, the equation determining the indifference price $p_a(\psi)$ reads as follows:
\begin{align*}
U(x_0)
&= v(x_0+p_a(\psi)-\bar V;\psi).
\end{align*}
Using the expansion of $v$ from Theorem~\ref{thm:second order}, straightforward computations yield\footnote{A second-order expansion for the ask price can also be obtained, but does not offer much additional insight.}
\begin{align}
\label{eqn:price:flat ask}
p_a(\psi)
&= \bar V + \wt w \psi + o(\psi).
\end{align}
If you are flat, your ask price hence exceeds the reference value $\bar V$ of the option by a premium $\wt w \psi + o(\psi)$. We therefore call $\wt w$ the \emph{cash equivalent of small uncertainty aversion}; at the leading order, it is the (normalised) premium (or discount if you are the buyer, see the next paragraph) over the reference value that you demand for assuming a position that is vulnerable to volatility misspecification.

If you are flat and decide to buy the option for a price $p_b(\psi)$, your initial P\&L is $x_0 - p_b(\psi) + \bar V$. Buying the option is the same as selling the negative of the option. Moreover, the cash equivalents of the option and its negative coincide because the cash gamma enters the source term \eqref{eqn:source term:first order} of the PDE \eqref{eqn:PDE:first order} for $\wt w$ as a square.\footnote{This symmetry generally breaks down for the second-order term $\wh w$; cf.~the corresponding source term \eqref{eqn:source term:second order}. Hence, for a second-order expansion of the indifference bid price, one has to use the $\wh w$ corresponding to the negative of the option.} Thus, $p_b(\psi)$ is determined by
\begin{align*}
U(x_0)
&= v(x_0 - p_b(\psi) + \bar V;\psi),
\end{align*}
which yields
\begin{align}
\label{eqn:price:flat bid}
p_b(\psi)
&= \bar V - \wt w \psi + o(\psi).
\end{align}
In analogy to the ask price in \eqref{eqn:price:flat ask}, you demand a discount $\wt w\psi + o(\psi)$ on the reference value $\bar V$ to buy the option. Comparing \eqref{eqn:price:flat bid} with \eqref{eqn:price:flat ask}, we see that starting from a flat position, your bid-ask spread due to uncertainty aversion is $2\wt w \psi + o(\psi)$.

\begin{remark}
\label{rem:BSB}
The above results can be formally linked to the uncertain volatility model as follows. Consider the Black--Scholes--Barenblatt equation
\begin{align}
\label{eqn:rem:BSB}
\begin{split}
V^\psi_t(t,s) + \sup_{\varsigma \in [\lambda(t,s),\Lambda(t,s)]} \frac{1}{2}(\bar\sigma(t,s) + \psi \varsigma)^2 s^2V^\psi_{ss}
&= 0,\\
V^\psi(T,s)
&= G(s),
\end{split}
\end{align}
where $\psi > 0$ is a (small) parameter and $\lambda \leq 0 \leq \Lambda$ are suitable functions. This equation corresponds to the problem of finding the smallest initial capital that allows to superreplicate the option $G(S_T)$ for any volatility process evolving in the random interval
\begin{align*}
[\bar\sigma(t,S_t) + \psi \lambda(t,S_t), \bar\sigma(t,S_t) + \psi \Lambda(t,S_t)];
\end{align*}
see \cite{Lyons1995}.\footnote{This specification is a special case of the general ``random $G$-expectation'' \cite{Nutz2013}.} As Lyons \cite[Section 5]{Lyons1995} points out, the solution $V^\psi$ to the Black--Scholes--Barenblatt PDE \eqref{eqn:rem:BSB} has the formal asymptotic expansion
\begin{align}
\label{eqn:rem:BSB:expansion}
V^\psi(t,s)
&= \bar V(t,s) + \wt V(t,s)\psi + o(\psi) \quad (\psi \downarrow 0),
\end{align}
where $\bar V$ is the solution to the Black--Scholes PDE \eqref{eqn:hedging:PDE} and $\wt V$ solves the following linear parabolic PDE with source term:\footnote{\cite{Lyons1995} imposes bounds on the instantaneous variance of prices instead of the volatility of returns. Hence, the PDE for $\wt V$ there looks slightly different. The PDE presented here is a slight generalisation of the one derived in \cite{FouqueRen2014}.}
\begin{align}
\label{eqn:rem:BSB:PDE:first order}
\begin{split}
\wt V_t(t,s) + \frac{1}{2}\bar\sigma(t,s)^2 s^2 \wt V_{ss}(t,s) + \sup_{\varsigma \in [\lambda(t,s),\Lambda(t,s)]} \varsigma \bar\sigma(t,s) s^2 \bar V_{ss}(t,s)
&= 0,\\
\wt V(T,s)
&= 0.
\end{split}
\end{align}
The expansion \eqref{eqn:rem:BSB:expansion} has been proved by Fouque and Ren \cite{FouqueRen2014} for the special case where $\bar\sigma$ is a constant, $\Lambda \equiv 1$, and $\lambda \equiv 0$.

In view of \eqref{eqn:rem:BSB:expansion}, $\wt V$ can be interpreted as the (normalised) leading-order premium over the reference value that an agent with infinite risk aversion but small volatility band demands as a compensation for assuming a position that is vulnerable to volatility misspecification. This is analogous to the interpretation of $\wt w$ as the cash equivalent of small uncertainty aversion. More specifically, if $f''$ does not depend on $y$ and if we choose
\begin{align*}
\Lambda(t,s) 
&= -\lambda(t,s)
= \frac{1}{2}\wt\sigma(t,s)\sgn(\bar V_{ss}(t,s))
= \frac{\bar\sigma(t,s) s^2 \left\vert\bar V_{ss}(t,s)\right\vert}{2 f''(t,s;\bar\sigma(t,s))},
\end{align*}
then the PDE \eqref{eqn:rem:BSB:PDE:first order} for $\wt V$ reduces to the PDE \eqref{eqn:PDE:first order} for $\wt w$. Moreover, in this case, our almost optimal strategy $\theta^\psi_t = \bar V_s(t,S_t) + \wt w_s(t,s)\psi$ coincides with the delta hedge corresponding to the expansion \eqref{eqn:rem:BSB:expansion}. Thus, our results are formally equivalent to those obtained for an infinitely risk averse agent who, however, uses a random and time-dependent volatility band that depends on both the option (through its cash gamma) and her uncertainty aversion (through $f''$).
\end{remark}

\subsection{Existence of probability scenarios}
\label{sec:second result}

Our main result, Theorem~\ref{thm:second order}, is formulated for abstract sets of trading strategies $\cA$ and volatility scenarios $\cV$, which are (i) large enough to contain the almost optimal strategies $\theta^\psi$ and volatilities $\sigma^\psi$ and (ii) small enough to ensure that a corresponding measure  $P \in \fP(\theta,\sigma)$ exists and all the technical prerequisites for our verification theorem are satisfied. In this section, we propose concrete choices $\cA_0$ and $\cV_0$ that meet these requirements. 

To this end, let $\cA_0$ denote the set of all real-valued processes $\theta = (\theta_t)_{t\in[0,T]}$ of the form
\begin{align}
\label{eqn:A0}
\theta_t
&= \bar V_s(t,S_t) + \breve\theta_t \1_{\lbrace t < \tau \rbrace},
\end{align}
where $\breve\theta$ is a bounded, progressively measurable process such that $\breve\theta_t:\Omega \to \RR$ is continuous for each $t\in[0,T]$, and the stopping time $\tau$ is defined as in Theorem~\ref{thm:second order}. In words, every strategy in $\cA_0$ has to fall back to the reference delta hedge as soon as the corresponding P\&L process $Y$ deviates too far from its initial value. Together with Assumption~\ref{ass:second order}, this guarantees (cf.~Theorem~\ref{thm:existence}) that the P\&L process associated to these strategies fulfils our uniform boundedness assumption \eqref{eqn:process bounds}. Prior to the switch, however, the strategy can deviate from the reference hedge in any possibly path-dependent way. 

Next, denote by $\cV_0$ the set of all real-valued processes $\sigma = (\sigma_t)_{t\in[0,T]}$ of the form
\begin{align}
\label{eqn:V0}
\sigma_t
&= \acute\sigma(t,S_t,Y_t) \1_{\lbrace S_t \in (K^{-1},K) \rbrace},
\end{align}
where $\acute\sigma: [0,T]\times\RR^2 \to [0,K]$ is a (globally) Lipschitz continuous function. Hence, for any volatility process in $\cV_0$, once $S$ leaves $(K^{-1},K)$, its volatility vanishes and the stock price freezes.

The following theorem shows that under suitable regularity assumptions, $\cA_0$ and $\cV_0$ satisfy the assumption in Theorem~\ref{thm:second order}:
\begin{align*}
(\theta^\psi, \sigma^\psi)
&\in \cA_0 \times \cV_0 \subset \cZ \quad \text{for }\psi > 0\text{ small enough.}
\end{align*}
Let us emphasise, however, that $\theta^\psi$ and $\sigma^\psi$ remain almost optimal also in \emph{any} larger classes $\cA \supset \cA_0$ and $\cV \supset \cV_0$ of trading strategies and volatility scenarios satisfying $\cA \times \cV \subset \cZ$.

\begin{theorem}
\label{thm:existence}
Suppose that Assumption~\ref{ass:second order} holds, that $\bar V_s$, $\bar V_{ss}$ and $\bar\sigma$ are Lipschitz continuous on $(0,T)\times(K^{-1},K)$, and that $y_l < y_0-2-\frac{1}{2}K^3 T$ and $y_u > y_0 + 2 + \frac{1}{2}K^3T$. Then $\cA_0 \times \cV_0 \subset \cZ$, i.e., for any $(\theta,\sigma) \in \cA_0\times\cV_0$, there is a probability measure $P$ on $(\Omega, \cF)$ such that $(\theta,\sigma,P) \in \fS$. If in addition, $f''$ is Lipschitz continuous on $[0,T]\times[K^{-1},K]\times(y_l, y_u)\times[0,K]$, then $(\theta^\psi, \sigma^\psi) \in \cA_0 \times \cV_0$ for all $\psi>0$ small enough.
\end{theorem}

The proof of Theorem~\ref{thm:existence} is provided in Section~\ref{sec:existence}.

\section{Heuristics and extensions}
\label{sec:extensions}

In practice, you will rarely need to hedge a single option, but rather a whole book of different contracts. For instance, your book may consist of call and put options at various strikes and maturities and miscellaneous exotic options like Asian options, barrier options, lookback options, or options on the realised variance of the stock returns. Moreover, some options like calls and puts at standardised maturities and with strikes close to the current price of the underlying might be liquidly traded and thus be available as additional hedging instruments for the non-traded options in your book. 

Such more involved problems can also be tackled with the methodology of this article. Since the corresponding rigorous verification would become even more technical, we only develop these extensions on a heuristic level here. The resulting formulas explain how the sensitivities of an option portfolio affect the cash equivalent of small uncertainty aversion for practically relevant cases. We start with an informal derivation of Theorem~\ref{thm:second order}, thereby providing some intuition for the rigorous proof in Section~\ref{sec:proof of main result}. Subsequently, we adapt the general procedure to more complicated settings including exotic options of Asian-, lookback-, or barrier-type, as well as options on the realised variance. We also explain how to deal with option portfolios as well as with static hedging using vanilla options, and discuss how this allows to interpret our ``cash equivalent'' as a ``measure of model uncertainty'' in the sense of Cont~\cite{Cont2006}.

\subsection{General procedure and the case of Theorem~\ref{thm:second order}}
\label{sec:general}

Our starting point is the dynamic programming approach to two-player, zero-sum stochastic differential games. The central idea is to find an asymptotic solution to the Hamilton--Jacobi--Bellman--Isaacs (henceforth HJBI) equation associated to the hedging problem and corresponding almost optimal controls. For the convenience of the reader, we provide a derivation of the HJBI equation which starts from a general sufficient criterion for optimality  (Proposition~\ref{prop:optimality principle}) known as \emph{principle of optimality} \cite{Davis1979} or \emph{martingale optimality principle} \cite[V.15]{RogersWilliams2000}; see also \cite[Proposition 4.1]{Seifried2010} for a version of the martingale optimality principle in the context of a zero-sum game. After that, we explain how the HJBI equation together with an appropriate ansatz can be used to derive candidates for the value and the almost optimal controls of our hedging problems. Next, the formulas provided in Section~\ref{sec:main results} for the specific case of hedging a single vanilla option are recovered. Finally, we summarise the general procedure that is used in the remaining subsections to derive the corresponding candidates also in more general settings.

\paragraph{A sufficient criterion for optimality.}
Consider an optimisation problem of the form
\begin{align}
\label{eqn:general:value}
v
&:= \sup_{\theta \in \cA} \inf_{\sigma\in\cV}\EX[\theta,\sigma]{N_T},
\end{align}
where $\cA$ and $\cV$ are sets of admissible controls, for each $(\theta,\sigma) \in \cA\times\cV$, $\EX[\theta,\sigma]{\cdot}$ denotes the expectation under a given measure $P^{\theta,\sigma}$, and $N$ is a sufficiently integrable process.\footnote{E.g., if $N$ is integrable under $P^{\theta,\sigma}$ for each $(\theta,\sigma) \in \cA\times\cV$, then $v$ is a well-defined number in the extended real line $[-\infty,+\infty]$.}
\pagebreak
\begin{proposition}[Martingale optimality principle]
\label{prop:optimality principle}
If there is a pair $(\theta^*,\sigma^*) \in \cA\times\cV$ such that
\renewcommand{\labelenumi}{(\roman{enumi})}
\begin{enumerate}
\item for each $\theta \in \cA$, $N$ is a supermartingale under $P^{\theta,\sigma^*}$,
\item for each $\sigma \in \cV$, $N$ is a submartingale under $P^{\theta^*,\sigma}$,
\end{enumerate}
then
\begin{align}
\label{eqn:prop:optimality principle}
N_0 
&= \EX[\theta^*,\sigma^*]{N_T} 
= \sup_{\theta\in\cA} \inf_{\sigma\in\cV} \EX[\theta,\sigma]{N_T}
= \inf_{\sigma\in\cV} \sup_{\theta\in\cA} \EX[\theta,\sigma]{N_T}.
\end{align}
In particular, $\theta^*$ and $\sigma^*$ are optimal controls for \eqref{eqn:general:value} with value $v = N_0$.
\end{proposition}

\begin{proof}
As $N$ is a martingale under $P^{\theta^*,\sigma^*}$ by (i) and (ii), the first equality in \eqref{eqn:prop:optimality principle} is clear. Moreover, (i) and (ii) imply
\begin{align*}
\sup_{\theta\in\cA} \EX[\theta,\sigma^*]{N_T}
&\leq N_0
\leq \inf_{\sigma\in\cV} \EX[\theta^*,\sigma]{N_T}.
\end{align*}
Thus,
\begin{align*}
N_0
&\leq \inf_{\sigma\in\cV} \EX[\theta^*,\sigma]{N_T}
\leq \sup_{\theta\in\cA} \inf_{\sigma\in\cV} \EX[\theta,\sigma]{N_T}
\leq \inf_{\sigma\in\cV} \sup_{\theta\in\cA} \EX[\theta,\sigma]{N_T}
\leq \sup_{\theta\in\cA} \EX[\theta,\sigma^*]{N_T}
\leq N_0,
\end{align*}
and the remaining equalities in \eqref{eqn:prop:optimality principle} follow.
\end{proof}

\paragraph{Derivation of the Hamilton--Jacobi--Bellman--Isaacs equation.}
Suppose that we are given controls $\theta^*$ and $\sigma^*$ satisfying the conditions of Proposition~\ref{prop:optimality principle}. In addition, assume that the dynamics of $N$ under $P^{\theta,\sigma}$ are of the form
\begin{align*}
\diff N_t
&= \diff M^{\theta,\sigma}_t + \nu^{\theta,\sigma}_t \dd t
\end{align*}
for some $P^{\theta,\sigma}$-martingale $M^{\theta,\sigma}$ and a drift rate $\nu^{\theta,\sigma}$. Then the conditions (i) and (ii) of Proposition~\ref{prop:optimality principle} imply that the drift rates $\nu^{\theta,\sigma}$ satisfy
\begin{align*}
\sup_{\theta\in\cA} \nu_t^{\theta,\sigma^*}
&\leq 0
\leq \inf_{\sigma\in\cV} \nu_t^{\theta^*,\sigma}.
\end{align*}
Using these inequalities, we find
\begin{align*}
0
&\leq \inf_\sigma \nu_t^{\theta^*,\sigma}
\leq \nu_t^{\theta^*,\sigma^*}
\leq \sup_\theta \nu_t^{\theta,\sigma^*}
\leq 0,
\end{align*}
as well as 
\begin{align*}
0
&\leq \inf_\sigma \nu_t^{\theta^*,\sigma}
\leq \sup_\theta \inf_\sigma \nu_t^{\theta,\sigma}
\leq \inf_\sigma \sup_\theta \nu_t^{\theta,\sigma}
\leq \sup_\theta \nu_t^{\theta,\sigma^*}
\leq 0.
\end{align*}
Hence, we have equality everywhere and, in particular:
\begin{align}
\label{eqn:general:saddle point}
\nu_t^{\theta^*,\sigma^*}
&= \sup_\theta \inf_\sigma \nu_t^{\theta,\sigma}
= \inf_\sigma \sup_\theta \nu_t^{\theta,\sigma} = 0.
\end{align}
To obtain the drift rates more explicitly, let us now assume that there are sufficiently regular functions $g$ and $w$, and a state process $Z$ that is a (multidimensional) It\^o diffusion under each $P^{\theta,\sigma}$, such that for any $(\theta,\sigma) \in \cA\times\cV$,
\begin{align}
\label{eqn:general:N}
N_t
&= \int_0^t g(u,Z_u;\theta_u,\sigma_u) \dd u + w(t,Z_t) \quad P^{\theta,\sigma}\text{-a.s.}
\end{align}
This form is motivated by the structure of our hedging problem, cf.~\eqref{eqn:asymptotic:value} below. More concretely, if $Z$ has dynamics
\begin{align*}
\diff Z_t
&= b(t,Z_t;\theta_t,\sigma_t) \dd t + a(t,Z_t;\theta_t,\sigma_t) \dd W_t
\end{align*}
under $P^{\theta,\sigma}$, then applying It\^o's formula under each $P^{\theta,\sigma}$ to the right-hand side of \eqref{eqn:general:N} yields 
\begin{align}
\label{eqn:general:drift representation}
\nu^{\theta,\sigma}_t
&= w_t(t,Z_t) + H(t,Z_t,\nabla w(t,Z_t),\D^2 w(t,Z_t);\theta_t,\sigma_t)
\end{align}
where $\nabla w$ and $\D^2 w$ denote the gradient and the Hessian of $w(t,z)$ with respect to the $z$ variable, respectively, and
\begin{align*}
H(t,z,p,A;\vartheta,\varsigma)
&:= g(t,z;\vartheta,\varsigma) + b(t,z;\vartheta,\varsigma) \cdot p + \frac{1}{2} \Trace\left((aa^\tr)(t,z;\vartheta,\varsigma) A\right).
\end{align*}
Substituting \eqref{eqn:general:drift representation} into \eqref{eqn:general:saddle point} yields\footnote{We would obtain the same results if we interchanged the order of the infimum and the supremum. In the language of two-player, zero-sum stochastic differential games, this indicates that the game ``has a value''.}
\begin{align*}
w_t(t,Z_t) + \sup_{\theta} \inf_{\sigma} H(t,Z_t,\nabla w(t,Z_t), \D^2 w(t,Z_t);\theta_t,\sigma_t)
&= 0.
\end{align*}
A sufficient condition for this to hold is that $w$ satisfies the PDE
\begin{align}
\label{eqn:general:HJBI}
w_t + \sup_\vartheta \inf_\varsigma H(t,z,\nabla w, \D^2 w;\vartheta,\varsigma)
&= 0,
\end{align}
which is called the Hamilton--Jacobi--Bellman--Isaacs equation. In addition, candidates for the optimal controls are given by $\vartheta^*(t,Z_t)$ and $\varsigma^*(t,Z_t)$, where $\vartheta^*$ and $\varsigma^*$ are the saddle points of the HJBI equation \eqref{eqn:general:HJBI} in the sense that, for each $(t,z)$:
\begin{align*}
H(t,z,\nabla w(t,z), \D^2 w(t,z);\vartheta^*(t,z),\varsigma^*(t,z))
&= \sup_\vartheta \inf_\varsigma H(t,z,\nabla w(t,z), \D^2 w(t,z);\vartheta,\varsigma).
\end{align*}
Also note that $N_0 = w(0,Z_0)$ from \eqref{eqn:general:N}, so that $w(0,Z_0)$ is a candidate for the value of the optimisation problem \eqref{eqn:general:value} by Proposition~\ref{eqn:prop:optimality principle}. Therefore, the function $w$ is also called the \emph{value function} of the optimisation problem \eqref{eqn:general:value}.

The above derivation of the HJBI equation is of course merely formal. To prove rigorously that the candidate controls are indeed optimal, a rigorous verification theorem (using for instance the sufficient conditions from Proposition~\ref{prop:optimality principle}) is needed.

\paragraph{Ansatz for asymptotic solution.}
For each $\psi > 0$, consider the hedging problem\footnote{For the heuristic derivation in this section, we tacitly assume that for each $(\theta,\sigma)$, $P^{\theta,\sigma}$ attains the infimum in \eqref{eqn:value}, so that the additional infimum over measures in \eqref{eqn:value} disappears.}
\begin{align}
\label{eqn:asymptotic:value}
v(\psi)
&:= \sup_{\theta\in\cA} \inf_{\sigma\in\cV} \EX[\theta,\sigma]{\frac{1}{\psi}\int_0^T U'(Y_t)f(t,S_t,Y_t;\sigma_t) \dd t + U(Y_T)},
\end{align}
where under each $P^{\theta,\sigma}$, $S$ and $Y$ have dynamics of the form
\begin{align}
\dd S_t
&= S_t \sigma_t \dd W_t,\notag\\
\label{eqn:asymptotic:Y}
\dd Y_t
&=  (\theta_t - \bar\Delta(t, S_t)) \dd S_t +  b(t,S_t;\sigma_t) \dd t,
\end{align}
for some functions $\bar\Delta$ and $b$. The drift rate $b$ is further required to satisfy $b(t,s;\bar\sigma(t,s)) = 0$ as the theoretical P\&L should be ``locally drift-less'' whenever the true volatility $\sigma_t$ coincides with the reference volatility $\bar\sigma(t,S_t)$.\footnote{This covers most of the specific choices that are dealt with in the following subsections, except for additional state variables needed for some exotic options in Section~\ref{sec:exotic}. To explain the general procedure, we first focus here on the easiest case with just two state variables, the stock price $S$ and the P\&L process $Y$.}

To derive the HJBI equation for \eqref{eqn:asymptotic:value}, we first recast the expression inside the expectation of \eqref{eqn:asymptotic:value} into the form \eqref{eqn:general:N}. To this end, let $w^\psi(t,s,y)$ be a function satisfying the terminal condition $w^\psi(T,s,y) = U(y)$. Then the HJBI equation corresponding to \eqref{eqn:asymptotic:value} reads as follows:
\begin{align}
\label{eqn:asymptotic:HJBI}
w^\psi_t + \sup_\vartheta \inf_\varsigma \left\lbrace \frac{1}{\psi}U' f(\varsigma) + b(\varsigma)w^\psi_y + \frac{1}{2}\varsigma^2 s^2 \left(w^\psi_{ss} + 2(\vartheta-\bar\Delta)w^\psi_{sy} + (\vartheta-\bar\Delta)^2 w^\psi_{yy}\right)\right\rbrace
&= 0.
\end{align}
(We suppress the arguments $(t,s,y)$ to ease notation.) Explicit solutions to this equation are typically not available. Therefore, our goal is to obtain an asymptotic solution for small values of the uncertainty aversion parameter $\psi$. More precisely, we want to find strategies $\theta^\psi$, volatilities $\sigma^\psi$, and correction terms $\wt w$ and $\wh w$ such that
\begin{align*}
v(\psi)
&= U(y_0) - U'(y_0) \wt w(0,s_0,y_0)\psi + U'(y_0) \wh w(0,s_0,y_0) \psi^2 + o(\psi^2)\\
&= \EX[\theta^\psi,\sigma^\psi]{\frac{1}{\psi}\int_0^T U'(Y_t)f(t,S_t,Y_t;\sigma^\psi_t) \dd t + U(Y_T)} + o(\psi^2).
\end{align*}
The first equality is a second-order expansion of the value as a function of the uncertainty aversion parameter $\psi$. The second equality says that the strategy $\theta^\psi$ and the volatility $\sigma^\psi$ are optimal controls at the next-to-leading order $O(\psi^2)$.

We next use the HJBI equation and an appropriate ansatz to derive candidates for the asymptotic expansion of the value function and the almost optimal controls.
Recall from Section~\ref{sec:hedging} that vanilla options can be perfectly replicated in the reference local volatility model. Hence, without model uncertainty, it is optimal to use the delta hedge for the option, and the value of the hedging problem is simply the utility from the initial P\&L; cf.~Lemma~\ref{lem:hedging}. This motivates the following ansatz for the asymptotic expansion of the value function and the almost optimal controls of \eqref{eqn:asymptotic:value}:
\begin{align}
\label{eqn:asymptotic:ansatz value}
w^\psi(t,s,y)
&= U(y) - U'(y)\wt w(t,s,y)\psi,\\
\label{eqn:asymptotic:ansatz volatility}
\sigma^\psi(t,s,y)
&= \bar\sigma(t,s) + \wt\sigma(t,s,y)\psi,\\
\label{eqn:asymptotic:ansatz strategy}
\theta^\psi(t,s,y)
&= \bar\Delta(t,s) + \wt\theta(t,s,y)\psi,
\end{align}
for functions $\wt w, \wt\sigma, \wt\theta$ to be determined. With this ansatz, the HJBI equation yields equations for $\wt w, \wt\sigma,\wt\theta$. Indeed, substituting \eqref{eqn:asymptotic:ansatz value}--\eqref{eqn:asymptotic:ansatz strategy} into \eqref{eqn:asymptotic:HJBI} (we assume that $\theta^\psi$ and $\sigma^\psi$ are point-wise saddle points for the $\sup_\vartheta \inf_\varsigma$), using also the Taylor expansions $f(\varsigma) \approx \frac{1}{2}f''(\bar\sigma)(\varsigma-\bar\sigma)^2$ (because of \eqref{eqn:penalty function:minimum conditions}) and $b(\varsigma) \approx b'(\bar\sigma)(\varsigma-\bar\sigma)$ (because of the assumption that $b(\bar\sigma) = 0$; $b'$ denotes the partial derivative with respect to $\varsigma$), and ordering by powers of $\psi$, we obtain
\begin{align}
\label{eqn:asymptotic:HJBI expansion}
-U' \times \left(\wt w_t - \frac{1}{2}f''(\bar\sigma)\wt\sigma^2 -b'(\bar\sigma)\wt\sigma + \frac{1}{2}\bar\sigma^2 s^2 \wt w_{ss} \right)\psi + o(\psi)
&= 0.
\end{align}
Our candidate for $\wt\sigma$ is the minimiser of the $O(\psi)$ term. Solving the first-order condition $f''(\bar\sigma)\wt\sigma + b'(\bar\sigma) = 0$ for $\wt\sigma$ yields
\begin{align}
\label{eqn:asymptotic:sigma tilde}
\wt\sigma 
&= \frac{-b'(\bar\sigma)}{f''(\bar\sigma)}.
\end{align}
Plugging this candidate back into the $O(\psi)$ term in \eqref{eqn:asymptotic:HJBI expansion} and setting the result equal to zero yields a PDE for the cash equivalent $\wt w$:
\begin{align}
\label{eqn:asymptotic:PDE:first order}
\wt w_t + \frac{1}{2}\bar\sigma^2 s^2 \wt w_{ss} + \frac{b'(\bar\sigma)^2}{2f''(\bar\sigma)}
&= 0.
\end{align}
To find the candidate for the optimal strategy, in view of the HJBI equation \eqref{eqn:asymptotic:HJBI}, we only need to maximise 
\begin{align}
\label{eqn:asymptotic:second order optimisation}
2(\vartheta-\Delta)w^\psi_{sy} + (\vartheta-\Delta)^2 w^\psi_{yy}.
\end{align}
Substituting the ansatz \eqref{eqn:asymptotic:ansatz value} and \eqref{eqn:asymptotic:ansatz strategy} into \eqref{eqn:asymptotic:second order optimisation}, we find 
\begin{align*}
\left(2\wt\theta\frac{\partial^2}{\partial s\partial y}\left(-U'\wt w\right) + {\wt\theta}^2 U''\right)\psi^2 + o(\psi^2).
\end{align*}
The $O(\psi^2)$ term simplifies to $-2\wt\theta\left(U' \wt w_{sy} + U'' \wt w_s\right) + {\wt\theta}^2 U''$ and the maximiser of this quadratic equation in $\wt\theta$ is
\begin{align}
\label{eqn:asymptotic:theta tilde}
\wt\theta
&= \wt w_s + \frac{U'}{U''} \wt w_{sy}.
\end{align}
Finally, a PDE for the second-order correction term $\wh w$ can be obtained in the same way, using a second-order ansatz $w^\psi(t,s,y) = U(y) - U'(y) \wt w(t,s,y)\psi + U'(y) \wh w(t,s,y) \psi^2$ for the value function instead of the first-order ansatz \eqref{eqn:asymptotic:ansatz value} (the ansatz for the controls remains the same) and setting the $O(\psi^2)$ term in the expanded HJBI equation, evaluated at the candidate strategy and candidate volatility, equal to zero. We omit the lengthy computations.

\paragraph{The case of Theorem~\ref{thm:second order}.}
If a single vanilla option needs to be hedged, the P\&L dynamics have been derived in Section~\ref{sec:hedging}. Comparing the dynamics of $Y$ in \eqref{eqn:setup:Y} with the form \eqref{eqn:asymptotic:Y} considered above, we find that 
\begin{align}
\label{eqn:asymptotic:b for single option}
b(t,s;\varsigma)
&= \frac{1}{2}s^2 \bar V_{ss}(t,s)(\bar\sigma(t,s)^2 - \varsigma^2).
\end{align}
Thus, $b'(t,s;\bar\sigma(t,s)) = -s^2 \bar V_{ss}(t,s) \bar\sigma(t,s)$. Plugging this formula for $b'$ into the formulas for $\wt\sigma$ in \eqref{eqn:asymptotic:sigma tilde} and into the PDE for $\wt w$ in \eqref{eqn:asymptotic:PDE:first order}, we recover exactly \eqref{eqn:sigma tilde} and \eqref{eqn:PDE:first order}. The formula for $\wt\theta$ in \eqref{eqn:asymptotic:theta tilde} also coincides with \eqref{eqn:theta tilde}.

\paragraph{Summary.}
The general procedure to find the candidate controls and the cash equivalent can be summarised as follows.

\renewcommand{\labelenumi}{(\roman{enumi})}
\begin{enumerate}
\item Introduce as many state variables (e.g., the running maximum or minimum for lookback options, the running integral over the stock price for Asian options, etc.) as necessary to express the theoretical value of each option as a function of time and these state variables.

\item Write down the theoretical P\&L process $Y$ corresponding to your option position and your trading strategy. Use It\^o's formula to determine the drift and diffusion coefficients for $Y$.

\item Write down the HJBI equation corresponding to the control problem.

\item Plug an appropriate ansatz as in \eqref{eqn:asymptotic:ansatz value}--\eqref{eqn:asymptotic:ansatz volatility} for the value and the volatility into the HJBI equation and expand the result in the $\psi$ variable. Minimise the $O(\psi)$ term over $\wt\sigma$ to obtain the candidate volatility.

\item Substitute the candidate for $\wt\sigma$ back into the $O(\psi)$ term of the HJBI to find a PDE for the cash equivalent $\wt w$.

\item Plug the ansatz \eqref{eqn:asymptotic:ansatz strategy} for the strategy into the expanded HJBI equation. Maximise the $O(\psi^2)$ term over $\wt\theta$ to obtain the candidate strategy.
\end{enumerate}

In the following subsections, we illustrate this approach in a number of practically relevant applications.

\subsection{Some exotic options}
\label{sec:exotic}

In this section, we consider some exotic options whose payoffs depend on the whole path of the stock price and not only on the stock price at maturity. For instance, their payoffs could depend on the average of the stock price path over some time period (Asian options), its maximum or minimum (lookback options), or on the realised variance of the stock returns (e.g., a variance swap). The reference values of such options can still be represented as solutions to a PDE, provided one introduces suitable additional state variables that keep track of the path-dependent features of the contract.

Step (i) is to represent the theoretical value of the exotic option under consideration in the local volatility model as a solution to a PDE. To this end, let us introduce a generic additional state variable $A$ with dynamics of the form
\begin{align*}
\diff A_t
&=\alpha(t,S_t,A_t,M_t) \dd t + \beta(t,S_t,A_t,M_t) \dd\langle S \rangle_t + \gamma(t,S_t,A_t,M_t)\dd S_t + \delta(t,S_t,A_t,M_t)\dd M_t,
\end{align*}
where $M_t := \max_{u\in[0,t]} S_u$ is the running maximum of the stock price. Note that we do not merge the $\diff \langle S \rangle_t$-term with the $\diff t$-term as $\langle S \rangle$ depends on the true volatility $\sigma$. If you want to value and hedge an option with maturity $T$ and payoff of the form $G(S_T,A_T,M_T)$ in the local volatility model, you can essentially employ the same approach as in Section~\ref{sec:hedging}. Indeed, applying It\^o's formula to a sufficiently regular function $\bar V(t,s,a,m)$, and assuming that the true dynamics of $S$ are $\dd S_t = S_t \sigma_t \dd W_t$, we obtain, dropping most arguments:
\begin{align}
\label{eqn:exotic:ito}
\begin{split}
\diff \bar V(t,S_t,A_t,M_t)
&= (\bar V_s + \gamma \bar V_a) \dd S_t
+ (\bar V_m + \delta \bar V_a) \dd M_t\\
&\qquad+ \left\lbrace \bar V_t + (\alpha + \beta \sigma_t^2 S_t^2) \bar V_a + \frac{1}{2}\sigma_t^2 S_t^2\left(\bar V_{ss} + 2 \gamma \bar V_{sa} + \gamma^2 \bar V_{aa}\right)\right\rbrace \dd t.
\end{split}
\end{align}
Note that if $\sigma_t = \bar\sigma(t,S_t)$, i.e., if your local volatility model is correct, and if $\bar V$ solves the PDE
\begin{align}
\label{eqn:exotic:PDE}
\begin{split}
\bar V_t + (\alpha + \beta \bar\sigma^2 s^2) \bar V_a + \frac{1}{2}\bar\sigma^2 S_t^2\left(\bar V_{ss} + 2 \gamma \bar V_{sa} + \gamma^2 \bar V_{aa}\right)
&= 0,\\
\bar V_m + \delta \bar V_a 
&=0 \quad \text{whenever } s = m,\\
\bar V(T,s,a,m)
&= G(s,a,m),
\end{split}
\end{align}
then \eqref{eqn:exotic:ito} yields, using that $M_t$ only increases when $S_t = M_t$,
\begin{align*}
G(S_T,A_T,M_T)
&= \bar V(T,S_T,A_T,M_T)
= \bar V(0,S_0,A_0,M_0) + \int_0^T \bar\Delta(t,S_t,A_t,M_t) \dd S_t,
\end{align*}
where $\bar\Delta := \bar V_s + \gamma \bar V_a$.\footnote{Note that this hedge $\bar\Delta$ reflects the option's sensitivity to price moves in the underlying both directly through $S$ and indirectly through the additional state variable $A$.} Hence, in the local volatility model and under these assumptions, the option $G$ can be perfectly replicated by self-financing trading with initial capital $\bar V(0,S_0,A_0,M_0)$ and trading strategy $\bar\Delta$.

Step (ii) is to write down the P\&L process $Y$ and to determine its dynamics in terms of the trading strategy $\theta$ and the true volatility $\sigma$. In analogy to the case of a vanilla option treated in Section~\ref{sec:hedging}, your theoretical P\&L at time $t$ is 
\begin{align}
\label{eqn:exotic:Y definition}
Y_t
&= x_0 + \int_0^t \theta_u \dd S_u - \bar V(t,S_t,A_t,M_t).
\end{align}
Using the PDE \eqref{eqn:exotic:PDE} to substitute $\bar V_t$ in \eqref{eqn:exotic:ito} and plugging the result into \eqref{eqn:exotic:Y definition} yields
\begin{align*}
\diff Y_t 
&= \left(\theta_t - (\bar V_s + \gamma \bar V_a)\right)\dd S_t
+ S_t^2 \left(\beta \bar V_a + \frac{1}{2}(\bar V_{ss} + 2\gamma \bar V_{sa} + \gamma^2 \bar V_{aa})\right)\left(\bar\sigma^2 - \sigma_t^2 \right) \dd t.
\end{align*}
Taking into account $\bar\Delta = \bar V_s + \gamma \bar V_a$ and setting $b(\varsigma) := s^2 \left(\beta \bar V_a + \frac{1}{2}(\bar V_{ss} + 2\gamma \bar V_{sa} + \gamma^2 \bar V_{aa})\right)\left(\bar\sigma^2 - \varsigma^2 \right)$, this simplifies to
\begin{align*}
\diff Y_t 
&= \left(\theta_t - \bar\Delta\right)\dd S_t
+ b(\sigma_t) \dd t.
\end{align*}

Step (iii) is to write down the HJBI equation. Recall that this boils down to using It\^o's formula to derive the drift rate of the process $N_t := \frac{1}{\psi}\int_0^t U'(Y_t)f(t,S_t,Y_t;\sigma_t) \dd t + w^\psi(t,S_t,Y_t,A_t,M_t)$. We obtain
\begin{align}
\label{eqn:exotic:HJBI}
w^\psi_t +
&\sup_\vartheta \inf_\varsigma
\left\lbrace \frac{1}{\psi}U' f(\varsigma) + b(\varsigma)w^\psi_y + (\alpha + \beta s^2 \varsigma^2)w^\psi_a \right.\\
&\quad\left.+ \frac{1}{2}\varsigma^2 s^2\left(w^\psi_{ss} + 2(\vartheta - \bar\Delta) w^\psi_{sy} + (\vartheta-\bar\Delta)^2 w^\psi_{yy} + 2 \gamma w^\psi_{sa} + 2 \gamma (\vartheta-\bar\Delta) w^\psi_{ya} + \gamma^2 w^\psi_{aa}\right) \right\rbrace
= 0.\notag
\end{align}
In analogy to the valuation PDE \eqref{eqn:exotic:PDE} for the exotic option, we also need that $w^\psi_m + \delta w^\psi_a = 0$ whenever $s = m$, so that the $\diff M_t$-term of $N$ vanishes as well.

Step (iv) is to plug the ansatz
\begin{align*}
w^\psi(t,s,y,a,m)
&= U(y) - U'(y)\wt w(t,s,y,a,m)\psi,\\
\sigma^\psi(t,s,y,a,m)
&= \bar\sigma(t,s) + \wt\sigma(t,s,y,a,m)\psi,
\end{align*}
into the HJBI equation \eqref{eqn:exotic:HJBI}, expand the result in the $\psi$ variable and minimise the $O(\psi)$ term over $\wt\sigma$ to find the candidate for the volatility. The $O(\psi)$ term in the expansions reads as
\begin{align*}
-U' \times \left(\wt w_t - \frac{1}{2}f''(\bar\sigma)\wt\sigma^2 - b'(\bar\sigma)\wt\sigma + (\alpha + \beta s^2 \bar\sigma^2)\wt w_a + \frac{1}{2}\bar\sigma^2 s^2 \left(\wt w_{ss} + 2\gamma \wt w_{sa} + \gamma^2 \wt w_{aa} \right) \right) \psi.
\end{align*}
Minimising this over $\wt\sigma$ and using the definition of $b$, we find the candidate volatility:
\begin{align}
\label{eqn:exotic:candidate volatility}
\sigma^\psi
&= \bar\sigma + \frac{\bar\sigma}{f''(\bar\sigma)} s^2 \left(2 \beta \bar V_a + (\bar V_{ss} + 2 \gamma \bar V_{sa} + \gamma^2 \bar V_{aa})\right)\psi.
\end{align}

In step (v), we plug the candidate for $\wt\sigma$ back into the $O(\psi)$ term of the HJBI. Then the PDE for $\wt w$ obtained by setting the $O(\psi)$ term equal to zero has the following (formal) Feynman--Kac representation:
\begin{align}
\label{eqn:exotic:Feynman-Kac}
\wt w(t,s,y,a,m)
&= \EX[t,s,y,a,m]{\int_t^T \frac{\bar\sigma^2}{2f''(\bar\sigma)}\left( S_u^2 \left(2 \beta \bar V_a + (\bar V_{ss} + 2 \gamma \bar V_{sa} + \gamma^2 \bar V_{aa})\right)\right)^2 \dd u}.
\end{align}
Here, the expectation is computed under a measure such that $S$ has reference dynamics with initial conditions $S_t = s$, $Y_t = y$, $A_t = a$, $M_t = m$.

Finally, step (vi) is to find the candidate strategy by substituting the ansatz
\begin{align*}
\theta^\psi(t,s,y,a,m)
&= \bar\Delta(t,s,a,m) + \wt\theta(t,s,y,a,m)\psi
\end{align*}
into the HJBI equation and maximising the $O(\psi^2)$ term over $\wt\theta$. After some computations, we find
\begin{align}
\label{eqn:exotic:candidate strategy}
\theta^\psi
&= \bar\Delta + \left( \wt w_s + \gamma \wt w_a + \frac{U'}{U''}(\wt w_{sy} + \gamma \wt w_{ya})\right)\psi.
\end{align}

Let us discuss these results for some more specific examples:

\begin{example}[Asian and lookback options]
\label{ex:exotic:asian and lookback}
Suppose the payoff of the option depends on the arithmetic average of the stock price over the period $[0,T]$. In this case, we introduce the state variable $\diff A_t = S_t \dd t$ and write the payoff in the form $G(S_T, A_T)$. For example, the payoff of a floating strike Asian call is $(S_T - A_T/T)^+$. In the above setting, we thus have $\alpha(t,s,a,m) = s$ and $\beta = \gamma = \delta = 0$, and we immediately see that the general formulas \eqref{eqn:exotic:candidate volatility}, \eqref{eqn:exotic:Feynman-Kac}, and \eqref{eqn:exotic:candidate strategy} for the candidate controls and the cash equivalent $\wt w$ all reduce to those derived for a single vanilla option, except that the reference value $\bar V$ of the option depends on the additional state variable, and hence so do the candidate controls and $\wt w$. As a consequence, the cash gamma of the option is still the central determinant of the cash equivalent of small uncertainty aversion.

Likewise, for lookback-type options like the floating strike lookback put with payoff $M_T - S_T$, the formulas for the candidate controls and the cash equivalent $\wt w$ also essentially reduce to those of a vanilla option.
\end{example}

\begin{example}[Options on realised variance]
\label{ex:exotic:realised variance}
Here, the state variable $\diff A_t = \frac{1}{S_t^2}\dd\langle S \rangle_t$ tracks the cumulative realised variance of returns in the sense that if the true dynamics of $S$ are $\diff S_t = S_t \sigma_t \dd W_t$, then $A_t = \int_0^t \sigma_u^2 \dd u$. For example, the variance swap with strike volatility $\sigma_\textrm{strike}$ has the payoff $\frac{1}{T}A_T - \sigma_\textrm{strike}^2$; it pays the difference between the average realised variance over the period $[0,T]$ and a given strike variance $\sigma_\textrm{strike}^2$. In the above setting, we have $\beta(t,s,a,m) = s^{-2}$ and $\alpha = \gamma = \delta = 0$. We then obtain from \eqref{eqn:exotic:Feynman-Kac} that
\begin{align}
\label{eqn:exotic:realised variance:Feynman-Kac}
\wt w(t,s,y,a)
&= \EX[t,s,y,a]{\int_t^T \frac{(\bar\sigma(2 \bar V_a + \bar\Gamma^\$))^2}{2 f''(\bar\sigma)} \dd u}.
\end{align}
We see that the option's sensitivity with respect to changes in the realised variance and its cash gamma play symmetric roles here. Whence -- unlike for models with discrete trading \cite{HayashiMykland2005} or transaction costs \cite{KallsenMuhleKarbe2015} -- the quadratic variation of the reference hedge is generally not the right sufficient statistic here; cf.~the following Remark \ref{rem:general structure}.
\end{example}

\begin{remark}
\label{rem:general structure}
Suppose that the reference model has dynamics $\diff S_t = S_t \bar\sigma_t \dd W_t$ for some general, possibly path-dependent, volatility process $\bar\sigma$. Moreover, assume that $\bar\theta$ is a replicating strategy for an option $G(S_T,A_T)$ in that reference model. If $\bar\sigma_t = \bar\sigma(t,S_t)$ is actually of local volatility type, then $\bar\theta_t = \bar V_s(t,S_t,A_t)$ and by It\^o's formula, assuming $\gamma = 0$ so that $A$ is of finite variation,
\begin{align*}
S_t^2 \diff \langle \bar\theta \rangle_t
&= S_t^2 \diff \langle \bar V_s(\cdot,S_\cdot,A_\cdot) \rangle_t
= \left(\bar\sigma_t S_t^2 \bar V_{ss}(t,S_t,A_t)\right)^2 \dd t.
\end{align*}
Thus, in view of the representation of the cash equivalent in \eqref{eqn:Feynman-Kac}, which essentially also holds for Asian, lookback, and barrier options (cf.~\eqref{eqn:barrier:Feynman-Kac} below), one could be led to expect that the cash equivalent for possibly path-dependent volatilities generalises to
\begin{align}
\label{eqn:rem:general structure}
\EX{\int_0^T \frac{S_t^2}{2f''(\bar\sigma_t)} \dd \langle \bar\theta \rangle_t}.
\end{align}
However, in Example \ref{ex:exotic:realised variance}, the replicating strategy of the variance swap in the reference model is $\bar\theta_t = \bar V_s(t,S_t,A_t)$ and $\gamma = 0$, but the cash equivalent \eqref{eqn:exotic:realised variance:Feynman-Kac} differs from \eqref{eqn:rem:general structure}. This shows that \eqref{eqn:rem:general structure} is not the correct general form for the cash equivalent.
\end{remark}

\subsection{Barrier options}
\label{sec:barrier}

The payoff of barrier options depends on whether or not the stock price has hit a given barrier over its lifetime. For instance, a \emph{knock-out call} with barrier $B$ is an option that has the payoff of a vanilla call option provided that the stock price has not hit the barrier $B$ at any time before maturity. If the barrier has been hit, the payoff becomes zero. Conversely, the payoff of a \emph{knock-in call} becomes active only if the barrier has been hit before maturity, otherwise the payoff is zero. Barrier options are path-dependent in a rather weak sense. Their payoff depends only on two possible states -- whether or not the barrier has been hit. This allows to value barrier options in the local volatility model without introducing additional state variables by imposing suitable boundary conditions.

For simplicity, we focus on knock-out options whose payoff $G(S_T)$ is knocked out if the stock price breaches the barrier $B > S_0$ before maturity. It follows as in Sections~\ref{sec:hedging} and \ref{sec:exotic} that the fair value of such an option in the local volatility model can be expressed as the solution to the PDE
\begin{align}
\label{eqn:barrier:PDE}
\begin{split}
\bar V_t(t,s) + \frac{1}{2} \bar\sigma(t,s)^2 s^2 \bar V_{ss}(t,s)
&=0, \quad (t,s) \in [0,T)\times(0,B),\\
\bar V(t,B)
&= 0, \quad t\in[0,T),\\
\bar V(T,s)
&= G(s), \quad s \in (0,B).
\end{split}
\end{align}
$\bar V(t,s)$ represents the value of the knock-out option at time $t$ provided the stock price is $s$ \emph{and the option has not been knocked out yet}. The boundary condition $\bar V(t,B) = 0$ reflects the fact that the knock-out option becomes worthless if the stock price hits the barrier $B$ before maturity. If you have sold such an option, your theoretical P\&L can be written as
\begin{align}
\label{eqn:barrier:Y definition}
Y_t
&= x_0 + \int_0^t \theta_u \dd S_u - \bar V(t,S_t)\1_{\lbrace t\leq\rho \rbrace},
\end{align}
where $\rho$ is the first time that $S$ hits the barrier $B$. Now, note that the boundary condition in \eqref{eqn:barrier:PDE} implies $\bar V(t,S_t)\1_{\lbrace t\leq\rho \rbrace} = \bar V(t\wedge\rho,S_{t\wedge\rho})$. It\^o's formula in turn yields
\begin{align*}
\bar V(t\wedge\rho,S_{t\wedge\rho})
&= \bar V(0,s_0)
+ \int_0^t \bar V_s(u,S_u)\1_{\lbrace u < \rho \rbrace}\dd S_u \\
&\qquad+ \int_0^t \left( \bar V_t(u,S_u)+ \frac{1}{2} \sigma_u^2 S_u^2 \bar V_{ss}(u,S_u) \right)\1_{\lbrace u < \rho \rbrace} \dd u.
\end{align*}
Using the PDE \eqref{eqn:barrier:PDE} to substitute the $\bar V_t$ term and plugging the result into \eqref{eqn:barrier:Y definition} then gives the dynamics of $Y$:
\begin{align}
\label{eqn:barrier:Y}
\dd Y_t
&= (\theta_t - \bar\Delta_t) \dd S_t
+ \frac{1}{2}\bar\Gamma^\$_t (\bar\sigma(t,S_t)^2 -\sigma_t^2) \dd t,
\end{align}
where $\bar\Delta_t = \bar V_s(t,S_t) \1_{\lbrace t < \rho \rbrace}$ is the delta of the knock-out option and $\bar\Gamma^\$_t = S_t^2 \bar V_{ss}(t,S_t)\1_{\lbrace t < \rho \rbrace}$ is its cash gamma. As one would expect, these quantities are zero after the barrier has been hit.

To find the candidate optimal controls, we have to distinguish two cases. After the barrier has been hit, the option is worthless and no longer needs to be hedged. Therefore, the candidate strategy is simply $0$. With this strategy, your P\&L stays constant independently of which volatility ``nature'' chooses (cf.~\eqref{eqn:barrier:Y}). Therefore, there is no incentive for ``nature'' to deviate from the reference dynamics and the candidate volatility is simply the reference volatility $\bar\sigma$. The HJBI equation corresponding to the hedging problem before the barrier has been hit is the same as in the case of a single vanilla option (but with an additional boundary condition). Therefore, we obtain essentially the same candidates for the optimal controls, namely
\begin{align*}
\sigma^\psi_t
&= \bar\sigma(t,S_t) + \frac{\bar\sigma(t,S_t)\bar\Gamma^\$_t}{f''(t,S_t,Y_t;\bar\sigma(t,S_t))}\psi,\\
\theta^\psi_t
&= \bar\Delta_t + \left(\wt w_s(t,S_t,Y_t) + \frac{U'(Y_t)}{U''(Y_t)}\wt w_{sy}(t,S_t,Y_t)\right)\1_{\lbrace t < \rho\rbrace} \psi,
\end{align*}
where the cash equivalent $\wt w$ solves the PDE \eqref{eqn:PDE:first order} with the boundary condition ${\wt w(t,B,y)}= 0$. The (formal) Feynman--Kac representation of this PDE reads as
\begin{align}
\label{eqn:barrier:Feynman-Kac}
\wt w(t,s,y)
&= \EX[t,s]{\int_t^T \frac{\left(\bar\sigma(u,S_u)\bar\Gamma^\$_u\right)^2}{2f''(u,S_u,y;\bar\sigma(u,S_u))} \dd u}.
\end{align}
As a result, the expected volatility-weighted cash gamma accumulated over the remaining lifetime of the barrier option is again the major driver of the cash equivalent of small uncertainty aversion.

\subsection{Option portfolios}
\label{sec:portfolio}

Instead of a single option with maturity $T$, we now consider a whole portfolio of vanilla options with possibly different maturities.\footnote{Portfolios including exotic options can be treated along the same lines; we do not pursue this here to ease notation.} Suppose that you have sold $N$ options with maturities $T_1,\ldots,T_N \in [0,T]$ and payoffs $G_1(S_{T_1}),\ldots,G_N(S_{T_N})$, respectively. Let $\bar V^i(t,s)$ denote the reference value of the option $G_i$, which solves
\begin{align}
\label{eqn:portfolio:pricing PDE}
\begin{split}
\bar V^i_t(t,s) + \frac{1}{2} \bar\sigma(t,s)^2 s^2 \bar V^i_{ss}(t,s)
&=0, \quad (t,s) \in [0,T_i)\times\RR_+,\\
\bar V^i(T_i,s)
&= G_i(s), \quad s \in \RR_+.
\end{split}
\end{align}
Your theoretical P\&L at time $t$ can be expressed as
\begin{align}
\label{eqn:portfolio:Y definition}
Y_t
&= x_0 + \int_0^t \theta_u \dd S_u - \sum_{i=1}^N \bar V^i(t\wedge T_i,S_{t\wedge T_i}).
\end{align}
We emphasise that for $t = T$, using the terminal conditions of the PDEs \eqref{eqn:portfolio:pricing PDE},
\begin{align*}
Y_T
&= x_0 + \int_0^T \theta_u \dd S_u - \sum_{i=1}^N G_i(S_{T_i})
\end{align*}
is your \emph{actual} final P\&L. We next determine the drift and diffusion parts of $Y$ if the true stock volatility is given by some process $\sigma$. As in Section~\ref{sec:hedging}, using It\^o's formula and the PDEs \eqref{eqn:portfolio:pricing PDE}, we find that
\begin{align}
\label{eqn:portfolio:theoretical value}
\bar V^i(t\wedge T_i,S_{t\wedge T_i}) 
&= \bar V^i(0,s_0) + \int_0^t \bar\Delta^i(u,S_u) \dd S_u - \frac{1}{2}\int_0^t \bar\Gamma^{\$,i}(u,S_u) (\bar\sigma(u,S_u)^2 - \sigma_u^2) \dd u,
\end{align}
where $\bar\Delta^i(u,s) := \bar V^i_s(u,s)\1_{\lbrace u < T_i\rbrace}$ is the delta of the option $G_i$ and $\bar\Gamma^{\$,i}(u,s) := s^2 \bar V^i_{ss}(u,s)\1_{\lbrace u < T_i\rbrace}$ is its cash gamma. Substituting \eqref{eqn:portfolio:theoretical value} into \eqref{eqn:portfolio:Y definition}, we obtain
\begin{align*}
\diff Y_t
&= \left(\theta_t - \bar\Delta(t,S_t)\right)\dd S_t +\frac{1}{2}\bar\Gamma^{\$}(t,S_t) (\bar\sigma(t,S_t)^2 - \sigma_t^2) \dd t,
\end{align*}
where $\bar\Delta(t,s) := \sum_{i=1}^N \bar\Delta^i(t,s)$ is the \emph{net delta} of your option portfolio and $\bar\Gamma^\$(t,s):= \sum_{i=1}^N \bar\Gamma^{\$,i}(t,s)$ is its \emph{net cash gamma}. We see that the dynamics of $Y$ have exactly the same form as for the case of a single vanilla option (cf.~\eqref{eqn:asymptotic:Y} and \eqref{eqn:asymptotic:b for single option}). The only difference is that the delta and the cash gamma of the single option are replaced by the net delta and net cash gamma of the option portfolio. Therefore, we obtain analogous candidates for the optimal controls (in feedback form) and the cash equivalent:
\begin{align}
\label{eqn:portfolio:candidate volatility}
\sigma^\psi(t,s,y)
&= \bar\sigma(t,s) + \frac{\bar\sigma(t,s) \bar\Gamma^{\$}(t,s)}{f''(t,s,y;\bar\sigma(t,s))}\psi,\\
\label{eqn:portfolio:candidate strategy}
\theta^\psi(t,s,y)
&= \bar\Delta(t,s) + \left(\wt w_s(t,s,y) + \frac{U'(y)}{U''(y)} \wt w_{sy}(t,s,y)\right)\psi,\\
\label{eqn:portfolio:Feynman-Kac}
\wt w(t,s,y)
&= \EX[t,s]{\int_t^T \frac{\left(\bar\sigma(u,S_u)\bar\Gamma^\$(u,S_u)\right)^2}{2f''(u,S_u,y;\bar\sigma(u,S_u))} \dd u}.
\end{align}

\subsection{Static hedging with vanilla options}
\label{sec:static}

So far, the only hedging instrument available was the stock. In practice, however, if liquidly traded options are available, these may be used as additional hedging instruments for more complex derivatives. Therefore, we now assume that in addition to trading in the stock you can buy or sell, at time $0$, any quantity of $M$ vanilla options with maturities $T_1,\ldots,T_M\in[0,T]$ and payoffs $F_1(S_{T_1}),\ldots,F_M(S_{T_M})$. We suppose that these options are available for prices $p_1,\ldots,p_M$. In the context of worst-case superhedging, this setup is known as the Lagrangian uncertain volatility model \cite{AvellanedaParas1996}; also compare \cite{Mykland2003.options}. For each $i = 1,\ldots,M$, let $\bar V^i(t,s)$ be the reference value of the option $F_i$, which solves
\begin{align}
\label{eqn:static:pricing PDE}
\begin{split}
\bar V^i_t(t,s) + \frac{1}{2} \bar\sigma(t,s)^2 s^2 \bar V^i_{ss}(t,s)
&=0, \quad (t,s) \in [0,T_i)\times\RR,\\
\bar V^i(T_i,s)
&= F_i(s), \quad s \in \RR.
\end{split}
\end{align}
We require that the reference model is consistent with the observed prices at time $0$ in the sense that $\bar V^i(0,s_0)= p_i$ for $i=1,\ldots,M$.\footnote{That is, the local volatility model is calibrated to the observed market prices of the liquid options at time $0$.} For notational simplicity, we assume that you have to hedge a portfolio of $N$ vanilla options with maturities $T_{M+1},\ldots,T_{M+N} \in [0,T]$ and payoffs $G_{M+1}(S_{T_{M+1}}),\ldots,G_{M+N}(S_{T_{M+N}})$.\footnote{Portfolios of barrier options as in \cite{AvellanedaBuff1999} or other exotics can be treated along the same lines, but require a more extensive notation.} We assume that for each $i={M+1},\ldots, {M+N}$, the reference value $\bar V^i$ of the option $G_i(S_{T_i})$ satisfies the PDE \eqref{eqn:portfolio:pricing PDE}.

Suppose that you \emph{buy} $\lambda_i$ options with payoff $F_i(S_{T_i})$ for price $p_i$ at time $0$ (a negative $\lambda_i$ indicates a short sale) for $i=1,\ldots,M$ and follow a self-financing trading strategy $\theta$ for the stock. Then your theoretical P\&L at time $t$ is
\begin{align}
\label{eqn:static:Y definition}
Y_t
&= x_0 + \int_0^t \theta_u \dd S_u - \sum_{i=M+1}^{M+N}\bar V^i(t \wedge T_i,S_{t \wedge T_i}) + \sum_{i=1}^M \lambda_i\left(\bar V^i(t\wedge T_i, S_{t\wedge T_i}) - p_i \right).
\end{align}
Note that the consistency condition $\bar V^i(0,s_0)= p_i$ implies that our choice of $\lambda_i$ does not affect the theoretical P\&L at time $0$. Moreover, by \eqref{eqn:static:pricing PDE} and \eqref{eqn:portfolio:pricing PDE},
\begin{align*}
Y_T
&= x_0 + \int_0^T \theta_u \dd S_u -  \sum_{i=M+1}^{M+N} G_i(S_{T_i})   + \sum_{i=1}^M \lambda_i\left(F_i(S_{T_i}) - p_i \right)
\end{align*}
is your \emph{actual} final P\&L.

Looking at \eqref{eqn:static:Y definition}, we recognise that (up to linear transformations that can be incorporated into the option payoffs) we are exactly in the setting of an option portfolio discussed in Section~\ref{sec:portfolio} (with $N$ replaced by $N+M$). Hence, we obtain the same candidate controls \eqref{eqn:portfolio:candidate volatility}--\eqref{eqn:portfolio:candidate strategy} and cash equivalent \eqref{eqn:portfolio:Feynman-Kac} with net delta
\begin{align*}
\bar\Delta
&:= \sum_{i=M+1}^{M+N}\bar\Delta^i - \sum_{i=1}^M\lambda_i\bar\Delta^i
\end{align*}
and net cash gamma
\begin{align*}
\bar\Gamma^{\$}
&:= \sum_{i=M+1}^{M+N}\bar\Gamma^{\$,i} - \sum_{i=1}^M \lambda_i \bar\Gamma^{\$,i}.
\end{align*}
In particular, denoting by $\bar\Gamma^{\$,0} := \sum_{i=M+1}^{M+N} \bar\Gamma^{\$,i}$ the net cash gamma of your original book (before buying or selling other options), the cash equivalent of the combined portfolio has the following representation:
\begin{align*}
\wt w(t,s,y)
&= \EX[t,s]{\int_t^T \frac{\left(\bar\sigma(u,S_u)(\bar\Gamma^{\$,0}-\sum_{i=1}^M \lambda_i\bar\Gamma^{\$,i})(u,S_u)\right)^2}{2f''(u,S_u,y;\bar\sigma(u,S_u))} \dd u}.
\end{align*}

This yields a criterion to manage a portfolio's sensitivity to volatility uncertainty by trading statically in options: find $\lambda_i$'s that minimise $\wt w$, i.e., that minimise the expected volatility-weighted net cash gamma accumulated over the remaining lifetime of the option portfolio. In Section~\ref{sec:measure}, we show that this minimised cash equivalent, viewed as a function of the original portfolio's net cash gamma, satisfies certain axiomatic properties that have been advocated for measures of model uncertainty for derivatives.

\subsection{The cash equivalent as a measure of model uncertainty}
\label{sec:measure}

Consider a mapping $\mu$ which assigns a nonnegative number to any contingent claim that has a well-defined value in any model of a given family. Cont \cite{Cont2006} calls $\mu$ a \emph{measure of model uncertainty} if it satisfies four axioms that reflect the possibility of full or partial hedges in the underlying or in liquidly traded options; cf.~\cite[Section 4.1]{Cont2006} for more details. Here, we restrict attention to the linear space $\cX$ of claims of the form 
\begin{align}
\label{eqn:measure:claim}
G
&= c + \sum_{i=M+1}^{M+N} G_i(S_{T_i}) + \int_0^T \theta_t \dd S_t,
\end{align}
where $N \in \NN$, the payoffs $G_i(S_{T_i})$ and maturities $T_i$ are as in Section~\ref{sec:static}, $c\in\RR$ is a constant, and $\theta$ is a sufficiently regular trading strategy so that the stochastic integral can be defined path-wise.\footnote{For instance, if $\theta$ is of finite variation, then the stochastic integral can be defined path-wise via the integration by parts formula.} We have seen that we can associate to each of these claims its cash gamma $\bar\Gamma^{\$,G}$ as the sum $\sum_{i=M+1}^{M+N} s^2 \bar V^i_{ss}(t,s)\1_{\lbrace t < T_i \rbrace}$ (note that $c$ and the stochastic integral in \eqref{eqn:measure:claim} do not contribute).

Recall the setup and notation of Section~\ref{sec:static}. For notational convenience, we combine the cash gammas $\bar\Gamma^{\$,i}$, $i=1,\ldots,M$, into a vector $\bar\Gamma^\$$ of functions. Now, define (for some fixed $t,s$) the function $\mu: \cX \to \RR_+$ by
\begin{align*}
\mu(G) 
&= \inf_{\lambda\in\RR^M}\EX[t,s]{\int_t^T \frac{\left(\bar\sigma(u,S_u)(\bar\Gamma^{\$,G}-\lambda \cdot \bar\Gamma^\$)(u,S_u)\right)^2}{2f''(u,S_u,y;\bar\sigma(u,S_u))} \dd u},
\end{align*}
which maps a claim $G \in \cX$ to its cash equivalent of small uncertainty aversion minimised over all static hedges in liquidly traded options. This mapping fulfills the following (suitably modified\footnote{Unlike \cite{Cont2006}, we disregard bid-ask spreads for the liquidly traded options.}) axioms of \cite{Cont2006}:
\begin{enumerate}
\item There is no model uncertainty for liquidly traded options:
\begin{align*}
\mu(F_i(S_{T_i}))
&= 0 \quad\text{for all } i=1,\ldots,M.
\end{align*}
Moreover, $\mu(c) = 0$ for all constants $c\in\RR$.
\item $\mu$ accounts for hedging possibilities provided by dynamic trading in the underlying:
\begin{align*}
\mu\left(G + \int_0^T \theta_t \dd S_t\right)
&= \mu(G) \quad \text{for all trading strategies }\theta\text{ and }G \in \cX.
\end{align*}

\item Diversification decreases the model uncertainty of a portfolio:
\begin{align*}
\mu(\nu G + (1-\nu) G')
&\leq \nu\mu(G) + (1-\nu)\mu(G')
\quad\text{for all }\nu\in[0,1]\text{ and }G,G' \in \cX.
\end{align*}

\item $\mu$ accounts for hedging possibilities provided by static hedges with liquidly traded options:
\begin{align*}
\mu\left(G + \sum_{i=1}^M \lambda_i F_i(S_{T_i})\right) 
&= \mu(G)\quad\text{for all }\lambda\in\RR^M\text{ and } G \in \cX.
\end{align*}
\end{enumerate}

We also note that by construction, $\mu$ becomes smaller as the set of liquidly traded options expands; this is another natural requirement that has been pointed out in \cite{Cont2006}. Properties (i), (ii), and (iv) are immediate from the definition of $\mu$. The convexity property (iii) can be verified as follows. Fix $\nu \in [0,1]$, $G,G' \in \cX$, and denote by $\bar\Gamma^{\$,G}$ and $\bar\Gamma^{\$,G'}$ the cash gammas of $G$ and $G'$. Fix $\varepsilon > 0$. By the definition of $\mu$, we may choose $\lambda, \lambda' \in \RR^M$ such that
\begin{align}
\label{eqn:measure:pf:lambda choice}
\mu(G)
&\leq \EX[t,s]{\int_t^T \frac{\left(\bar\sigma(u,S_u)(\bar\Gamma^{\$,G}-\lambda \cdot\bar\Gamma^{\$})(u,S_u)\right)^2}{2f''(u,S_u,y;\bar\sigma(u,S_u))} \dd u} + \frac{\varepsilon}{2}
\end{align}
and the analogous inequality with $G$ and $\lambda$ replaced by $G'$ and $\lambda'$ holds as well. Define $\lambda'' := \nu \lambda + (1-\nu)\lambda'$ and $G'' := \nu G + (1-\nu) G'$. Then
\begin{align*}
\bar\Gamma^{\$,G''} - \lambda''\cdot\bar\Gamma^\$
&= \nu\bar\Gamma^{\$,G}+(1-\nu)\bar\Gamma^{\$,G'}-(\nu \lambda + (1-\nu)\lambda')\cdot\bar\Gamma^\$\\
&= \nu(\bar\Gamma^{\$,G} - \lambda\cdot\bar\Gamma^\$) + (1-\nu)(\bar\Gamma^{\$,G'} - \lambda'\cdot\bar\Gamma^\$).
\end{align*}
Together with the convexity of $x\mapsto x^2$, this yields
\begin{align}
\label{eqn:measure:pf:convexity}
(\bar\Gamma^{\$,G''} - \lambda''\cdot\bar\Gamma^\$)^2
&\leq \nu(\bar\Gamma^{\$,G} - \lambda\cdot\bar\Gamma^\$)^2 + (1-\nu)(\bar\Gamma^{\$,G'} - \lambda'\cdot\bar\Gamma^\$)^2.
\end{align}
Using the definition of $\mu$, \eqref{eqn:measure:pf:convexity} and \eqref{eqn:measure:pf:lambda choice}, we find
\begin{align*}
\mu(G'')
&\leq \EX[t,s]{\int_t^T \frac{\left(\bar\sigma(u,S_u)(\bar\Gamma^{\$,G''}-\lambda'' \cdot\bar\Gamma^{\$})(u,S_u)\right)^2}{2f''(u,S_u,y;\bar\sigma(u,S_u))} \dd u} \\
&\leq \nu \mu(G) + (1-\nu)\mu(G') + \varepsilon.
\end{align*}
The assertion now follows by taking the limit $\varepsilon \downarrow 0$.

\section{Proofs}
\label{sec:proofs}

In this section, we rigorously prove the results from Section~\ref{sec:main results}. Throughout, we assume that Assumption~\ref{ass:second order} is in force. To ease notation, define for $(t,s,y) \in \bfD =  (0,T)\times(K^{-1},K)\times(y_l,y_u)$ and $\psi > 0$,
\begin{align}
\label{eqn:candidate functions}
\theta^\psi(t,s,y)
&:= \bar V_s(t,s) + \wt\theta(t,s,y)\psi,
&\sigma^\psi(t,s,y)
&:=\bar\sigma(t,s) + \wt\sigma(t,s,y)\psi.
\end{align}
Note that we use the symbols $\theta^\psi$ and $\sigma^\psi$ for both the functions defined in 
\eqref{eqn:candidate functions} and the candidate controls defined in Theorem~\ref{thm:second order}. This is, of course, motivated by the relationships
\begin{align*}
\theta^\psi_t
&= \theta^\psi(t,S_t,Y_t) \1_{\lbrace t < \tau \rbrace} + \bar\Delta_t \1_{\lbrace t \geq \tau \rbrace}
= \bar\Delta_t + \wt\theta(t,S_t,Y_t)  \1_{\lbrace t < \tau \rbrace} \psi,\\
\sigma^\psi_t
&= \sigma^\psi(t,S_t,Y_t).
\end{align*}

\subsection{Value expansion and almost optimality of the candidate strategy}
\label{sec:proof of main result}

In this section, we prove Theorem~\ref{thm:second order}. Throughout, we assume that $\cA$, $\cV$, and $\psi_c > 0$ are chosen such that $(\theta^\psi, \sigma^\psi) \in \cA \times \cV \subset \cZ$ for every $\psi \in (0,\psi_c)$. The concrete construction of such scenarios summarised in Theorem~\ref{thm:existence} is carried out in Section~\ref{sec:existence}. 

Define the \emph{candidate value function} $w:\ol\bfD\times \RR \to \RR$ by
\begin{align}
\label{eqn:candidate value function}
w(t,s,y;\psi)
&:= w^\psi(t,s,y)
:=  U(y) - U'(y) \wt w(t,s,y)\psi + U'(y) \wh w(t,s,y)\psi^2.
\end{align}
We note that by our assumptions on $\wt w$, $\wh w$, and $U$ in Assumption~\ref{ass:second order}, there is a constant $K' > 0$ (depending only on $K$, $y_l$, $y_u$, and $U$) such that
\begin{align}
\label{eqn:candidate value function:bounds}
\vert w^\psi_t \vert, \vert w^\psi_s \vert, \vert w^\psi_y \vert, \vert w^\psi_{ss} \vert, \vert w^\psi_{sy} \vert, \vert w^\psi_{yy} \vert \leq K' \text{ on } \bfD \times [-1,1].
\end{align}
Moreover, all partial derivatives of $w$ in \eqref{eqn:candidate value function:bounds} are evidently $C^\infty$ in $\psi$.

In essence, the proof of Theorem~\ref{thm:second order} boils down to showing that our candidate value function $w^\psi$, candidate strategy $\theta^\psi$, and candidate volatility $\sigma^\psi$ are approximate solutions to the Hamilton--Jacobi--Bellman--Isaacs equation associated to the hedging problem \eqref{eqn:value} in the sense that
\begin{align*}
w^\psi_t(t,s,y) + \sup_\vartheta \inf_\varsigma H^\psi(t,s,y;\vartheta,\varsigma)
&= o(\psi^2) \quad \text{as } \psi \downarrow 0,\text{ uniformly in } (t,s,y);
\end{align*}
cf.~Lemma~\ref{lem:horrible calculation} below. Here, for $\psi \neq 0$, the Hamiltonian $H^\psi: \bfD \times \RR \times [0,K]\to \RR$ is given by
\begin{align}
\label{eqn:hamiltonian}
H^\psi(t,s,y;\vartheta,\varsigma)
&= \frac{1}{\psi} U'(y) f(t, s, y; \varsigma)
+ \frac{1}{2} s^2 \bar V_{ss}(t,s) (\bar\sigma(t,s)^2 - \varsigma^2) w^\psi_y(t,s,y)\notag\\
&\quad+ \frac{1}{2} \varsigma^2 s^2 \left(w^\psi_{ss}(t,s,y) + 2 (\vartheta - \bar V_s(t,s)) w^\psi_{sy}(t,s,y) + (\vartheta - \bar V_s(t,s))^2 w^\psi_{yy}(t,s,y) \right);
\end{align}
recall \eqref{eqn:asymptotic:HJBI} and \eqref{eqn:asymptotic:b for single option} from the heuristic derivation of the HJBI equation in Section~\ref{sec:general}. This part of the proof is purely analytic and is carried out in Section \ref{sec:HJBI}, the main ingredients being the implicit function theorem and Taylor expansions. Then, adapting classical verification arguments to the asymptotic setting allows us to prove the two inequalities
\begin{align}
\label{eqn:first inequality}
\inf_{\sigma\in\cV} \inf_{P\in\fP(\theta^\psi,\sigma)} J^\psi(\sigma,P)
&\geq w^\psi_0 + o(\psi^2),\\
\label{eqn:second inequality}
\sup_{\theta\in\cA} \inf_{P\in\fP(\theta,\sigma^\psi)} J^\psi(\sigma^\psi,P)
&\leq w^\psi_0 + o(\psi^2),
\end{align}
where $w^\psi_0 := w^\psi(0,s_0,y_0)$; cf.~Lemmas~\ref{lem:first inequality} and \ref{lem:second inequality}. Denoting by $\lesssim$ ``less or equal up to a term of order $o(\psi^2)$'', we obtain from \eqref{eqn:first inequality}--\eqref{eqn:second inequality} that
\begin{align*}
w^\psi_0
&\lesssim \inf_{\sigma\in\cV} \inf_{P\in\fP(\theta^\psi,\sigma)} J^\psi(\sigma,P)
\lesssim \sup_{\theta \in \cA} \inf_{\sigma\in\cV} \inf_{P\in\fP(\theta,\sigma)} J^\psi(\sigma,P)
\lesssim \sup_{\theta \in \cA} \inf_{P\in\fP(\theta,\sigma^\psi)} J^\psi(\sigma^\psi,P)
\lesssim w^\psi_0
\end{align*}
and
\begin{align*}
w^\psi_0
&\lesssim \inf_{\sigma\in\cV} \inf_{P\in\fP(\theta^\psi,\sigma)} J^\psi(\sigma,P)
\lesssim \inf_{P\in\fP(\theta^\psi,\sigma^\psi)} J^\psi(\sigma^\psi,P)
\lesssim \sup_{\theta \in \cA} \inf_{P\in\fP(\theta,\sigma^\psi)} J^\psi(\sigma^\psi,P)
\lesssim w^\psi_0.
\end{align*}
Hence, we have equality up to a term of order $o(\psi^2)$ everywhere. In particular,
\begin{align*}
\sup_{\theta \in \cA} \inf_{\sigma\in\cV} \inf_{P\in\fP(\theta,\sigma)} J^\psi(\sigma,P)
&= w^\psi_0 + o(\psi^2)
= \inf_{P\in\fP(\theta^\psi,\sigma^\psi)} J^\psi(\sigma^\psi,P) + o(\psi^2).
\end{align*}
This completes the proof of Theorem~\ref{thm:second order} modulo the proofs of \eqref{eqn:first inequality}--\eqref{eqn:second inequality}.

\subsubsection{Approximate solution to HJBI equation}
\label{sec:HJBI}

We first determine the minimiser of the Hamiltonian with respect to the volatility variable $\varsigma$, and identify it at the leading order as the candidate $\sigma^\psi$ from \eqref{eqn:candidate functions}.

\begin{lemma}
\label{lem:optimal volatility}
Fix $L > 0$. Then there are constants $C_1 > 0$ and $\psi_1 > 0$ (depending on $L$) such that for every $(t,s,y)\in\bfD$, $\wt\vartheta\in[-L,L]$, and $\psi \in (0,\psi_1)$, the function
\begin{align}
\label{eqn:lem:optimal volatility:function}
[0,K] \ni \varsigma \mapsto H^\psi(t,s,y;\bar V_s(t,s) + \wt\vartheta\psi, \varsigma)
\end{align}
has a minimiser $\varsigma_*^\psi(t,s,y,\wt\vartheta)$ that satisfies the first-order condition
\begin{align}
\label{eqn:lem:optimal volatility:foc}
\frac{\partial H^\psi}{\partial\varsigma} (t,s,y;\bar V_s(t,s) + \wt\vartheta\psi, \varsigma_*^\psi(t,s,y,\wt\vartheta) )
&=0,
\end{align}
and
\begin{align}
\label{eqn:lem:optimal volatility:estimate}
\left\lvert\varsigma_*^\psi(t,s,y,\wt\vartheta) - \sigma^\psi(t,s,y)\right\vert
&\leq C_1 \psi^2.
\end{align}
\end{lemma}

\begin{proof}
As $H^\psi$ is continuous in $\varsigma$ and $[0,K]$ is compact, there exists a minimiser $\varsigma^\psi_* = \varsigma^\psi_*(t,s,y,\wt\vartheta)$ of \eqref{eqn:lem:optimal volatility:function} for every $(t,s,y)\in\bfD$, $\wt\vartheta \in [-L,L]$, and $\psi > 0$. Next, the basic idea is to employ the convexity of the penalty function $f$ to show that for sufficiently small $\psi$, $\varsigma^\psi_*$ has to lie in the interior of $[0,K]$. As a consequence, it satisfies the first-order condition.

To make this precise, first note from \eqref{eqn:ass:reference price} and \eqref{eqn:candidate value function:bounds} that there is a constant $K'' > 0$ such that for all $(t,s,y)\in\bfD$, $\wt\vartheta \in [-L,L]$, and $\psi \in (-1,1)$,
\begin{align}
\label{eqn:lem:optimal volatility:pf:diffusion term estimate}
\left\vert \frac{1}{2}s^2 \left(w^\psi_{ss}(t,s,y) + 2 \wt\vartheta \psi w^\psi_{sy}(t,s,y) + (\wt\vartheta \psi)^2 w^\psi_{yy}(t,s,y) -\bar V_{ss}(t,s) w^\psi_y(t,s,y)\right)\right\vert
&\leq K''.
\end{align}
Since $\bar\sigma(t,s) \in [0,K]$ on $[0,T]\times[K^{-1},K]$ by \eqref{eqn:ass:reference volatility}, we have
\begin{align*}
H^\psi(t,s,y;\bar V_s(t,s) + \wt\vartheta\psi, \varsigma^\psi_*(t,s,y,\wt\vartheta))
&\leq H^\psi(t,s,y;\bar V_s(t,s) + \wt\vartheta\psi, \bar\sigma(t,s)).
\end{align*}
On the one hand, using the definition of $H^\psi$ together with \eqref{eqn:penalty function:minimum conditions} and rearranging terms, this inequality implies
\begin{align}
&\frac{1}{\psi}U'(y)f(t,s,y;\varsigma^\psi_*)\notag\\
&\;\leq \left(\bar\sigma(t,s)^2 - (\varsigma^\psi_*)^2 \right)
\frac{1}{2}s^2 \left(w^\psi_{ss}(t,s,y) + 2 \wt\vartheta \psi w^\psi_{sy}(t,s,y) + (\wt\vartheta \psi)^2 w^\psi_{yy}(t,s,y)-\bar V_{ss}(t,s)w^\psi_y(t,s,y)\right)\notag\\
&\;\leq 2KK''\left\vert\bar\sigma(t,s) - \varsigma^\psi_* \right\vert.
\label{eqn:lem:optimal volatility:pf:10}
\end{align}
On the other hand, assumption \eqref{eqn:ass:penalty function} together with \eqref{eqn:penalty function:minimum conditions} yields
\begin{align}
\label{eqn:lem:optimal volatility:pf:20}
f(t,s,y;\varsigma^\psi_*)
&\geq \frac{1}{2K} (\bar\sigma(t,s) - \varsigma^\psi_*)^2.
\end{align}
Combining \eqref{eqn:lem:optimal volatility:pf:10}--\eqref{eqn:lem:optimal volatility:pf:20} and rearranging terms gives
\begin{align}
\label{eqn:lem:optimal volatility:pf:30}
\left\vert \bar\sigma(t,s) - \varsigma^\psi_* \right\vert
&\leq \frac{4K^2K''}{U'(y_u)} \psi; 
\end{align}
note that this inequality is trivially true if $\bar\sigma(t,s) = \varsigma^\psi_*$, so the division by $\vert \bar\sigma(t,s) - \varsigma^\psi_* \vert$ in the last step is justified. Now, since $\bar\sigma(t,s)$ is uniformly in the interior of $[0,K]$ by assumption \eqref{eqn:ass:reference volatility}, it follows from \eqref{eqn:lem:optimal volatility:pf:30} that there is $\psi_1 \in (0,1)$ such that for every $(t,s,y)\in \bfD$, $\wt\vartheta \in [-L,L]$, and $\psi \in (0,\psi_1)$, we have $\varsigma^\psi_*(t,s,y,\wt\vartheta) \in (0,K)$. This implies that $\varsigma^\psi_*$ satisfies the first-order condition \eqref{eqn:lem:optimal volatility:foc} or equivalently (through multiplication by $\psi > 0$),
\begin{align}
&U'(y)f'(t,s,y;\varsigma^\psi_*) - s^2 \bar V_{ss}(t,s)\varsigma^\psi_* w^\psi_y(t,s,y)\psi\notag\\
\label{eqn:lem:optimal volatility:pf:foc}
&\qquad+ \varsigma^\psi_* s^2\left(w^\psi_{ss}(t,s,y) + 2 \wt\vartheta \psi w^\psi_{sy}(t,s,y) + (\wt\vartheta \psi)^2 w^\psi_{yy}(t,s,y)\right)\psi
=0.
\end{align}

It remains to prove \eqref{eqn:lem:optimal volatility:estimate}. For each $\lambda=(t,s,y,\wt\vartheta) \in \bfD\times[-L,L]$, we define the function ${F_\lambda:(-\psi_1,\psi_1) \times [0,K] \to \RR}$, $(\psi, \varsigma) \mapsto F_\lambda(\psi,\varsigma)$, by the left-hand side of \eqref{eqn:lem:optimal volatility:pf:foc} with $\varsigma^\psi_*$ replaced by $\varsigma$. As $F_\lambda$ is a polynomial in $\psi$ and $f$ is $C^4$ in $\varsigma$, $F_\lambda$ is $C^3$. Fix $\lambda = (t,s,y,\wt\vartheta)$. By construction,
\begin{align}
\label{eqn:lem:optimal volatility:pf:implicit characterisation}
F_\lambda(\psi,\varsigma^\psi_*)
&= 0, \quad \psi \in (0,\psi_1).
\end{align}
Now, we want to invoke the implicit function theorem to show that $\psi\mapsto\varsigma^\psi_*$ can be extended via \eqref{eqn:lem:optimal volatility:pf:implicit characterisation} to a $C^3$ function on $(-\psi_1,\psi_1)$ (choosing $\psi_1$ smaller if necessary). To this end, it suffices to show that $\frac{\partial F_\lambda}{\partial\varsigma} \geq \varepsilon$ for some $\varepsilon > 0$. Using \eqref{eqn:ass:penalty function} and \eqref{eqn:lem:optimal volatility:pf:diffusion term estimate}, we obtain for all $\lambda = (t,s,y,\wt\vartheta) \in \bfD \times [-L,L]$ and $\psi\in(-\psi_1,\psi_1)$,
\begin{align*}
\frac{\partial F_\lambda}{\partial\varsigma}(\psi,\varsigma)
&= U'(y)f''(t,s,y;\varsigma) - s^2 \bar V_{ss}(t,s) w^\psi_y(t,s,y)\psi\\
&\qquad+ s^2\left(w^\psi_{ss}(t,s,y) + 2 \wt\vartheta \psi w^\psi_{sy}(t,s,y) + (\wt\vartheta \psi)^2 w^\psi_{yy}(t,s,y)\right)\psi\\
&\geq \frac{U'(y_u)}{K} - 2K''\vert\psi\vert.
\end{align*}
Hence, choosing $\psi_1$ smaller if necessary, there is $\varepsilon > 0$ such that 
\begin{align}
\label{eqn:lem:optimal volatility:pf:positive sigma derivative}
\frac{\partial F_\lambda}{\partial\varsigma}(\psi,\varsigma)
&\geq \varepsilon, \quad \lambda \in \bfD\times[-L,L], \psi \in (-\psi_1, \psi_1), \varsigma \in [0,K],
\end{align}
and the implicit function theorem implies that for each fixed $\lambda$, $\psi \mapsto \varsigma^\psi_*$ can be extended via \eqref{eqn:lem:optimal volatility:pf:implicit characterisation} to $(-\psi_1,\psi_1)$ and is $C^3$. As $F_\lambda(0,\varsigma) = U'(y)f'(t,s,y;\varsigma)$, the uniqueness assertion of the implicit function theorem together with \eqref{eqn:penalty function:minimum conditions} also yields that $\varsigma^0_* = \bar\sigma(t,s)$. To compute $\frac{\partial\varsigma^\psi_*}{\partial\psi}(0)$, we observe using the first-order condition \eqref{eqn:lem:optimal volatility:pf:foc} and the fact that $w^\psi_y(t,s,y) = U'(y)$ for $\psi = 0$, that
\begin{align*}
&f''(t,s,y;\bar\sigma(t,s)) \frac{\partial\varsigma^\psi_*}{\partial\psi}(0)
= \lim_{\psi \downarrow 0} \frac{1}{\psi}\left(f'(t,s,y;\varsigma^\psi_*) - f'(t,s,y;\varsigma^0_*)\right)
= \lim_{\psi \downarrow 0} \frac{1}{\psi}f'(t,s,y;\varsigma^\psi_*)\\
&\;=\lim_{\psi\downarrow 0} \frac{1}{U'(y)}\left(s^2 \bar V_{ss}(t,s)\varsigma^\psi_* w^\psi_y(t,s,y) - \varsigma^\psi_* s^2\big(w^\psi_{ss}(t,s,y) + 2 \wt\vartheta \psi w^\psi_{sy}(t,s,y) + (\wt\vartheta \psi)^2 w^\psi_{yy}(t,s,y)\big)\right)\\
&\;=s^2 \bar V_{ss}(t,s)\bar\sigma(t,s).
\end{align*}
Solving for $\frac{\partial\varsigma^\psi_*}{\partial\psi}(0)$ gives $\frac{\partial\varsigma^\psi_*}{\partial\psi}(0) = \wt\sigma(t,s,y)$. Now, a Taylor expansion of $\varsigma^\psi_*$ around $\psi = 0$ yields
\begin{align*}
\varsigma^\psi_*
= \bar\sigma(t,s) + \wt\sigma(t,s,y)\psi + \frac{1}{2}\frac{\partial^2\varsigma^\psi_*}{\partial\psi^2}(\psi_L)\psi^2
= \sigma^\psi(t,s,y) + \frac{1}{2}\frac{\partial^2\varsigma^\psi_*}{\partial\psi^2}(\psi_L)\psi^2
\end{align*}
for some $\psi_L = \psi_L(t,s,y,\wt\vartheta;\psi)$ between $0$ and $\psi \in (0,\psi_1)$. Comparing this with \eqref{eqn:lem:optimal volatility:estimate}, it remains to show that $\frac{\partial^2\varsigma^\psi_*}{\partial\psi^2}$ can be uniformly bounded in $\lambda = (t,s,y,\wt\vartheta)\in\bfD\times[-L,L]$ and $\psi\in(0,\psi_1)$. We already know from \eqref{eqn:lem:optimal volatility:pf:positive sigma derivative} that $\frac{\partial F_\lambda}{\partial\varsigma}$ is bounded away from zero, uniformly over ${\lambda \in \bfD \times [-L,L]}$, $\psi \in (-\psi_1,\psi_1)$, and $\varsigma \in [0,K]$. In addition, it is straightforward to check that our boundedness assumptions imply that all the second order partial derivatives of $F_\lambda$ are uniformly bounded (in the same sense as above). Therefore, $\frac{\partial^2\varsigma^\psi_*}{\partial\psi^2}$ is uniformly bounded by Lemma~\ref{lem:implicit function bounds} (with $M_1 = 0$). This completes the proof.
\end{proof}

Conversely, we next determine the maximiser of the Hamiltonian with respect to the strategy variable, show that it coincides at the leading order with the candidate $\theta^\psi$ from \eqref{eqn:candidate functions} and that it is independent of the volatility variable.

\begin{lemma}
\label{lem:optimal strategy}
There are constants $C_1>0$ and $\psi_1>0$ such that for every $(t,s,y) \in \bfD$, $\varsigma \in [0,K]$, and $\psi \in (0,\psi_1)$, the function
\begin{align}
\label{eqn:lem:optimal strategy:function}
\RR\ni\vartheta \mapsto H^\psi(t,s,y;\vartheta,\varsigma)
\end{align}
has a maximiser $\vartheta^\psi_*(t,s,y)$ independent of $\varsigma$ that satisfies the first-order condition
\begin{align}
\label{eqn:lem:optimal strategy:foc}
\frac{\partial H^\psi}{\partial \vartheta}(t,s,y;\vartheta^\psi_*(t,s,y),\varsigma)
&= 0,
\end{align}
and
\begin{align}
\label{eqn:lem:optimal strategy:estimate}
\left\vert \vartheta^\psi_*(t,s,y) - \theta^\psi(t,s,y) \right\vert
&\leq C_1 \psi^2.
\end{align}
\end{lemma}

\begin{proof}
By the definition of $H^\psi$ in \eqref{eqn:hamiltonian}, finding the maximiser of \eqref{eqn:lem:optimal strategy:function} is equivalent to finding the maximiser of
\begin{align}
\label{eqn:lem:optimal strategy:pf:function}
\RR\ni\vartheta \mapsto 2 (\vartheta - \bar V_s(t,s)) w^\psi_{sy}(t,s,y) + (\vartheta - \bar V_s(t,s))^2 w^\psi_{yy}(t,s,y).
\end{align}
This is simply a quadratic equation in $\vartheta$ and independent of $\varsigma$. First, we show that the coefficient of the quadratic term is uniformly negative for small $\psi$. Note from the definition of $w^\psi$ in \eqref{eqn:candidate value function} that
\begin{align}
\label{eqn:lem:optimal strategy:pf:quadratic coefficient}
w^\psi_{yy}(t,s,y)
&= U''(y) + \frac{\partial^2}{\partial y^2}\Big(-U'(y)\wt w(t,s,y) + U'(y) \wh w(t,s,y) \psi\Big) \psi.
\end{align}
Since $U$ is $C^3$ and $U''<0$ there is $\varepsilon > 0$ such that $U'' \leq -2\varepsilon$ on $[y_l, y_u]$. By \eqref{eqn:ass:pde derivatives}, the partial derivative on the right-hand side of \eqref{eqn:lem:optimal strategy:pf:quadratic coefficient} can be bounded uniformly in $(t,s,y)\in\bfD$ and ${\psi \in (-1,1)}$. Hence, there is $\psi_1\in(0,1)$ such that for all $(t,s,y)\in \bfD$ and ${\psi \in (-\psi_1,\psi_1)}$,
\begin{align}
\label{eqn:lem:optimal strategy:pf:negative quadratic coefficient}
w^\psi_{yy}(t,s,y)
&\leq - \varepsilon.
\end{align}
Therefore for each $\psi \in (-\psi_1,\psi_1)$, \eqref{eqn:lem:optimal strategy:pf:function} has a maximiser $\vartheta^\psi_*(t,s,y)$ (which is also a maximiser of \eqref{eqn:lem:optimal strategy:function} if $\psi\neq 0$) that satisfies the first-order condition
\begin{align}
\label{eqn:lem:optimal strategy:pf:foc}
w^\psi_{sy}(t,s,y) + (\vartheta^\psi_*(t,s,y) - \bar V_s(t,s)) w^\psi_{yy}(t,s,y)
&= 0,
\end{align}
which is equivalent to \eqref{eqn:lem:optimal strategy:foc}.

To prove \eqref{eqn:lem:optimal strategy:estimate}, we argue similarly as in the proof of Lemma~\ref{lem:optimal volatility} using the implicit function theorem. Define for each $\lambda = (t,s,y)\in\bfD$, the function $F_\lambda:(-\psi_1,\psi_1)\times \RR \to \RR$ by
\begin{align*}
F_\lambda(\psi,\delta)
&= w^\psi_{sy}(t,s,y) + \delta w^\psi_{yy}(t,s,y).
\end{align*}
By construction, $F_\lambda$ is a polynomial in $(\psi, \delta)$ and hence $C^\infty$. Fix $\lambda = (t,s,y) \in \bfD$. By the first-order condition \eqref{eqn:lem:optimal strategy:pf:foc}, we have
\begin{align}
\label{eqn:lem:optimal strategy:pf:implicit characterisation}
F_\lambda(\psi,\vartheta^\psi_* - \bar V_s(t,s))
&= 0, \quad \psi \in (-\psi_1,\psi_1),
\end{align}
where $\vartheta^\psi_* = \vartheta^\psi_*(t,s,y)$. Since $\frac{\partial F_\lambda}{\partial \delta}(\psi,\delta) = w^\psi_{yy}(t,s,y) \leq -\varepsilon < 0$ for all $\psi \in (-\psi_1,\psi_1)$ by \eqref{eqn:lem:optimal strategy:pf:negative quadratic coefficient}, the implicit function theorem yields that $\vartheta^\psi_*$ is $C^\infty$ in $\psi$. As $F_\lambda(0,\delta) = \delta U''(y)$, the uniqueness assertion of the implicit function theorem also gives $\vartheta^0_* = \bar V_s(t,s)$. To compute $\frac{\partial\vartheta^\psi_*}{\partial\psi}(0)$, we divide by $\psi > 0$ on both sides of \eqref{eqn:lem:optimal strategy:pf:implicit characterisation} and let $\psi \downarrow 0$. Using also the definition of  $w^\psi$ in \eqref{eqn:candidate value function}, we obtain
\begin{align*}
0
&= \lim_{\psi \downarrow 0}\frac{1}{\psi} (w^\psi_{sy}(t,s,y) + (\vartheta^\psi_* - \bar V_s(t,s))w^\psi_{yy})\\
&= -\frac{\partial^2}{\partial s \partial y} \left(U'(y) \wt w(t,s,y)\right) + \lim_{\psi\downarrow0} \frac{\vartheta^\psi_* - \vartheta^0_*}{\psi} U''(y)\\
&= - U'(y) \wt w_{sy}(t,s,y) - U''(y) \wt w_s(t,s,y) + \frac{\partial\vartheta^\psi_*}{\partial\psi}(0) U''(y).
\end{align*}
Solving for $\frac{\partial\vartheta^\psi_*}{\partial\psi}(0)$ gives $\frac{\partial\vartheta^\psi_*}{\partial\psi}(0) = \wt\theta(t,s,y)$. Now, expanding $\vartheta^\psi_*$ around $\psi = 0$ yields
\begin{align*}
\vartheta^\psi_*
&= \bar V_s(t,s) + \wt\theta(t,s,y)\psi + \frac{1}{2} \frac{\partial^2 \vartheta^\psi_*}{\partial\psi^2}(\psi_L)\psi^2
= \theta^\psi(t,s,y) + \frac{1}{2} \frac{\partial^2 \vartheta^\psi_*}{\partial\psi^2}(\psi_L)\psi^2
\end{align*}
for some $\psi_L=\psi_L(t,s,y;\psi)$ between $0$ and $\psi \in (0,\psi_1)$. Comparing this with \eqref{eqn:lem:optimal strategy:estimate}, it remains to show that $\frac{\partial^2 \vartheta^\psi_*}{\partial\psi^2}$ can be uniformly bounded in $(t,s,y) \in \bfD$ and $\psi \in (0,\psi_1)$. We already know that $\frac{\partial F_\lambda}{\partial \delta} = w^\psi_{yy}(t,s,y)$ is bounded away from zero, uniformly over $\lambda \in \bfD$, $\psi \in (-\psi_1,\psi_1)$, and $\delta\in\RR$. In addition, using our boundedness assumptions, it is straightforward to check that there is $M > 0$ such that for all $\lambda \in \bfD$, $\psi\in(-\psi_1,\psi_1)$, and $\delta\in\RR$,
\begin{align*}
\left\vert\frac{\partial^2 F_\lambda}{\partial \psi}(\psi,\delta)\right\vert,
\left\vert\frac{\partial^2 F_\lambda}{\partial \psi^2}(\psi,\delta)\right\vert
&\leq M(1+\delta),
\quad
\left\vert\frac{\partial^2 F_\lambda}{\partial \psi\partial\delta}(\psi,\delta)\right\vert
\leq M,
\end{align*}
and that $\frac{\partial^2 F_\lambda}{\partial\delta^2} \equiv 0$. Then it follows from Lemma~\ref{lem:implicit function bounds} (with $y_\lambda(\psi) = \vartheta^\psi_*(t,s,y) - \bar V_s(t,s)$) that there is $M'>0$ such that for all $\lambda \in \bfD$ and $\psi \in (-\psi_1,\psi_1)$,
\begin{align*}
\left\vert\frac{\partial^2 \vartheta^\psi_*}{\partial\psi^2}(\psi)\right\vert
&= \left\vert\frac{\partial^2(\vartheta^\psi_* - \bar V_s(t,s))}{\partial\psi^2}(\psi)\right\vert
\leq M'(1+\vert\vartheta^\psi_* - \bar V_s(t,s)\vert).
\end{align*}
But $\vert\vartheta^\psi_* - \bar V_s(t,s)\vert = \left\vert\frac{w^\psi_{sy}(t,s,y)}{w^\psi_{yy}(t,s,y)}\right\vert \leq \frac{K'}{\varepsilon}$ by \eqref{eqn:lem:optimal strategy:pf:foc}, \eqref{eqn:lem:optimal strategy:pf:negative quadratic coefficient}, and \eqref{eqn:candidate value function:bounds}. This completes the proof.
\end{proof}

We now provide an asymptotic expansion of the HJBI equation at both the delta hedge and the candidate strategy $\theta^\psi$, both with respect to the candidate volatility $\sigma^\psi$.

\begin{lemma}
\label{lem:horrible calculation}
As $\psi \downarrow 0$, uniformly in $(t,s,y)\in\bfD$,
\begin{align}
\label{eqn:lem:horrible calculation:first order}
w^\psi_t(t,s,y) + H^\psi(t,s,y;\bar V_s(t,s),\sigma^\psi(t,s,y))
&= O(\psi^2),\\
\label{eqn:lem:horrible calculation:second order}
w^\psi_t(t,s,y) + H^\psi(t,s,y;\theta^\psi(t,s,y),\sigma^\psi(t,s,y))
&= O(\psi^3).
\end{align}
\end{lemma}

\begin{proof}
As $f$ is $C^4$ in $\varsigma$, Taylor's theorem together with \eqref{eqn:penalty function:minimum conditions} yields
\begin{align}
f(t,s,y;\varsigma)
&= \frac{1}{2}f''(t,s,y;\bar\sigma(t,s))(\varsigma-\bar\sigma(t,s))^2\notag\\
\label{eqn:lem:horrible calculation:pf:f expansion}
&\qquad+ \frac{1}{6} f^{(3)}(t,s,y;\bar\sigma(t,s))(\varsigma-\bar\sigma(t,s))^3
+ \frac{1}{24}f^{(4)}(t,s,y;\varsigma_L)(\varsigma-\bar\sigma(t,s))^4
\end{align}
for some $\varsigma_L = \varsigma_L(t,s,y;\varsigma)$ between $\bar\sigma(t,s)$ and $\varsigma\in[0,K]$. Recall that $f^{(4)}$ is uniformly bounded by \eqref{eqn:ass:penalty function}. Hence, using \eqref{eqn:lem:horrible calculation:pf:f expansion} for the candidate $\sigma^\psi(t,s,y)$, we obtain as $\psi\downarrow0$, uniformly in $(t,s,y)\in\bfD$,
\begin{align*}
&\frac{1}{\psi}f(t,s,y;\sigma^\psi(t,s,y))\\
&\; = \frac{1}{2}f''(t,s,y;\bar\sigma(t,s))\wt\sigma(t,s,y)^2 \psi
+ \frac{1}{6} f^{(3)}(t,s,y;\bar\sigma(t,s))\wt\sigma(t,s,y)^3 \psi^2 + O(\psi^3).
\end{align*}
Using this, it is easily seen from the definitions of $w^\psi$ and $H^\psi$ that the left-hand sides of \eqref{eqn:lem:horrible calculation:first order}--\eqref{eqn:lem:horrible calculation:second order} reduce to polynomials in $\psi$ up to terms of order $O(\psi^3)$. Using our boundedness assumptions, it is also straightforward to check that all the coefficients of these polynomials are uniformly bounded in $(t,s,y) \in \bfD$. Therefore, it suffices to check that the coefficients of the $O(1)$, $O(\psi)$ and, in the case of \eqref{eqn:lem:horrible calculation:second order}, $O(\psi^2)$ terms vanish. One readily verifies that the $O(1)$ term always vanishes and that the $O(\psi)$ term reduces in both cases to the PDE \eqref{eqn:PDE:first order} for $\wt w$. Finally, in the case of \eqref{eqn:lem:horrible calculation:second order}, a lengthy calculation shows that the $O(\psi^2)$ term reduces to the PDE \eqref{eqn:PDE:second order} for $\wh w$.\footnote{A Mathematica file containing these calculations is available from the authors upon request.}
\end{proof}

We next analyse the relevant minimised and maximised Hamiltonians if we plug in the leading-order candidate strategies and candidate volatilities from \eqref{eqn:candidate functions}, respectively.

\begin{lemma}
\label{lem:HJBI}
There are constants $C > 0$ and $\psi_0 > 0$ such that, for every $(t,s,y)\in \bfD$ and $\psi \in (0,\psi_0)$:
\begin{align}
\label{eqn:lem:HJBI:optimal strategy:first order}
w^\psi_t(t,s,y) + \inf_{\varsigma \in [0,K]} H^\psi(t,s,y;\bar V_s(t,s),\varsigma)
&\geq -C \psi^2,\\
\label{eqn:lem:HJBI:optimal strategy:second order}
w^\psi_t(t,s,y) + \inf_{\varsigma \in [0,K]} H^\psi(t,s,y;\theta^\psi(t,s,y),\varsigma)
&\geq -C \psi^3,\\
\label{eqn:lem:HJBI:optimal volatility}
w^\psi_t(t,s,y) + \sup_{\vartheta\in\RR} H^\psi(t,s,y;\vartheta,\sigma^\psi(t,s,y))
&\leq C \psi^3.
\end{align}
\end{lemma}

\begin{proof}
We first derive uniform bounds on the second-order partial derivatives of $H^\psi$ with respect to $\vartheta$ and $\varsigma$. As in the proof of Lemma~\ref{lem:optimal strategy} (cf.~\eqref{eqn:lem:optimal strategy:pf:negative quadratic coefficient}), there is $\psi_0\in(0,1)$ such that for all $(t,s,y)\in \bfD$ and $\psi \in (0,\psi_0)$, we have $w^\psi_{yy}(t,s,y) \leq -\varepsilon$. Together with \eqref{eqn:ass:reference price}--\eqref{eqn:ass:penalty function} and \eqref{eqn:candidate value function:bounds}, this implies that there is $K_1>0$ such that for all $(t,s,y)\in\bfD$, $\psi \in (0,\psi_0)$, $\wt\vartheta \in \RR$, and $\varsigma \in [0,K]$,
\begin{align}
\frac{\partial^2 H^\psi}{\partial \varsigma^2}(t,s,y;\bar V_s(t,s)+\wt\vartheta\psi,\varsigma)
&= \frac{1}{\psi}U'(y) f''(t,s,y;\varsigma) - s^2 \bar V_{ss}(t,s)w^\psi_y(t,s,y) \psi\notag\\
\label{eqn:lem:HJBI:pf:2nd derivative sigma:upper bound}
&\qquad + s^2\left(w^\psi_{ss}(t,s,y) + 2\wt\vartheta\psi w^\psi_{sy} + \wt\vartheta^2 \psi^2 w^\psi_{yy}(t,s,y) \right)
\leq \frac{K_1}{\psi}.
\end{align}
In view of \eqref{eqn:candidate value function:bounds}, choosing $K_1$ larger if necessary, we also have for all $(t,s,y)\in\bfD$, $\psi \in (0,\psi_0)$, $\vartheta \in \RR$, and $\varsigma \in [0,K]$,
\begin{align}
\label{eqn:lem:HJBI:pf:2nd derivative theta:lower bound}
\frac{\partial^2 H^\psi}{\partial \vartheta^2}(t,s,y;\vartheta,\varsigma)
&=\varsigma^2 s^2 w^\psi_{yy}(t,s,y)
\geq -K_1.
\end{align}
Also note that by definition of $\wt\theta$ in \eqref{eqn:theta tilde} and \eqref{eqn:ass:pde derivatives}, there is $L > 0$ such that $\vert\wt\theta\vert \leq L$ on $\bfD$. Choosing $\psi_0$ smaller if necessary, we may also assume that the estimates \eqref{eqn:lem:optimal volatility:estimate} and \eqref{eqn:lem:optimal strategy:estimate} of Lemmas~\ref{lem:optimal volatility}--\ref{lem:optimal strategy} (with that value of $L$) hold for $\psi \in (0,\psi_0)$ and that, using Lemma~\ref{lem:horrible calculation}, there is $C_2 > 0$ such that for all $(t,s,y)\in\bfD$ and $\psi \in (0,\psi_0)$,
\begin{align}
\label{eqn:lem:HJBI:pf:first order}
\left\vert w^\psi_t(t,s,y) + H^\psi(t,s,y;\bar V_s(t,s),\sigma^\psi(t,s,y)) \right\vert
&\leq C_2 \psi^2,\\
\label{eqn:lem:HJBI:pf:second order}
\left\vert w^\psi_t(t,s,y) + H^\psi(t,s,y;\theta^\psi(t,s,y),\sigma^\psi(t,s,y)) \right\vert
&\leq C_2 \psi^3.
\end{align}

We start with the proof of \eqref{eqn:lem:HJBI:optimal strategy:first order}--\eqref{eqn:lem:HJBI:optimal strategy:second order}. Fix $(t,s,y)\in\bfD$, $\psi\in(0,\psi_0)$, and $\wt\vartheta \in [-L,L]$ and define the function $h:[0,K] \to \RR$ by $h(\varsigma) = H^\psi(t,s,y;\bar V_s(t,s) + \wt\vartheta\psi, \varsigma)$. Let $\varsigma^\psi_* = \varsigma^\psi_*(t,s,y, \wt\vartheta)$ be the minimiser of $h$ from Lemma~\ref{lem:optimal volatility}. Expanding $h(\sigma^\psi)$ around $\varsigma^\psi_*$ and using the first-order condition \eqref{eqn:lem:optimal volatility:foc} gives
\begin{align}
\label{eqn:lem:HJBI:pf:sigma expansion}
h(\sigma^\psi)
&= h(\varsigma^\psi_*) + \frac{1}{2} \frac{\partial^2 H^\psi}{\partial \varsigma^2}(\varsigma_L)(\sigma^\psi - \varsigma^\psi_*)^2
\end{align}
for some $\varsigma_L = \varsigma_L(t,s,y,\wt\vartheta;\psi)$ between $\sigma^\psi = \sigma^\psi(t,s,y)$ and $\varsigma^\psi_*$. Rearranging \eqref{eqn:lem:HJBI:pf:sigma expansion} and using \eqref{eqn:lem:HJBI:pf:2nd derivative sigma:upper bound} as well as \eqref{eqn:lem:optimal volatility:estimate} yields
$h(\varsigma^\psi_*) \geq h(\sigma^\psi) - \frac{1}{2}K_1C_1^2 \psi^3$. We conclude that for all $(t,s,y)\in\bfD$, $\wt\vartheta\in[-L,L]$, and $\psi\in(0,\psi_0)$,
\begin{align}
\label{eqn:lem:HJBI:pf:lower bound}
H^\psi(t,s,y;\bar V_s(t,s) + \wt\vartheta\psi, \varsigma^\psi_*(t,s,y,\wt\vartheta))
&\geq H^\psi(t,s,y;\bar V_s(t,s) + \wt\vartheta\psi, \sigma^\psi(t,s,y)) -\frac{1}{2}K_1C_1^2 \psi^3.
\end{align}
Combining \eqref{eqn:lem:HJBI:pf:lower bound} for the choice $\wt\vartheta = \wt\theta(t,s,y)$  with \eqref{eqn:lem:HJBI:pf:second order} yields
\begin{align*}
&w^\psi_t(t,s,y) + \inf_{\varsigma \in [0,K]} H^\psi(t,s,y;\theta^\psi(t,s,y),\varsigma)\\
&\;= w^\psi_t(t,s,y) + H^\psi(t,s,y;\theta^\psi(t,s,y),\varsigma^\psi_*(t,s,y,\wt\theta(t,s,y)))\\
&\;\geq w^\psi_t(t,s,y) + H^\psi(t,s,y;\bar V_s(t,s) + \wt\vartheta\psi, \sigma^\psi(t,s,y)) - \frac{1}{2}K_1C_1^2 \psi^3 \\
&\;\geq -\big(C_2 + \frac{1}{2}K_1C_1^2 \big)\psi^3.
\end{align*}
This proves \eqref{eqn:lem:HJBI:optimal strategy:second order}. \eqref{eqn:lem:HJBI:optimal strategy:first order} follows analogously from \eqref{eqn:lem:HJBI:pf:lower bound} with $\wt\vartheta = 0$ and \eqref{eqn:lem:HJBI:pf:first order}.

The proof of \eqref{eqn:lem:HJBI:optimal volatility} is an almost verbatim repetition of the previous paragraph with the roles of $\theta$ and $\sigma$ exchanged. We give it for completeness. Fix $(t,s,y)\in\bfD$ and $\psi\in(0,\psi_0)$ and define the function $h:\RR \to \RR$ by $h(\vartheta) = H^\psi(t,s,y;\vartheta, \sigma^\psi(t,s,y))$. Let $\vartheta^\psi_* = \vartheta^\psi_*(t,s,y)$ be the maximiser of $h$ from Lemma~\ref{lem:optimal strategy}. Expanding $h(\theta^\psi)$ around $\vartheta^\psi_*$ and using the first-order condition \eqref{eqn:lem:optimal strategy:foc} gives
\begin{align}
\label{eqn:lem:HJBI:pf:theta expansion}
h(\theta^\psi)
&= h(\vartheta^\psi_*) + \frac{1}{2} \frac{\partial^2 H^\psi}{\partial \vartheta^2}(\vartheta_L)(\theta^\psi - \vartheta^\psi_*)^2
\end{align}
for some $\vartheta_L = \vartheta_L(t,s,y;\psi)$ between $\theta^\psi = \theta^\psi(t,s,y)$ and $\vartheta^\psi_*$. Rearranging \eqref{eqn:lem:HJBI:pf:theta expansion} and using \eqref{eqn:lem:HJBI:pf:2nd derivative theta:lower bound} as well as \eqref{eqn:lem:optimal strategy:estimate} yields
$h(\vartheta^\psi_*) \leq h(\theta^\psi) + \frac{1}{2}K_1C_1^2 \psi^4$. We conclude that for all $(t,s,y)\in\bfD$ and $\psi\in(0,\psi_0)$,
\begin{align}
\label{eqn:lem:HJBI:pf:upper bound}
H^\psi(t,s,y;\vartheta^\psi_*(t,s,y),\sigma^\psi(t,s,y))
&\leq H^\psi(t,s,y;\theta^\psi(t,s,y), \sigma^\psi(t,s,y)) + \frac{1}{2}K_1C_1^2 \psi^4.
\end{align}
Combining \eqref{eqn:lem:HJBI:pf:upper bound}  with \eqref{eqn:lem:HJBI:pf:second order} proves \eqref{eqn:lem:HJBI:optimal volatility} via
\begin{align*}
&w^\psi_t(t,s,y) + \sup_{\vartheta\in\RR} H^\psi(t,s,y;\vartheta,\sigma^\psi(t,s,y))\\
&\;= w^\psi_t(t,s,y) + H^\psi(t,s,y;\vartheta^\psi_*(t,s,y),\sigma^\psi(t,s,y))\\
&\;\leq w^\psi_t(t,s,y) + H^\psi(t,s,y;\theta^\psi(t,s,y), \sigma^\psi(t,s,y)) + \frac{1}{2}K_1C_1^2 \psi^4 \\
&\;\leq \big(C_2 + \frac{1}{2}K_1C_1^2 \psi_0 \big)\psi^3.\qedhere
\end{align*}
\end{proof}

\subsubsection{The lower bound \texorpdfstring{\eqref{eqn:first inequality}}{(\ref{eqn:first inequality})}}
\label{sec:first inequality}

Before we can prove the lower bound \eqref{eqn:first inequality}, we need two more preliminary results.

\renewcommand{\labelenumi}{(\roman{enumi})}
\begin{proposition}
\label{prop:inequalities preliminaries}
Fix $(\theta,\sigma,P)\in\fS$.
\begin{enumerate}
\item Let $\rho := \inf \lbrace t \in [0,T] : S_t \not\in (K^{-1},K) \rbrace \wedge T$ be the first time that $S$ leaves $(K^{-1},K)$. Then $\rho$ is a stopping time and, for each $\psi > 0$:
\begin{align}
\label{eqn:prop:inequalities preliminaries:S leaves interval}
U(Y_T)
&= w^\psi(\rho, S_\rho, Y_\rho) \quad P\text{-a.s.}
\end{align}
\item The local martingale part of the canonical decomposition of $Y$ under $P$, $\int_0^\cdot (\theta_t - \bar\Delta_t) \dd S_t$, is a bounded $P$-martingale.
\end{enumerate}
\end{proposition}

\begin{proof}
(i): It is an easy exercise to show that $\rho$ is a stopping time for the (non-augmented, non-right-continuous) filtration $\FF$. This uses the fact that all paths of $S$ are continuous and $(K^{-1},K)$ is open; cf.~\cite[Problem 2.7]{KaratzasShreve1998}. To prove \eqref{eqn:prop:inequalities preliminaries:S leaves interval}, fix $(\theta,\sigma,P)\in\fS$ and $\psi > 0$. On $\lbrace \rho = T \rbrace$, the terminal conditions of $\wt w$ and $\wh w$ in \eqref{eqn:PDE:first order}--\eqref{eqn:PDE:second order} yield
\begin{align*}
U(Y_T)
&= w^\psi(T,S_T,Y_T)
= w^\psi(\rho,S_\rho,Y_\rho) \quad P\text{-a.s.}
\end{align*}
On $\lbrace \rho < T \rbrace$, we have $S_\rho \in \lbrace K^{-1}, K \rbrace$ and the boundary conditions of $\wt w$ and $\wh w$ yield $w^\psi(\rho, S_\rho, Y_\rho) = U(Y_\rho)$ $P$-a.s. It remains to show that $Y_\rho = Y_T$ $P$-a.s.

Recall from \eqref{eqn:setup:Y} the dynamics of $Y$ under $P$:
\begin{align*}
Y
&= y_0 + \int_0^\cdot (\theta_t - \bar V_s(t,S_t)) \dd S_t + \int_0^\cdot \frac{1}{2}S_t^2 \bar V_{ss}(t,S_t) (\bar\sigma(t,S_t)^2 - \sigma_t^2) \dd t.
\end{align*}
By definition of $\ol\fS$, $S$ is a local martingale under $P$. By \eqref{eqn:process bounds}, $S$ is even a $[K^{-1},K]$-valued martingale under $P$. Hence, $S$ has to stay constant after hitting $K^{-1}$ or $K$. In particular, the local martingale term in the dynamics of $Y$ is constant after time $\rho$. Moreover, being constant, $S$ has constant quadratic variation after time $\rho$, so that by \eqref{eqn:setup:covariation}, $\sigma_t S_t = 0$ $\diff t \times P$-a.e. on $\lbrace(t,\omega) : t \geq \rho(\omega) \rbrace$. Together with the assumption \eqref{eqn:ass:reference volatility} that $\bar\sigma(\cdot, K) \equiv \bar\sigma(\cdot,K^{-1}) \equiv 0$, this yields that the drift term of the dynamics of $Y$ stays constant after time $\rho$ as well. Hence, $Y_\rho = Y_T$ $P$-a.s.

(ii): By \eqref{eqn:process bounds}, $Y$ and $\sigma$ are bounded. Moreover, by \eqref{eqn:ass:reference price}, $S_t^2 \bar V_{ss}(t,S_t)$ is bounded $\diff t\times P$-a.e. on $\lbrace(t,\omega) : t < \rho(\omega) \rbrace$. Since we also know from the proof of part~(i) that the drift part of $Y$ stays constant after time $\rho$, we conclude that the drift part is bounded. It follows that the local martingale part is also bounded.
\end{proof}

\begin{lemma}
\label{lem:probability estimate}
There is a constant $C > 0$ such that, for every $\psi \in (0,\psi_c)$, $\sigma \in \cV$, and $P \in \fP(\theta^\psi, \sigma)$ satisfying
\begin{align}
\label{eqn:lem:probability estimate:condition}
J^\psi(\sigma,P)
&\leq \inf_{\sigma' \in \cV} \inf_{P'\in\fP(\theta^\psi, \sigma')} J^\psi(\sigma', P') + 1,
\end{align}
we have
\begin{align}
\label{eqn:lem:probability estimate}
\PR{\tau < T}
&\leq C \psi
\end{align}
for the stopping time $\tau$ defined in Theorem~\ref{thm:second order}.
\end{lemma}

\begin{proof}
As an auxiliary result, we first prove that there is a constant $C_1 > 0$ such that, for every $\psi \in (0,\psi_c)$, $\sigma \in \cV$, and $P \in \fP(\theta^\psi, \sigma)$ satisfying \eqref{eqn:lem:probability estimate:condition}, we have
\begin{align}
\label{eqn:lem:probability estimate:pf:auxiliary estimate}
\EX[P]{\int_0^\rho \left(\sigma_t - \bar\sigma(t,S_t)\right)^2 \dd t}
&\leq C_1 \psi,
\end{align}
where $\rho$ is the first time that $S$ leaves $(K^{-1},K)$.

To this end, first fix $(t,s,y)\in\bfD$. Expanding the function $[0,K] \ni \varsigma \mapsto f(t,s,y;\varsigma)$ around $\bar\sigma(t,s)$ and using the minimum conditions \eqref{eqn:penalty function:minimum conditions}, we obtain
\begin{align*}
f(t,s,y;\varsigma)
&= \frac{1}{2}f''(t,s,y;\varsigma^L(t,s,y;\varsigma)) (\varsigma - \bar\sigma(t,s))^2
\end{align*}
for some $\varsigma^L(t,s,y;\varsigma) \in [0,K]$. Then the assumption that $\frac{1}{K} \leq f'' \leq K$ on $\bfD\times[0,K]$ from \eqref{eqn:ass:penalty function} gives
\begin{align*}
\frac{2}{K} f(t,s,y;\varsigma)
&\leq (\varsigma - \bar\sigma(t,s))^2
\leq 2K f(t,s,y;\varsigma), \quad (t,s,y)\in\bfD, \varsigma\in[0,K].
\end{align*}
Using also that $U'(y)$ is positive, bounded, and bounded away from zero over $y \in (y_l,y_u)$, it follows that there is a constant $K'' > 0$ such that 
\begin{align}
\label{eqn:lem:probability estimate:pf:taylor estimate}
\frac{1}{K''} (\varsigma - \bar\sigma(t,s))^2
&\leq  U'(y) f(t,s,y;\varsigma)
\leq K'' (\varsigma - \bar\sigma(t,s))^2, \quad (t,s,y)\in\bfD, \varsigma\in[0,K].
\end{align}

Next, set $C_1 :=T (K'')^2 K^6 \psi_c + K''(U(y_u)+ 1 - U(y_l))$ and fix $\psi \in (0,\psi_c)$, $\sigma\in\cV$, and $P \in \fP(\theta^\psi,\sigma)$. Moreover, choose any $P^\psi \in \fP(\theta^\psi,\sigma^\psi)$.   On the one hand, using the left-hand inequality in \eqref{eqn:lem:probability estimate:pf:taylor estimate}, the definition of $J^\psi$ in \eqref{eqn:objective}, and \eqref{eqn:lem:probability estimate:condition} gives
\begin{align}
\label{eqn:lem:probability estimate:pf:10}
\begin{split}
\EX[P]{\int_0^\rho \left(\sigma_t - \bar\sigma(t,S_t)\right)^2 \dd t}
&\leq K''\psi \EX[P]{\frac{1}{\psi}\int_0^\rho U'(Y_t)f(t,S_t,Y_t;\sigma_t) \dd t}\\
&\leq K'' \psi (J^\psi(\sigma, P) - \EX[P]{U(Y_T)})\\
&\leq K'' \psi \left(\inf_{\sigma' \in \cV} \inf_{P'\in\fP(\theta^\psi, \sigma')} J^\psi(\sigma', P') + 1 - U(y_l)\right).
\end{split}
\end{align}
On the other hand, using first that $\sigma^\psi \in \cV$, and then the right-hand inequality in \eqref{eqn:lem:probability estimate:pf:taylor estimate}, we obtain
\begin{align}
\label{eqn:lem:probability estimate:pf:20}
\begin{split}
\inf_{\sigma' \in \cV} \inf_{P'\in\fP(\theta^\psi, \sigma')} J^\psi(\sigma', P')
&\leq J(\sigma^\psi, P^\psi)
= \EX[P^\psi]{\frac{1}{\psi}\int_0^T U'(Y_t) f(t,S_t,Y_t;\sigma^\psi_t) \dd t + U(Y_T)}\\
&\leq \frac{K''}{\psi} \EX[P^\psi]{\int_0^T \left(\sigma^\psi_t - \bar\sigma(t,S_t)\right)^2 \dd t} + U(y_u).
\end{split}
\end{align}
In view of \eqref{eqn:candidate functions}, \eqref{eqn:sigma tilde}, and \eqref{eqn:ass:reference volatility}--\eqref{eqn:source term:boundary}, we find that
\begin{align*}
\left\vert\sigma^\psi_t - \bar\sigma(t,S_t)\right\vert
&= \left\vert\wt\sigma(t,S_t,Y_t)\right\vert\psi
= \left\vert \frac{\bar\sigma(t,S_t) S_t^2 \bar V_{ss}(t,S_t)}{f''(t,S_t,Y_t;\bar\sigma(t,S_t))} \right\vert \psi
\leq K^3 \psi \quad \diff t\times P^\psi\text{-a.e.} 
\end{align*} 
Combining this with \eqref{eqn:lem:probability estimate:pf:10}--\eqref{eqn:lem:probability estimate:pf:20} yields \eqref{eqn:lem:probability estimate:pf:auxiliary estimate} via
\begin{align*}
\EX[P]{\int_0^\rho \left(\sigma_t - \bar\sigma(t,S_t)\right)^2 \dd t}
&\leq K''\psi\left(\frac{K''}{\psi}T K^6\psi^2 +U(y_u) + 1 - U(y_l)\right)\\
&\leq \left(T (K'')^2 K^6 \psi_c + K''(U(y_u)+ 1 - U(y_l)) \right)\psi
= C_1 \psi.
\end{align*}

To prove \eqref{eqn:lem:probability estimate}, first note that by definition of $\wt\theta$ in \eqref{eqn:theta tilde} and our assumptions on $\wt w$ and $U$, there is a constant $L>0$ (depending only on $K$, $y_l$, $y_u$, and $U$) such that $\vert\wt\theta\vert \leq L$ on $\bfD$. Moreover, by standard estimates for It\^o processes (cf.~\cite[Lemma V.11.5]{RogersWilliams2000}), there is a constant $C' > 0$ (depending only on $T$) such that
\begin{align*}
&\EX[P]{\sup_{0\leq t \leq \rho} \vert Y_t - y_0 \vert^2}\\
&\;\leq C' \EX[P]{\int_0^\rho\left(\left( \big(\theta^\psi_t - \bar V_s(t,S_t) \big)S_t \sigma_t\right)^2 + \left(\frac{1}{2}S_t^2 \bar V_{ss}(t,S_t)(\bar\sigma(t,S_t)^2 - \sigma_t^2)\right)^2\right) \dd t}.
\end{align*}
Now, set $C := C'TL^2K^4\psi_c + C'K^4 C_1$ and fix $\psi \in (0,\psi_c)$, $\sigma \in \cV$, and $P \in \fP(\theta^\psi, \sigma)$ satisfying \eqref{eqn:lem:probability estimate:condition}. Note that by \eqref{eqn:process bounds} and \eqref{eqn:ass:reference volatility}--\eqref{eqn:ass:reference price},
\begin{align*}
\left\vert \big(\theta^\psi_t - \bar V_s(t,S_t) \big)\sigma_t S_t \right\vert
&= \left\vert\wt\theta(t,S_t,Y_t)\1_{\lbrace t < \tau \rbrace}\psi S_t \sigma_t \right\vert
\leq L\psi K^2,\\
\left\vert\frac{1}{2}S_t^2 \bar V_{ss}(t,S_t)(\bar\sigma(t,S_t)^2 - \sigma_t^2)\right\vert
&\leq \frac{1}{2}K \left\vert\bar\sigma(t,S_t)+\sigma_t\right\vert\left\vert\bar\sigma(t,S_t) - \sigma_t\right\vert\\
&\leq  K^2 \left\vert\bar\sigma(t,S_t)-\sigma_t \right\vert\quad \text{$\diff t \times P$-a.e. on $\lbrace (t,\omega):t < \rho(\omega) \rbrace$}.
\end{align*}
Recall that $Y$ remains constant after time $\rho$. Therefore, Markov's inequality, the above estimates, and the auxiliary estimate \eqref{eqn:lem:probability estimate:pf:auxiliary estimate} yield the assertion:
\begin{align*}
\PR{\tau < T}
&\leq\PR{\tau \leq \rho}
\leq \PR{\sup_{0\leq t\leq \rho} \vert Y_t - y_0 \vert^2 \geq 1}
\leq \EX[P]{\sup_{0\leq t \leq \rho} \vert Y_t - y_0 \vert^2}\\
&\leq C'T\left(L\psi K^2\right)^2 + C'K^4 \EX[P]{\int_0^\rho \left(\bar\sigma(t,S_t)-\sigma_t\right)^2 \dd t}\\
&\leq \left(C'TL^2K^4\psi_c + C'K^4 C_1\right)\psi\\
&= C \psi.\qedhere
\end{align*}
\end{proof}

\begin{lemma}
\label{lem:first inequality}
As $\psi \downarrow 0$, inequality \eqref{eqn:first inequality} holds true:
\begin{align*}
\inf_{\sigma \in \cV} \inf_{P\in\fP(\theta^\psi, \sigma)} J^\psi(\sigma, P)
&\geq w^\psi_0 + o(\psi^2).
\end{align*}
\end{lemma}

\begin{proof}
Fix $\varepsilon > 0$. Choose $C > 0$ large enough and $\psi_0\in(0,\psi_c)$ small enough, so that we may use the assertions of Lemmas~\ref{lem:HJBI} and \ref{lem:probability estimate}. Moreover, choosing $\psi_0$ even smaller if necessary, we may assume that $TC(\psi_0 + C \psi_0) \leq \frac{1}{2}\varepsilon$ and $\frac{1}{2}\varepsilon\psi_0^2 \leq 1$. Fix $\psi \in (0,\psi_0)$. We need to show that 
\begin{align*}
\inf_{\sigma \in \cV} \inf_{P\in\fP(\theta^\psi, \sigma)} J^\psi(\sigma, P) - w^\psi_0
&\geq -\varepsilon\psi^2.
\end{align*}
Choose $\sigma \in \cV$ and $P \in \fP(\theta^\psi, \sigma)$ such that $J^\psi(\sigma,P) \leq \inf_{\sigma' \in \cV} \inf_{P'\in\fP(\theta^\psi, \sigma')} J^\psi(\sigma', P') +  \frac{1}{2}\varepsilon\psi^2$ (in particular, condition \eqref{eqn:lem:probability estimate:condition} of Lemma~\ref{lem:probability estimate} holds, so we may use \eqref{eqn:lem:probability estimate} later). Then
\begin{align}
\label{eqn:lem:first inequality:pf:10}
\inf_{\sigma' \in \cV} \inf_{P'\in\fP(\theta^\psi, \sigma')} J^\psi(\sigma', P') - w^\psi_0
&\geq J^\psi(\sigma,P) - w^\psi_0 - \frac{1}{2}\varepsilon\psi^2.
\end{align}
Let $\rho$ be the first time that $S$ leaves $(K^{-1},K)$. By Proposition~\ref{prop:inequalities preliminaries} (i), $U(Y_T) = w^\psi(\rho, S_\rho, Y_\rho)$ $P$-a.s. Using this together with the definition of $J^\psi$ in \eqref{eqn:objective}, It\^o's formula for the process $w^\psi(t,S_t,Y_t)$ (up to time $\rho$), and the formulas \eqref{eqn:setup:Y} and \eqref{eqn:setup:covariation} (with $\theta$ replaced by $\theta^\psi$) describing the dynamics of $S$ and $Y$ under $P$, we obtain
\begin{align}
\label{eqn:lem:first inequality:pf:20}
\begin{split}
J^\psi(\sigma, P) - w^\psi_0
&\geq\EX[P]{\int_0^\rho w^\psi_s(t,S_t,Y_t) \dd S_t + \int_0^\rho w^\psi_y(t,S_t,Y_t)(\theta^\psi_t - \bar\Delta_t) \dd S_t}\\
&\qquad+ \EX[P]{\int_0^\rho \left(w^\psi_t(t,S_t,Y_t) + H^\psi(t,S_t,Y_t;\theta^\psi_t,\sigma_t)\right) \dd t},
\end{split}
\end{align}
where $H^\psi$ is the Hamiltonian defined in \eqref{eqn:hamiltonian}. We claim that both stochastic integrals inside the first expectation on the right-hand side of \eqref{eqn:lem:first inequality:pf:20} are true martingales under $P$. First, recall that $S$ and $\int_0^\cdot (\theta_t^\psi - \bar\Delta_t)\dd S_t$ are bounded martingales under $P$; cf.~Proposition~\ref{prop:inequalities preliminaries} (ii). Since $w^\psi_s$ and $w^\psi_y$ are bounded by \eqref{eqn:candidate value function:bounds}, the claim follows. So the first expectation in \eqref{eqn:lem:first inequality:pf:20} vanishes and it remains to estimate the second term.

Splitting the $\diff t$-integral into two parts separated by $\tau\wedge\rho$, using that $\theta^\psi_t = \theta^\psi(t,S_t,Y_t)$ on $\lbrace t < \tau \rbrace$ and $\theta^\psi_t = \bar V_s(t,S_t)$ on $\lbrace t \geq \tau \rbrace$, and applying \eqref{eqn:lem:HJBI:optimal strategy:second order}, \eqref{eqn:lem:HJBI:optimal strategy:first order} and, in the penultimate inequality, \eqref{eqn:lem:probability estimate}, we obtain
\begin{align}
\label{eqn:lem:first inequality:pf:30}
\begin{split}
J^\psi(\sigma, P) - w^\psi_0
&\geq - \EX[P]{\int_0^{\rho\wedge\tau} C\psi^3 \diff t + \int_{\rho\wedge\tau}^\rho C\psi^2 \diff t}
\geq -TC\psi^2(\psi + \PR{\tau < T})\\
&\geq -TC\psi^2(\psi + C \psi)
\geq -\frac{1}{2}\varepsilon\psi^2.
\end{split}
\end{align}
Combining \eqref{eqn:lem:first inequality:pf:30} with \eqref{eqn:lem:first inequality:pf:10} completes the proof.
\end{proof}

\subsubsection{The upper bound \texorpdfstring{\eqref{eqn:second inequality}}{(\ref{eqn:second inequality})}}
\label{sec:second inequality}

\begin{lemma}
\label{lem:second inequality}
As $\psi \downarrow 0$, inequality \eqref{eqn:second inequality} holds true:
\begin{align*}
\sup_{\theta \in \cA} \inf_{P\in\fP(\theta,\sigma^\psi)} J^\psi(\sigma^\psi, P)
&\leq w^\psi_0 + o(\psi^2).
\end{align*}
\end{lemma}

The proof is analogous to the proof of Lemma~\ref{lem:first inequality}, but easier.

\begin{proof}
Fix $\varepsilon > 0$. Choose $C > 0$ large enough and $\psi_0\in(0,\psi_c)$ small enough, so that we may use the assertions of Lemma~\ref{lem:HJBI}. Moreover, choosing $\psi_0$ even smaller if necessary, we may assume that $TC\psi_0 \leq \frac{1}{2}\varepsilon$. Fix $\psi \in (0,\psi_0)$. We need to show that 
\begin{align*}
\sup_{\theta \in \cA} \inf_{P\in\fP(\theta,\sigma^\psi)} J^\psi(\sigma^\psi, P) - w^\psi_0
&\leq \varepsilon \psi^2.
\end{align*}
Choose $\theta \in \cA$ such that $\inf_{P\in\fP(\theta,\sigma^\psi)} J^\psi(\sigma^\psi,P) + \frac{1}{2}\varepsilon\psi^2
\geq \sup_{\theta' \in \cA} \inf_{P\in\fP(\theta',\sigma^\psi)} J^\psi(\sigma^\psi, P)$ and fix any $P \in \fP(\theta,\sigma^\psi)$. Then
\begin{align}
\label{eqn:lem:second inequality:pf:10}
\sup_{\theta' \in \cA} \inf_{P\in\fP(\theta',\sigma^\psi)} J^\psi(\sigma^\psi, P) - w^\psi_0
&\leq J^\psi(\sigma^\psi,P) - w^\psi_0 + \frac{1}{2}\varepsilon\psi^2.
\end{align}

Let $\rho$ be the first time that $S$ leaves $(K^{-1},K)$. Recall from the proof of Proposition~\ref{prop:inequalities preliminaries}~(i) that $S$ has to stay constant (at $K^{-1}$ or $K$) after time $\rho$ (as it is a $[K^{-1},K]$-valued $P$-martingale). Hence, $\sigma^\psi_t = \bar\sigma(t,S_t) = 0$ on $\lbrace t \geq \rho \rbrace$ by \eqref{eqn:candidate functions} and \eqref{eqn:source term:boundary}. Then $f(t,S_t,Y_t;\sigma^\psi_t) = 0$ on $\lbrace t \geq \rho \rbrace$ by \eqref{eqn:penalty function:minimum conditions} and thus
\begin{align*}
J^\psi(\sigma^\psi,P)
&= \EX[P]{\frac{1}{\psi}\int_0^\rho U'(Y_t)f(t,S_t,Y_t;\sigma^\psi_t) \dd t + U(Y_T)}.
\end{align*}
Using this and proceeding as in the proof of Lemma~\ref{lem:first inequality}, we obtain
\begin{align}
\label{eqn:lem:second inequality:pf:20}
\begin{split}
J^\psi(\sigma^\psi, P) - w^\psi_0
&=\EX[P]{\int_0^\rho w^\psi_s(t,S_t,Y_t) \dd S_t + \int_0^\rho w^\psi_y(t,S_t,Y_t)(\theta_t - \bar\Delta_t) \dd S_t}\\
&\qquad+ \EX[P]{\int_0^\rho \left(w^\psi_t(t,S_t,Y_t) + H^\psi(t,S_t,Y_t;\theta_t,\sigma^\psi_t)\right) \dd t}.
\end{split}
\end{align}
By the same argument as in the proof of Lemma~\ref{lem:first inequality}, the first expectation in \eqref{eqn:lem:second inequality:pf:20} vanishes and it remains to estimate the second term. Using \eqref{eqn:lem:HJBI:optimal volatility} from Lemma~\ref{lem:HJBI} yields
\begin{align}
\label{eqn:lem:second inequality:pf:30}
J^\psi(\sigma, P) - w^\psi_0
&\leq  \EX[P]{\int_0^\rho C\psi^3 \dd t}
\leq TC\psi^3
\leq \frac{1}{2}\varepsilon\psi^2.
\end{align}
Combining \eqref{eqn:lem:second inequality:pf:30} with \eqref{eqn:lem:second inequality:pf:10} completes the proof.
\end{proof}

\subsection{Feynman--Kac representation}
\label{sec:proof of Feynman-Kac}

\begin{proof}[Proof of Proposition~\ref{prop:Feynman-Kac}]
Fix $(t,s) \in [0,T]\times[K^{-1},K]$ and let ${\check{\bar\sigma}(u,\cdot):\RR\to \RR}$ be a continuous extension of $\left.\bar\sigma(u,\cdot)\right\vert_{(K^{-1},K)}$ to $\RR$ for each $u\in[0,T]$. By \eqref{eqn:ass:reference volatility}, $\check{\bar\sigma}$ can also be chosen bounded. Therefore, a standard result (see, e.g., \cite[Theorems 21.9 and 21.7]{Kallenberg2002}) yields the existence of a weak solution to the SDE
\begin{align}
\label{eqn:prop:Feynman-Kac:pf:SDE}
\diff \check S_u
&= \check S_u \check{\bar\sigma}(u,\check S_u) \dd \check W_u, \quad \check S_t = s.
\end{align}
That is, there is a filtered probability space $(\check\Omega,\check\cF,\check\FF,\check P)$ supporting a Brownian motion $\check W$ and a process $\check S = (\check S_u)_{u \in [t,T]}$ satisfying \eqref{eqn:prop:Feynman-Kac:pf:SDE}. Now, define the process $S = (S_u)_{u\in[t,T]}$ by $S_u = \check S_{\rho \wedge u}$ where $\rho := {\inf\lbrace u\in[t,T] : S_u \not\in(K^{-1},K) \rbrace \wedge T}$ is the first time that $S$ leaves $(K^{-1},K)$. Hence, $S$ evolves like $\check S$, but is stopped as soon as it hits $K^{-1}$ or $K$. Using that $\bar\sigma(u,\cdot) = 0$ on $\lbrace K^{-1},K \rbrace$ by \eqref{eqn:ass:reference volatility}, it is easy to show that $S$ has the reference dynamics, i.e.,
\begin{align*}
\diff S_u
&= S_u \bar\sigma(u,S_u) \dd \check W_u, \quad S_t = s.
\end{align*}

Next, fix $y \in (y_l,y_u)$. Applying It\^o's formula to $\wt w(u,S_u,y)$ (up to time $\rho$) and using the terminal and boundary conditions of $\wt w$ in \eqref{eqn:PDE:first order} gives
\begin{align*}
0
=\wt w(\rho,S_\rho,y)
&= \wt w(t,s,y)
+ \int_t^\rho \wt w_s(u,S_u,y) \dd S_u\\
&\qquad
+ \int_t^\rho \left\lbrace \wt w_t(u,S_u,y) + \frac{1}{2}\bar\sigma(u,S_u)^2 S_u^2 \wt w_{ss}(u,S_u,y)\right\rbrace \dd u.
\end{align*}
Using the PDE \eqref{eqn:PDE:first order} to substitute the drift term and rearranging terms, we obtain
\begin{align}
\label{eqn:prop:Feynman-Kac:pf:before expectation}
\wt w(t,s,y)
&= -\int_t^\rho \wt w_s(u,S_u,y) \dd S_u
+ \int_t^\rho \wt g(u,S_u,y) \dd u.
\end{align}
Using our definition that $\wt g(t,s,y) = 0$ for $s \in \lbrace K^{-1},K\rbrace$ (cf.~\eqref{eqn:source term:boundary}) and the fact that $S$ stays constant (at $K^{-1}$ or $K$) after time $\rho$, we may replace the upper limit of the last integral in \eqref{eqn:prop:Feynman-Kac:pf:before expectation} by $T$. Moreover, by the boundedness of $\bar\sigma$ and $\wt w_s$ (cf.~\eqref{eqn:ass:pde derivatives}), the stochastic integral is a martingale under $\check P$. Therefore, taking expectations in \eqref{eqn:prop:Feynman-Kac:pf:before expectation} yields the Feynman--Kac representation \eqref{eqn:Feynman-Kac} (where the integrand is understood as zero if $S_u \in \lbrace K^{-1}, K \rbrace$).
\end{proof}

\subsection{Existence of probability scenarios}
\label{sec:existence}

\begin{proof}[Proof of Theorem~\ref{thm:existence}]
To prove the first assertion, fix $(\theta,\sigma) \in \cA_0\times\cV_0$. The goal is to construct continuous, adapted processes $S'$ and $Y'$ on some probability space $(\Omega',\cF',\FF',P')$ satisfying \eqref{eqn:process bounds} as well as
\begin{align}
\label{eqn:thm:existence:Sprime:goal}
S'_t
&= s_0 + \int_0^t S'_u \sigma_u((S',Y')) \dd W'_u,\\
\label{eqn:thm:existence:Yprime:goal}
Y'_t
&= y_0 + \int_0^t (\theta_u((S',Y')) - \bar V_s(u,S'_u)) \dd S'_u + \int_0^t \frac{1}{2}(S'_u)^2 \bar V_{ss}(u,S'_u) \big(\bar\sigma(u,S'_u)^2 - \sigma_u((S',Y'))^2 \big) \dd u,
\end{align}
for a Brownian motion $W'$. Comparing this with \eqref{eqn:setup:Y}--\eqref{eqn:setup:covariation}, we see that the image measure $P:=P'\circ(S',Y')^{-1}$ on $(\Omega,\cF)$ then satisfies $(\theta,\sigma,P)\in\fS$.

We proceed as follows. Note that by \eqref{eqn:A0}, the diffusion coefficient of $Y'$ changes to $0$ at time $\tau$. Hence, we first apply a general existence result to obtain a weak solution up to time $\tau$, ignoring for the moment that $S'$ or $Y'$ might exceed the bounds \eqref{eqn:process bounds}. Then we employ Theorem~\ref{thm:SDE:existence} to extend the process after time $\tau$. Finally, we can stop the extended process at the first time that $S$ leaves $(K^{-1},K)$ to obtain a candidate for $(S',Y')$.

First, introduce the cut-off function $h:\RR \to [K^{-1},K]$, $h(s) := (s \wedge K)\vee K^{-1}$, and recall that $\theta$ and $\sigma$ are of the form \eqref{eqn:A0} and \eqref{eqn:V0}, respectively. Consider the SDE
\begin{align}
\label{eqn:thm:existence:pf:S1}
S^{(1)}_t
&= s_0 + \int_0^t h(S^{(1)}_u) \acute\sigma(u,S^{(1)}_u,Y^{(1)}_u) \dd W'_u,\\
\label{eqn:thm:existence:pf:Y1}
\begin{split}
Y^{(1)}_t
&= y_0 + \int_0^t \breve\theta_u((S^{(1)},Y^{(1)})) h(S^{(1)}_u) \acute\sigma(u,S^{(1)}_u,Y^{(1)}_u) \dd W'_u\\
&\qquad + \int_0^t \frac{1}{2} h(S^{(1)}_u)^2 \check{\bar V}_{ss}(u,h(S^{(1)}_u)) \left(\check{\bar\sigma}(u,h(S^{(1)}_u))^2 - \acute\sigma(u,S^{(1)}_u,Y^{(1)}_u)^2\right) \dd u,
\end{split}
\end{align}
where $\check{\bar V}_{ss}$ and $\check{\bar\sigma}$ are Lipschitz continuous extensions of $\bar V_{ss},\bar\sigma:(0,T)\times(K^{-1},K) \to \RR$ to the closure $[0,T]\times[K^{-1},K]$.
By our assumptions, the drift and diffusion coefficients of this SDE are uniformly bounded and continuous on $\Omega$ for each fixed $u\in(0,T)$. Therefore, the SDE has a weak solution (see, e.g., \cite[Theorems 21.9 and 21.7]{Kallenberg2002}). In other words, there exist a filtered probability space $(\Omega',\cF',\FF',P')$ carrying a $(P',\FF')$-Brownian motion $W'$ and continuous, $\FF'$-adapted processes $S^{(1)}$, $Y^{(1)}$ satisfying \eqref{eqn:thm:existence:pf:S1}--\eqref{eqn:thm:existence:pf:Y1}. Recall that $\tau$ is an $\FF$-stopping time on $(\Omega,\cF)$. Applying Theorem~\ref{thm:SDE:existence} to $X^{(1)}:=(S^{(1)}, Y^{(1)})$ and the $\FF$-stopping time $\tau$, there exists an $\FF'$-adapted process $X^{(2)} = (S^{(2)}, Y^{(2)})$ satisfying
\begin{align*}
S^{(2)}_t
&= s_0 + \int_0^t h(S^{(2)}_u) \acute\sigma(u,S^{(2)}_u,Y^{(2)}_u) \dd W'_u,\\
Y^{(2)}_t
&= y_0 + \int_0^t \breve\theta_u((S^{(2)},Y^{(2)})) h(S^{(2)}_u) \acute\sigma(u,S^{(2)}_u,Y^{(2)}_u) \1_{\lbrace u<\tau(X^{(1)})\rbrace} \dd W'_u\\
&\qquad + \int_0^t \frac{1}{2} h(S^{(2)}_u)^2 \check{\bar V}_{ss}(u,h(S^{(2)}_u)) \big(\check{\bar\sigma}(u,h(S^{(2)}_u))^2 - \acute\sigma(u,S^{(2)}_u,Y^{(2)}_u)^2 \big) \dd u;
\end{align*}
note that only the diffusion coefficient of the $Y$-component is set to $0$ after time $\tau(X^{(1)})$ and that the remaining drift and diffusion coefficients are uniformly bounded and Lipschitz continuous in all variables. Next, define $X'=(S',Y')$ by $S'_t := S^{(2)}_{t \wedge \rho(X^{(2)})}$ and $Y'_t := Y^{(2)}_{t \wedge \rho(X^{(2)})}$ where $\rho := \inf \lbrace t \in [0,T] : S_t \not\in (K^{-1},K) \rbrace \wedge T$ is the first time that $S$ leaves $(K^{-1},K)$; note that $\rho$ is an $\FF$-stopping time on $(\Omega,\cF)$. Then
\begin{align}
\label{eqn:thm:existence:pf:Sprime1}
S'_t
&= s_0 + \int_0^t h(S'_u) \acute\sigma(u,S'_u,Y'_u)\1_{\lbrace u<\rho(X^{(2)})\rbrace} \dd W'_u,\\
\label{eqn:thm:existence:pf:Yprime1}
\begin{split}
Y'_t
&= y_0 + \int_0^t \breve\theta_u((S',Y')) h(S'_u) \acute\sigma(u,S'_u,Y'_u) \1_{\lbrace u<\tau(X^{(1)})\rbrace}\1_{\lbrace u<\rho(X^{(2)})\rbrace} \dd W'_u\\
&\qquad + \int_0^t \frac{1}{2} h(S'_u)^2 \check{\bar V}_{ss}(u,h(S'_u)) \big(\check{\bar\sigma}(u,h(S'_u))^2 - \acute\sigma(u,S'_u,Y'_u)^2 \big) \1_{\lbrace u < \rho(X^{(2)})\rbrace} \dd u.
\end{split}
\end{align}
Now note that on $\lbrace u < \rho(X^{(2)})\rbrace$, $S'_u \in (K^{-1},K)$, so $h(S'_u) = S'_u$, $\check{\bar\sigma}(u,S'_u) = \bar\sigma(u,S'_u)$ and $\check{\bar V}_{ss}(u,S'_u) = \bar V_{ss}(u,S'_u)$. Moreover, for each $\omega' \in \Omega'$, as $X^{(1)}_u(\omega') = X^{(2)}_u(\omega')$ for $u \leq \tau(X^{(1)}(\omega'))$ by construction of $X^{(2)}$, Galmarino's test \cite[Theorem IV.100]{DellacherieMeyer1978} implies that $\tau(X^{(1)}(\omega')) = \tau(X^{(2)}(\omega'))$. Similarly, we find that $\rho(X^{(2)}) = \rho(X')$. Combining these observations with \eqref{eqn:thm:existence:pf:Sprime1}--\eqref{eqn:thm:existence:pf:Yprime1} yields
\begin{align}
\label{eqn:thm:existence:pf:Sprime2}
S'_t
&= s_0 + \int_0^t S'_u \acute\sigma(u,S'_u,Y'_u)\1_{\lbrace u<\rho(X')\rbrace} \dd W'_u,\\
\label{eqn:thm:existence:pf:Yprime2}
\begin{split}
Y'_t
&= y_0 + \int_0^t \breve\theta_u((S',Y')) S'_u \acute\sigma(u,S'_u,Y'_u) \1_{\lbrace u<(\tau \wedge \rho)(X^{(2)})\rbrace} \dd W'_u\\
&\qquad + \int_0^t \frac{1}{2} (S'_u)^2 \bar V_{ss}(u,S'_u) \big(\bar\sigma(u,S'_u)^2 - \acute\sigma(u,S'_u,Y'_u)^2 \big) \1_{\lbrace u < \rho(X')\rbrace} \dd u.
\end{split}
\end{align}
First, note that $\acute\sigma(u,S'_u,Y'_u)\1_{\lbrace u<\rho(X')\rbrace} = \sigma_u((S',Y'))$ by \eqref{eqn:V0} and the fact that $S'$ stays constant after having hit $K^{-1}$ or $K$. In particular, $S'$ satisfies \eqref{eqn:thm:existence:Sprime:goal} as desired. Next, let us analyse the drift term in \eqref{eqn:thm:existence:pf:Yprime2}. To this end, note that on $\lbrace u \geq \rho(X') \rbrace$, $S'_u = S'_{\rho(X')} \in \lbrace K^{-1},K \rbrace$, and so $\bar\sigma(u,S'_u) = 0$ by \eqref{eqn:ass:reference volatility}. Thus, the drift term can be rewritten as 
\begin{align*}
\int_0^t \frac{1}{2} (S'_u)^2 \bar V_{ss}(u,S'_u) \big(\bar\sigma(u,S'_u)^2 - \sigma_u((S',Y'))^2 \big) \dd u
\end{align*}
in accordance with \eqref{eqn:thm:existence:Yprime:goal}. Finally, let us turn to the diffusion term in \eqref{eqn:thm:existence:pf:Yprime2}. $\tau \wedge \rho$ is an $\FF$-stopping time on $(\Omega, \cF)$. Moreover, for each $\omega'\in\Omega'$, as $X^{(2)}_u(\omega') = X'_u(\omega')$ for $u \leq \rho(X^{(2)}(\omega'))$ by construction of $X'$, Galmarino's test implies that $(\tau \wedge \rho)(X^{(2)}(\omega')) = (\tau \wedge \rho)(X'(\omega'))$. Using this, the diffusion term in \eqref{eqn:thm:existence:pf:Yprime2} can be rewritten as
\begin{align*}
\int_0^t \breve\theta_u((S',Y'))\1_{\lbrace u < \tau(X')\rbrace} S'_u \acute\sigma(u,S'_u,Y'_u) \1_{\lbrace u<\rho(X')\rbrace} \dd W'_u
&= \int_0^t \left(\theta_u((S',Y')) - \bar V_s(u,S'_u)\right)\dd S'_u,
\end{align*}
where we use \eqref{eqn:A0} and \eqref{eqn:thm:existence:pf:Sprime2} in the last step. This shows that $Y'$ satisfies \eqref{eqn:thm:existence:Yprime:goal} as desired.

It remains to check that $(S', Y')$ does not leave $[K^{-1},K]\times(y_l, y_u)$. It is clear that $S'$ evolves in $[K^{-1},K]$ as the diffusion coefficient in \eqref{eqn:thm:existence:Sprime:goal} is set to zero as soon as $S'$ hits $K^{-1}$ or $K$. Concerning $Y'$, we see from the definition of $\tau$ that the diffusion coefficient in \eqref{eqn:thm:existence:Yprime:goal} is set to $0$ as soon as $\vert Y' - y_0 \vert = 1$. Moreover, using \eqref{eqn:ass:reference price} and the fact that all volatilities take values in $[0,K]$, the absolute value of the drift term in \eqref{eqn:thm:existence:Yprime:goal} is bounded by $\frac{1}{2}K^3T$. It follows that $Y'$ evolves in $[y_0 - 1 - \frac{1}{2}K^3T, y_0 + 1 + \frac{1}{2}K^3 T] \subset (y_l,y_u)$. This completes the proof of the first assertion in Theorem~\ref{thm:existence}.

For the second assertion, we have to show that $(\theta^\psi, \sigma^\psi) \in \cA_0\times\cV_0$ for $\psi > 0$ small enough. Recall from Theorem~\ref{thm:second order} that
\begin{align*}
\theta^\psi_t
&= \bar V_s(t,S_t) + \wt\theta(t,S_t,Y_t)\1_{\lbrace t < \tau \rbrace} \psi,\\
\sigma^\psi_t
&= \bar\sigma(t,S_t) + \wt\sigma(t,S_t,Y_t)\psi,
\end{align*}
with
\begin{align*}
\wt\theta(t,s,y)
&= \wt w_s(t,s,y) + \frac{U'(y)}{U''(y)}\wt w_{sy}(t,s,y),\\
\wt\sigma(t,s,y)
&= \frac{\bar\sigma(t,s) s^2 \bar V_{ss}(t,s)}{f''(t,s,y;\bar\sigma(t,s))}.
\end{align*}
Under our assumptions, for any $\psi > 0$, $\breve\theta_t := \wt\theta(t,S_t,Y_t)\psi$ is clearly progressively measurable as well as bounded and continuous as a function on $[0,T]\times\Omega$. Thus, $\theta^\psi \in \cA_0$.

Setting $\acute\sigma(t,s,y):= \bar\sigma(t,s)+\wt\sigma(t,s,y)\psi$, we can write $\sigma^\psi_t = \acute\sigma(t,S_t,Y_t)\1_{\lbrace S_t \in (K^{-1},K)\rbrace}$ since $\bar\sigma(t,\cdot) = 0$ on $\lbrace K^{-1},K\rbrace$ by \eqref{eqn:ass:reference volatility}. To see that $\sigma^\psi \in \cV_0$, we thus have to check that $\acute\sigma(t,s,y)$ is Lipschitz continuous on $(0,T)\times(K^{-1},K)\times (y_l,y_u)$ (then it can be extended to a Lipschitz continuous function on $[0,T]\times\RR^2$) and takes values in $[0,K]$. Lipschitz continuity of $\wt\sigma$ follows from our assumptions (this uses in particular that $\frac{1}{K} \leq f'' \leq K$ from \eqref{eqn:ass:penalty function}). Whence, $\acute\sigma$ is Lipschitz continuous as well. Moreover, $\wt\sigma$ is bounded. Combining this with \eqref{eqn:ass:reference volatility}, we conclude that for $\psi > 0$ small enough, $\sigma^\psi$ takes values in $[0,K]$.
\end{proof}

\appendix

\section{Calculus}

The following result is an extension of the implicit function theorem that allows the defining function to depend on a parameter. In particular, it provides parameter-independent bounds for the first and second derivatives of the implicitly defined functions.

\begin{lemma}
\label{lem:implicit function bounds}
Let $\Lambda \neq \emptyset$ be a set and $U, V$ open subsets of $\RR$. For each $\lambda \in \Lambda$, let $F_\lambda: U \times V \to \RR$ and $y_\lambda: U \to V$ be twice continuously differentiable functions with
\begin{align}
\label{eqn:lem:implicit function bounds:defining equation}
F_\lambda(x,y_\lambda(x))
&= 0, \quad x \in U.
\end{align}
If there are constants $M_0 > 1$ and $M_1 \geq 0$ such that for each $\lambda \in \Lambda$ and all $(x,y) \in U \times V$,
\begin{align}
\label{eqn:lem:implicit function bounds:input bounds}
\left\vert \frac{\partial F_\lambda}{\partial x}(x,y)\right\vert,
\left\vert \frac{\partial^2 F_\lambda}{\partial x^2}(x,y) \right\vert
&\leq M_0 + M_1 \vert y\vert,
\quad
\left\vert\frac{\partial^2 F_\lambda}{\partial x \partial y}(x,y) \right\vert,
\left\vert\frac{\partial^2 F_\lambda}{\partial y^2} \right\vert
\leq M_0,
\quad
\left\vert \frac{\partial F_\lambda}{\partial y}(x,y)\right\vert
\geq \frac{1}{M_0},
\end{align}
then there is a constant $\wt M > 0$ such that for all $\lambda \in \Lambda$ and $x\in U$,
\begin{align*}
\vert y_\lambda'(x) \vert
&\leq \wt M(1 + M_1 \vert y_\lambda(x) \vert)
\quad\text{and}\quad
\vert y_\lambda''(x) \vert
\leq \wt M (1+M_1\vert y_\lambda (x)\vert + M_1\vert y_\lambda(x)\vert^2).
\end{align*}
Moreover, if $\frac{\partial^2 F_\lambda}{\partial y^2} \equiv 0$, then for all $\lambda \in \Lambda$ and $x\in U$,
\begin{align*}
\vert y_\lambda''(x) \vert
&\leq \wt M (1+M_1\vert y_\lambda (x)\vert).
\end{align*}
\end{lemma}

\begin{proof}
Taking the derivative of \eqref{eqn:lem:implicit function bounds:defining equation} with respect to $x$ yields
\begin{align}
\label{eqn:lem:implicit function bounds:pf:first derivative}
\frac{\partial F_\lambda}{\partial x}(x,y_\lambda(x)) + \frac{\partial F_\lambda}{\partial y}(x,y_\lambda(x))y_\lambda'(x)
&= 0, \quad \lambda \in \Lambda, x \in U.
\end{align}
Solving this for $y_\lambda'(x)$ and using the bounds \eqref{eqn:lem:implicit function bounds:input bounds} then gives
\begin{align}
\label{eqn:lem:implicit function bounds:pf:first bound}
\vert y_\lambda'(x)\vert
&\leq M_0(M_0 + M_1 \vert y_\lambda(x) \vert), \quad \lambda \in \Lambda, x \in U.
\end{align}
Taking the derivative in \eqref{eqn:lem:implicit function bounds:pf:first derivative}, we obtain for all $x\in U$,
\begin{align*}
\frac{\partial^2 F_\lambda}{\partial x^2}(x,y_\lambda(x)) + 2\frac{\partial^2 F_\lambda}{\partial x \partial y}(x,y_\lambda(x)) y_\lambda'(x)
+\frac{\partial^2 F_\lambda}{\partial y^2}(x,y_\lambda(x)) (y_\lambda'(x))^2 + \frac{\partial F_\lambda}{\partial y}(x,y_\lambda(x))y_\lambda''(x)
&= 0.
\end{align*}
Again, solving for $y_\lambda''(x)$ and using the bounds \eqref{eqn:lem:implicit function bounds:input bounds} as well as \eqref{eqn:lem:implicit function bounds:pf:first bound} gives
\begin{align}
\label{eqn:lem:implicit function bounds:pf:second bound}
\vert y_\lambda''(x) \vert
&\leq \ol M (1+M_1\vert y_\lambda (x)\vert + M_1\vert y_\lambda(x)\vert^2), \quad \lambda \in \Lambda, x \in U,
\end{align}
for some sufficiently large constant $\ol M$. Finally, if $\frac{\partial^2 F_\lambda}{\partial y^2} \equiv 0$, it is easily seen that the quadratic term in \eqref{eqn:lem:implicit function bounds:pf:second bound} vanishes.
\end{proof}

\section{Stochastic differential equations}

Fix an abstract filtered probability space $(\Omega, \cF, \FF = (\cF_t)_{t\geq0}, P)$ carrying an $r$-dimensional Brownian motion $W = (W^1_t,\ldots, W^r_t)_{t\geq0}$. The goal of this section is to prove the existence of  solutions of a class of stochastic differential equations (SDEs) whose coefficients change at a stopping time. More precisely, we consider SDEs of the form
\begin{align}
\label{eqn:SDE:SDE}
X_t
&= \xi + \int_0^t \sigma(s,X)\dd W_s + \int_0^t b(s,X)\dd s, \quad t \geq 0,
\end{align}
where $\xi$ is an $\RR^d$-valued $\cF_0$-measurable random vector, $X = (X^1_t,\ldots,X^d_t)_{t\geq0}$ is a continuous semimartingale in $\RR^d$, and
\begin{align}
\label{eqn:SDE:diffusion coefficient}
\sigma(s,X)
&= \sigma^{(1)}(s,X)\1_{\lbrace s<\tau(X)\rbrace} + \sigma^{(2)}(s,X_s)\1_{\lbrace s\geq\tau(X)\rbrace},\\
b(s,X)
&= b^{(1)}(s,X)\1_{\lbrace s<\tau(X)\rbrace} + b^{(2)}(s,X_s)\1_{\lbrace s\geq\tau(X)\rbrace}.\notag
\end{align}
Here, $\tau$ is a stopping time for the filtration induced by the canonical process on $C(\RR_+;\RR^d)$, $\sigma^{(1)},b^{(1)}$ are functions on $\RR_+\times C(\RR_+;\RR^d)$ that are progressive for the same filtration, and $\sigma^{(2)},b^{(2)}$ are measurable functions on $\RR_+ \times \RR^d$; all codomains are understood to be of suitable dimension.

First of all, note that we cannot apply general existence results directly to the coefficients $\sigma$ and $b$ since the stopping time $\tau$ is typically not a continuous function on $C(\RR_+;\RR^d)$. However, existence for this type of SDE is of course expected provided that solutions exist for both sets of coefficients separately. The obvious idea is to solve the SDE for the first set of cofficients, stop the solution at $\tau$, and solve from there the SDE with the second set of coefficients. This can be made precise as follows:

\begin{theorem}
\label{thm:SDE:existence}
Suppose that the process $X^{(1)}$ on $(\Omega,\cF,\FF,P)$ satisfies
\begin{align}
\label{eqn:thm:SDE:existence:X1}
X^{(1)}_t
&= \xi + \int_0^t \sigma^{(1)}(s,X^{(1)}) \dd W_s + \int_0^t b^{(1)}(s,X^{(1)}) \dd s, \quad t \geq 0.
\end{align}
Moreover, assume that there is a constant $K > 0$ such that for all $t, t' \geq 0$ and $x, x' \in \RR^d$,
\begin{align}
\label{eqn:thm:SDE:existence:Lipschitz:second}
\vert \sigma^{(2)}(t,x) - \sigma^{(2)}(t',x')\vert + \vert b^{(2)}(t,x)-b^{(2)}(t',x')\vert
&\leq K \vert (t,x) - (t',x') \vert,\\
\label{eqn:thm:SDE:existence:linear growth:second}
\vert \sigma^{(2)}(t,x)\vert + \vert b^{(2)}(t,x) \vert
&\leq K (1+ \vert (t,x) \vert)
\end{align}
where $\vert \cdot \vert$ denotes the Euclidean norm in the suitable dimension. Then there is a continuous, $\FF$-adapted, $\RR^d$-valued process $X$ satisfying \eqref{eqn:SDE:SDE}.
\end{theorem}

\begin{proof}
The solution prior to time $\hat\tau := \tau(X^{(1)})$ is already given. To construct the part of the solution after time $\hat\tau$, consider the time-shifted filtration $\wh\FF = (\wh\cF_t)_{t \geq 0}$ defined by $\wh\cF_t = \cF_{\hat\tau + t}$, and the time-shifted $(\wh\FF,P)$-Brownian motion $\wh W = (\wh W_t)_{t \geq 0}$ defined by $\wh W_t := W_{\hat\tau + t} - W_{\hat\tau}$. By our assumptions \eqref{eqn:thm:SDE:existence:Lipschitz:second}--\eqref{eqn:thm:SDE:existence:linear growth:second}, the coefficients of the SDE
\begin{align*}
\wh Y_t
&= \begin{pmatrix}\wh X\\\wh A\end{pmatrix}_t
= \zeta
+ \int_0^t \begin{pmatrix}\sigma^{(2)}(\wh A_s,\wh X_s)\\0\end{pmatrix}\dd \wh W_s
+ \int_0^t \begin{pmatrix}b^{(2)}(\wh A_s,\wh X_s)\\1\end{pmatrix} \dd s.
\end{align*}
fulfil the standard Lipschitz and linear growth assumptions that guarantee the existence of a $P$-a.s.~unique strong solution for any $\wh\cF_0$-measurable random vector $\zeta$ in $\RR^{d+1}$. In particular, there exists an $\wh\FF$-progressive process $\wh Y = (\wh X, \wh A)$ for the initial condition $\zeta := (X^{(1)}_{\hat\tau},\hat\tau)$. Clearly, $\wh A_t = \hat\tau + t$ plays the role of the shifted time variable. A simple time change now yields that $X^{(2)}_t := \wh X_{t-\hat\tau} \1_{\lbrace t \geq \hat\tau \rbrace}$ is $\FF$-progressive and satisfies
\begin{align}
\label{eqn:thm:SDE:existence:pf:X2}
X^{(2)}_t
&= X^{(1)}_{\hat\tau}
+ \int_{\hat\tau}^t \sigma^{(2)}(s,X^{(2)}_s) \dd W_s
+ \int_{\hat\tau}^t b^{(2)}(s,X^{(2)}_s) \dd s \quad \text{on } \lbrace t \geq \hat\tau \rbrace.
\end{align}

Finally, we verify that the process $X_t := X^{(1)}_t \1_{\lbrace t < \hat\tau \rbrace} + X^{(2)}_t\1_{\lbrace t \geq \hat\tau \rbrace}$ is a solution to the original SDE \eqref{eqn:SDE:SDE}. As $X^{(2)}_{\hat\tau} = X^{(1)}_{\hat\tau}$,
\begin{align*}
X_t
&= X^{(1)}_{t \wedge \hat\tau} + X^{(2)}_{t \vee \hat\tau} - X^{(2)}_{\hat\tau}.
\end{align*}
Plugging in \eqref{eqn:thm:SDE:existence:X1} and \eqref{eqn:thm:SDE:existence:pf:X2} gives
\begin{align}
X_t
&= \xi + \int_0^{t \wedge \hat\tau}  \sigma^{(1)}(s,X^{(1)}) \dd W_s
+ \int_0^{t\wedge\hat\tau} b^{(1)}(s,X^{(1)}) \dd s\notag\\
&\qquad +\int_{\hat\tau}^{t \vee \hat\tau} \sigma^{(2)}(s,X^{(2)}_s) \dd W_s
+ \int_{\hat\tau}^{t\vee\hat\tau} b^{(2)}(s,X^{(2)}_s) \dd s\notag\\
&= \xi + \int_0^t \sigma^{(1)}(s,X^{(1)})\1_{\lbrace s < \hat\tau \rbrace} \dd W_s
+ \int_0^t b^{(1)}(s,X^{(1)})\1_{\lbrace s < \hat\tau \rbrace} \dd s\notag\\
&\qquad +\int_0^t \sigma^{(2)}(s,X^{(2)}_s)\1_{\lbrace s\geq\hat\tau \rbrace} \dd W_s
+ \int_0^t b^{(2)}(s,X^{(2)}_s)\1_{\lbrace s\geq\hat\tau \rbrace} \dd s.
\label{eqn:thm:SDE:existence:pf:verification}
\end{align}
As $X^{(1)}_s = X_s$ on $\lbrace s \leq \hat\tau \rbrace$, Galmarino's test implies that $\hat\tau = \tau(X^{(1)}) = \tau(X)$. Moreover, since $\sigma^{(1)}$ is progressive, $\sigma^{(1)}(s,X^{(1)})$ only depends on the path of $X^{(1)}$ up to time $s$. Using also the definition of $\sigma$ in \eqref{eqn:SDE:diffusion coefficient}, we obtain
\begin{align*}
\sigma^{(1)}(s,X^{(1)})\1_{\lbrace s < \hat\tau \rbrace} + \sigma^{(2)}(s,X^{(2)}_s)\1_{\lbrace s \geq \hat\tau \rbrace}
&= \sigma(s,X).
\end{align*}
Using this and the analogous statement for the drift coefficients, we see from \eqref{eqn:thm:SDE:existence:pf:verification} that $X$ is a solution to \eqref{eqn:SDE:SDE}.
\end{proof}

\small
\providecommand{\bysame}{\leavevmode\hbox to3em{\hrulefill}\thinspace}
\providecommand{\MR}{\relax\ifhmode\unskip\space\fi MR }
\providecommand{\MRhref}[2]{%
  \href{http://www.ams.org/mathscinet-getitem?mr=#1}{#2}
}
\providecommand{\href}[2]{#2}

\end{document}